\documentclass[11pt]{article}

\usepackage[margin=1in]{geometry}

\usepackage{amsmath}
\usepackage{amssymb}
\usepackage{amsthm}
\usepackage{mathtools}
\usepackage{bm}

\usepackage{booktabs}
\usepackage{graphicx}
\usepackage{subcaption}
\usepackage{algorithm}
\usepackage{algpseudocode}
\usepackage{enumitem}
\usepackage{xcolor}

\usepackage{hyperref}
\usepackage[nameinlink,capitalise]{cleveref}   

\usepackage{authblk}

\usepackage[authoryear,round]{natbib}
\newtheorem{theorem}{Theorem}
\newtheorem{lemma}{Lemma}
\newtheorem{proposition}{Proposition}

\crefname{condition}{Condition}{Conditions}
\Crefname{condition}{Condition}{Conditions}

\crefname{theorem}{Theorem}{Theorems}
\crefname{lemma}{Lemma}{Lemmas}
\crefname{proposition}{Proposition}{Propositions}
\crefname{corollary}{Corollary}{Corollaries}
\crefname{assumption}{Assumption}{Assumptions}

\Crefname{theorem}{Theorem}{Theorems}
\Crefname{lemma}{Lemma}{Lemmas}
\Crefname{proposition}{Proposition}{Propositions}
\Crefname{corollary}{Corollary}{Corollaries}
\Crefname{assumption}{Assumption}{Assumptions}

\crefformat{equation}{#2(#1)#3}
\crefrangeformat{equation}{#3(#1)#4 to #5(#2)#6}
\crefmultiformat{equation}{#2(#1)#3}{ and #2(#1)#3}{, #2(#1)#3}{, and #2(#1)#3}

\newcommand{\R}{\mathbb{R}}
\newcommand{\C}{\mathbb{C}}
\newcommand{\E}{\mathbb{E}}
\newcommand{\Pp}{\mathbb{P}}

\newcommand{\ind}{\bm{1}}

\newcommand{\trans}{\top}

\DeclareMathOperator{\col}{col}
\DeclareMathOperator{\rank}{rank}
\DeclareMathOperator{\diag}{diag}
\DeclareMathOperator{\tr}{tr}
\DeclareMathOperator{\St}{St}
\newcommand{\op}{\mathrm{op}}

\newcommand{\lmin}{\lambda^{+}_{\min}}

\newcommand{\opn}[1]{\left\lVert #1\right\rVert_{\mathrm{op}}}
\newcommand{\frn}[1]{\left\lVert #1\right\rVert_{\mathrm{F}}}
\newcommand{\Lqn}[2]{\left\lVert #1\right\rVert_{L^{#2}}}

\newcommand{\cU}{\mathcal{U}}
\newcommand{\cM}{\mathcal{M}}
\newcommand{\cE}{\mathcal{E}}
\newcommand{\cI}{\mathcal{I}}
\newcommand{\Hc}{\mathcal{H}}

\newcommand{\bze}{\bm{0}}

\newcommand{\gmin}{\gamma_{\min}}
\newcommand{\gmax}{\gamma_{\max}}
\newcommand{\cemb}{c_{\mathrm{emb}}}
\newcommand{\Ogp}{\mathrm{O}}

\definecolor{cm}{rgb}{0,0,0.78}

\definecolor{yy}{rgb}{0.78,0,0.78}

\definecolor{xs}{rgb}{0,0.78,0}

\title{Bias-Corrected Subspace Intersection: Minimax-Optimal Shared Subspace Estimation in Multi-View Data}

\author[1]{Xianwen Song}
\author[2]{Yuepeng Yang}
\author[1]{Cong Ma}

\affil[1]{Department of Statistics, University of Chicago}
\affil[2]{Department of Statistics and Data Science, the Wharton School, University of Pennsylvania}

\date{September 4, 2026}

\allowdisplaybreaks
\begin{document}

\maketitle
\begingroup
\renewcommand{\thefootnote}{}
\footnotetext{The upper- and lower-bound proofs have also been formally verified in Lean 4; see~\url{https://github.com/congma1028/bcsi-lean-verification}.}
\endgroup


\begin{abstract}
Estimating a low-dimensional subspace shared across noisy data matrices is a fundamental problem in multi-view matrix estimation.
We study this problem under the two-view JIVE model, where each data matrix contains shared and view-specific low-rank components.
We demonstrate that standard {plug-in} subspace intersection, including AJIVE, suffers from a second-order bias caused by direction-dependent leakage of the empirical singular vectors.
We propose \emph{bias-corrected subspace intersection} (BCSI), which removes this bias before estimating the shared subspace.
We establish finite-sample risk bounds for BCSI that accommodate unequal view dimensions, signal strengths, and view-specific ranks and require no condition-number assumptions on the signal matrices.
When the shared and view-specific ranks are comparable, these bounds match our minimax lower bounds up to universal constants.
The resulting minimax rate contains a new second-order term, arising from quadratic leakage perturbations relative to the shrinking spectral gap when the view-specific subspaces are nearly aligned. This term is absent from previous JIVE minimax lower bounds. 
Numerical experiments demonstrate the advantage of BCSI over AJIVE when the leakage bias is pronounced. 
Along the way, we establish a nonasymptotic concentration result for the bias-corrected leakage Gram matrix of a rectangular spiked matrix, which may be of independent interest.
\end{abstract}

\section{Introduction}
\label{sec:introduction}

Modern data sets often measure the same collection of samples through
multiple data modalities, including different genomic assays
\citep{lock2013joint}, imaging technologies
\citep{zhao2021multimodal}, and other heterogeneous sources~\citep{Shi2023}. A central goal
of multi-view data integration is to separate variation that is shared across
views from variation that is specific to an individual view
\citep{lock2013joint,Feng2018,gaynanova2019structural,prothero2024data,gui2025indiseek}.

In this paper, we study the statistical problem of estimating the shared subspace in the
two-view setting. For $k=1,2$, we observe
\begin{equation}
\label{eq:JIVE-intro}
  \bm Y_k=\bm X_k+\bm E_k\in\mathbb R^{n\times d_k},
\end{equation}
where $\bm X_k$ is a rank-$p_k$ signal matrix and the entries of $\bm E_k$
are independent $\mathcal N(0,\sigma^2)$ random variables. Let
$\cM_k\coloneqq\operatorname{col}(\bm X_k)$ denote the population signal
space of view $k$. Our goal is to estimate the rank-$r$ shared subspace
\begin{equation}
\label{eq:target-intro}
  \cU
  \coloneqq
  \cM_1\cap\cM_2
  =
  \operatorname{col}(\bm X_1)\cap\operatorname{col}(\bm X_2)
\end{equation}
from the observations~$(\bm Y_1,\bm Y_2)$. This is the two-view shared-subspace estimation problem 
underlying the joint and individual variation explained (JIVE) framework \citep{lock2013joint}.

\subsection{Bias of plug-in subspace intersection}
\label{sec:ajive-bias}

A natural plug-in strategy is to estimate the two signal spaces separately and then look for directions that are approximately shared by the two estimated spaces. Let $\widehat{\bm\Phi}_k\in\mathbb R^{n\times p_k}$ contain the leading $p_k$ left singular vectors of $\bm Y_k$, so that $\operatorname{col}(\widehat{\bm\Phi}_k)$ estimates $\cM_k$. The right singular vectors associated with the smallest singular values of  $[\widehat{\bm\Phi}_1,\widehat{\bm\Phi}_2]$ identify pairs of directions, one from each estimated signal space, that are nearly the same up to sign. These directions therefore approximate the intersection of the two estimated signal spaces. In the two-view setting, this procedure is equivalent to AJIVE
\citep{Feng2018,yang2025estimating,sergazinov2026spectral}; see
\cref{sec:AJIVE} for a proof. 

The problem with this plug-in estimator is that the empirical singular
vectors leak outside their population signal spaces. This leakage creates a
systematic second-order bias and can make the estimator suboptimal for shared
subspace estimation. To see this, for any subspace
$\mathcal S\subseteq\mathbb R^n$, let $\bm P_{\mathcal S}$ denote the
orthogonal projector onto $\mathcal S$, and decompose
\[
  \widehat{\bm\Phi}_k=\widehat{\bm S}_k+\widehat{\bm L}_k,
  \qquad \text{where}\qquad
  \widehat{\bm S}_k\coloneqq\bm P_{\cM_k}\widehat{\bm\Phi}_k,
  \quad
  \widehat{\bm L}_k\coloneqq\bm P_{\cM_k^\perp}\widehat{\bm\Phi}_k.
\] 
The approximate intersection is determined by the bottom eigenspace of the
Gram matrix of the concatenated empirical bases
$[\widehat{\bm\Phi}_1,\widehat{\bm\Phi}_2]$. With
$\widehat{\bm S}\coloneqq[\widehat{\bm S}_1,\widehat{\bm S}_2]$ and
$\widehat{\bm L}\coloneqq[\widehat{\bm L}_1,\widehat{\bm L}_2]$, this Gram
matrix decomposes as
\begin{equation}
\label{eq:intro-ajive-gram}
\begin{aligned}
  [\widehat{\bm\Phi}_1,\widehat{\bm\Phi}_2]^\trans
  [\widehat{\bm\Phi}_1,\widehat{\bm\Phi}_2]
  =
  \widehat{\bm S}^\trans\widehat{\bm S}
  +
  \widehat{\bm S}^\trans\widehat{\bm L}
  +
  \widehat{\bm L}^\trans\widehat{\bm S}
  +
  \widehat{\bm L}^\trans\widehat{\bm L}.
\end{aligned}
\end{equation}
If there were no leakage and
$\operatorname{col}(\widehat{\bm S}_k)=\cM_k$ for $k=1,2$, then
$\widehat{\bm S}^\trans\widehat{\bm S}$ would have an $r$-dimensional
kernel consisting exactly of the linear relations that produce directions
in the shared subspace $\cU$.

In the noisy case, the two cross terms in
\cref{eq:intro-ajive-gram} are linear in the leakage, while
$\widehat{\bm L}^\trans\widehat{\bm L}$ is quadratic. Crucially, this
quadratic term has a direction-dependent bias and can therefore perturb the
nullspace of $\widehat{\bm S}^\trans\widehat{\bm S}$. 
To see why the bias is
direction dependent, consider the diagonal entries of 
$\widehat{\bm L}_k^\trans\widehat{\bm L}_k$. Let
$(\widehat s_{k,i},\widehat{\bm\phi}_{k,i},
\widehat{\bm\psi}_{k,i})$
be the $i$th empirical singular triplet and set
$a_k\coloneqq n-p_k$ and $b_k\coloneqq d_k-p_k$. The singular-vector relation
$\bm Y_k\widehat{\bm\psi}_{k,i}
=
\widehat s_{k,i}\widehat{\bm\phi}_{k,i}$,
together with
$\bm P_{\cM_k^\perp}\bm X_k=\bm 0$,
implies
\begin{equation}
\label{eq:leakage-heuristic}
\begin{aligned}
  (\widehat{\bm L}_k^\trans\widehat{\bm L}_k)_{ii}
  &=
  \bigl\|
    \bm P_{\cM_k^\perp}\widehat{\bm\phi}_{k,i}
  \bigr\|^2
  =
  \frac{
    \bigl\|
      \bm P_{\cM_k^\perp}
      \bm E_k\widehat{\bm\psi}_{k,i}
    \bigr\|^2
  }{
    \widehat s_{k,i}^2
  }
  \approx
  \frac{a_k\sigma^2}{\widehat s_{k,i}^2}.
\end{aligned}
\end{equation}
The last relation is heuristic because
$\widehat{\bm\psi}_{k,i}$ depends on the noise.  It
nevertheless shows that singular
directions associated with smaller empirical singular values have larger
leakage. 
Thus the leakage Gram is generally anisotropic, and this anisotropy biases the estimated intersection.

\subsection{Bias correction and main results}
\label{sec:intro-leakage-correction}

The heuristic calculation in \cref{eq:leakage-heuristic} suggests a natural
bias correction: rescale the empirical singular directions within each view
so that their leakage Gram is approximately isotropic. Specifically, replace
$\widehat{\bm\Phi}_k$ by
$\widehat{\bm\Phi}_k\diag(w_{k,1},\ldots,w_{k,p_k})$. 
At the heuristic level, this would give weights roughly proportional to
$\widehat s_{k,i}$. However, the empirical singular vectors themselves depend on the noise, so
this heuristic does not give the correct weight. Rectangular spiked-matrix
theory gives a more accurate expression for the leakage.

Let $s_{k,i}$ denote the $i$th population singular value in view $k$.  
Classical rectangular spiked-matrix theory
\citep{benaych2012singular,gavish2017optimal} predicts
\begin{equation}
\label{eq:intro-spike-leakage}
  \bigl\|
    \bm P_{\cM_k^\perp}\widehat{\bm\phi}_{k,i}
  \bigr\|^2
  \approx
  \frac{
    a_k\sigma^2(s_{k,i}^2+b_k\sigma^2)
  }{
    s_{k,i}^2(s_{k,i}^2+a_k\sigma^2)
  }.
\end{equation}
If $s_{k,i}$ were known, multiplying $\widehat{\bm\phi}_{k,i}$ by
$\{s_{k,i}^2(s_{k,i}^2+a_k\sigma^2)/
(s_{k,i}^2+b_k\sigma^2)\}^{1/2}$ would make its squared leakage approximately
$a_k\sigma^2$ for every $i$.  

This oracle rescaling is not directly implementable because the population
singular value $s_{k,i}$ is unknown. The same spiked-matrix theory provides the
missing link between the population and empirical singular values:
\begin{equation}
\label{eq:intro-spike-location}
  \widehat s_{k,i}^2
  \approx
  \frac{
    (s_{k,i}^2+a_k\sigma^2)
    (s_{k,i}^2+b_k\sigma^2)
  }{
    s_{k,i}^2
  }.
\end{equation}
We therefore invert \cref{eq:intro-spike-location} to estimate $s_{k,i}^2$
from $\widehat s_{k,i}$ and substitute this estimate into the oracle
rescaling above.

This leads to our proposed estimator, \emph{bias-corrected subspace
intersection} (BCSI).
BCSI modifies plug-in subspace intersection in three
steps. First, it reweights the empirical singular directions within each view
so that their leakage Gram is approximately
$a_k\sigma^2\bm I_{p_k}$. Second, it subtracts this common leakage level from
the corresponding diagonal block of the Gram matrix. Finally, it extracts the near-zero eigenspace of the corrected Gram matrix
and uses it to reconstruct the shared subspace. We denote the corresponding projector by
$\widehat{\bm P}_{\rm BCSI}$ and give the complete construction in
\cref{sec:bcsi-estimator}.

Our main results show that BCSI achieves the minimax rate for shared-subspace
estimation over a broad range of signal strengths and subspace geometries.  For $k=1,2$, let $s_k$ be a known lower bound on the 
smallest nonzero singular value of $\bm X_k$ and define the corresponding effective signal strength by
\begin{equation}
\label{eq:effective-strength}
  \gamma_k
  \coloneqq
  \frac{s_k^2}{\sqrt{s_k^2+d_k\sigma^2}}.
\end{equation}
Also let $\gmax \coloneqq \max \{\gamma_1, \gamma_2\}$, and $\gmin \coloneqq \min \{\gamma_1, \gamma_2\}$.  
The effective strength $\gamma_k$ accounts for the column dimension $d_k$.
For column-space estimation, the raw singular value $s_k$ alone does not
capture the effect of the $d_k$ noisy columns, so signals with the same
$s_k$ can have different estimation accuracy when their aspect ratios differ~\citep{cai2018rate, Cai2021Unbalanced}.
In particular, $\gamma_k$ is of order $s_k$ for a strong signal and
of order $s_k^2/(\sigma\sqrt{d_k})$ for a weaker one.

Suppose that the shared and view-specific subspaces all have rank of order $r$.  Let
$\theta\in(0,1/2]$ measure the separation between two view-specific subspaces,
with smaller $\theta$ corresponding to more closely aligned subspaces. 
Under the signal strength assumption
$\gmin \gtrsim\sigma\sqrt n$, BCSI satisfies
\begin{equation}
\label{eq:intro-upper-rate}
  \mathbb E
  \opn{
    \widehat{\bm P}_{\rm BCSI}-\bm P_{\cU}
  }
  \lesssim
  1\wedge
  \left\{
    \frac{\sigma\sqrt n}{\gmax}
    +
    \frac{\sigma\sqrt r}
         {\gmin \sqrt\theta}
    +
    \frac{\sigma^2\sqrt{nr}}
         {\gmin ^2\theta}
  \right\}.
\end{equation}
A matching minimax lower bound shows that all three terms in
\cref{eq:intro-upper-rate} are unavoidable up to universal constants.

The first term is the usual cost of estimating the shared subspace when the view-specific directions are known. Since both views contain information about the shared directions, their effective strengths combine, yielding $\sqrt{\gamma_1^2+\gamma_2^2}\asymp\gmax$. The remaining two terms reflect the additional difficulty of identifying which directions are shared when the view-specific subspaces are unknown. This identification becomes harder as the two view-specific subspaces approach one another, leading to the factors $\theta^{-1/2}$ and $\theta^{-1}$. Moreover, these terms are governed by $\gmin$: even if the stronger view accurately localizes the relevant signal space, the weaker view is still needed to distinguish the shared directions from the view-specific ones. The final $\theta^{-1}$ term is one of our main contributions. It is second
order in the noise level and shows that nearly aligned view-specific subspaces
create an additional statistical difficulty beyond the usual first-order
subspace estimation error. In the balanced-rank regime, the corresponding
quadratic term in the existing AJIVE bound is of order
$\sigma^2 n/(\gmin^2\theta)$, whereas BCSI improves it to
$\sigma^2\sqrt{nr}/(\gmin^2\theta)$, a factor $\sqrt{n/r}$ smaller.
Our minimax lower bound further shows that the latter $\theta^{-1}$ term is
unavoidable, and hence is intrinsic to shared-subspace estimation rather than
an artifact of the analysis of BCSI. The general results, allowing unequal
shared and view-specific ranks, are stated in
\cref{sec:upper-bound,sec:lower-bound}.

\subsection{Related work}
\paragraph{Joint and individual variation.}

The JIVE framework \citep{lock2013joint, zhou2015group} models multi-view data
through shared and view-specific low-rank components,
and has various variants including AJIVE \citep{Feng2018}, CJIVE
\citep{murden2022interpretive}, and the probabilistic formulation ProJIVE
\citep{murden2026projive}.
Related multiblock decompositions such as SLIDE \citep{gaynanova2019structural}, hierarchical nuclear-norm
penalization \citep{yi2023hierarchical}, and DIVAS
\citep{prothero2024data} also accommodate partially shared structure.
Among these methods, AJIVE is most closely related to the present work: it
first estimates the signal space of each view and then identifies shared
directions through the approximate intersection of the estimated spaces.
Similarly, \citet{sergazinov2026spectral} use the product-of-projections method to combine the estimates across different views.

\paragraph{Theory for JIVE.}
The closest statistical analyses are
\citet{yang2025estimating} and \citet{li2025heterojive}.
The former establishes finite-sample guarantees and minimax lower bounds for
AJIVE in a homogeneous multi-view model, while the latter allows heterogeneous
dimensions and signal strengths and studies strength-aware aggregation of the
estimated view projectors.
Our results allow unequal view dimensions, strengths, and view-specific ranks,
impose no condition-number restriction on the loading matrices, and identify
an additional $\theta^{-1}$ term in the minimax rate arising from the
interaction between estimation error and the geometry of the intersection.

\paragraph{Shared-subspace estimation.}
Related work considers shared-subspace estimation under alternative
multi-matrix models.  \citet{ma2026optimal} analyze
complete and partial sharing when the shared and unshared directions are
themselves singular vectors of the individual signal matrices, while the
partially shared subspace model \citep{yoon2026joint} assumes that the shared
directions are eigenvectors of each population covariance.  MOP-UP
\citep{tang2025modewise} considers a decomposition in which one component
has a column subspace shared across observations and the other has a shared
row subspace.  In the fully shared setting,  \citet{baharav2025stacked} derive strength-dependent aggregation rules, while \citet{agterberg2026statistically} studies estimation and adaptive inference for shared subspaces across symmetric low-rank matrices.  
Unlike these constructions, our model allows the loading matrices to arbitrarily mix the shared and
view-specific coordinates. Consequently, the shared directions need not be singular
vectors of any individual view and are instead identified through the
intersection of the two population signal spaces.

\paragraph{Spiked random matrix theory.}
Our bias correction is motivated by classical results for spiked covariance
and rectangular low-rank models, including the phase transition and
eigenvector behavior studied in
\citet{johnstone2001distribution}, \cite{baik2005phase}, and \cite{paul2007asymptotics},
and the singular-value and singular-vector asymptotics for rectangular
perturbations developed by \citet{benaych2012singular}.
These results underlie later work on data-dependent spectral shrinkage and
optimal singular-value transformations
\citep{nadakuditi2014optshrink,gavish2017optimal}.
BCSI uses the corresponding spike-location and singular-vector formulas to
construct direction-dependent weights that equalize the leading leakage of
the empirical singular vectors.

\paragraph{Finite-sample subspace estimation and spectral debiasing.}
Our analysis is also related to finite-sample singular-subspace perturbation
theory for rectangular matrices
\citep{cai2018rate,Cai2021Unbalanced} and to spectral debiasing methods such as
HeteroPCA \citep{zhang2022heteropca}.
HeteroPCA removes a diagonal bias due to heteroskedastic noise.  The bias considered here arises instead
after estimating the signal space of each view: different empirical singular
directions leak outside their population signal spaces by different amounts.
BCSI reweights these directions to make the leakage approximately isotropic
and then removes the resulting common shift before estimating the
intersection.

\subsection{Notation}
\label{sec:notation}

All matrices and vectors are set in bold. For
$\bm A\in\mathbb R^{n\times d}$, $\|\bm A\|_{\mathrm{op}}$ is the spectral norm,
$\bm A^{\dagger}$ the Moore--Penrose pseudoinverse, $\sigma_j(\bm A)$ the $j$-th
largest singular value, and $\sigma^{+}_{\min}(\bm A)$ the smallest nonzero
singular value. For symmetric $\bm A$ we write $\lambda_{\min}(\bm A)$ and
$\lambda^{+}_{\min}(\bm A)$ for the smallest and the smallest positive eigenvalues, respectively. We set
$\operatorname{St}(n,k)\coloneqq\{\bm W\in\mathbb R^{n\times k}:
\bm W^{\top}\bm W=\bm I_k\}$, write $\operatorname{col}(\bm A)$ for the column
space of $\bm A$, and write $\bm P_{\mathcal V}$ for the orthogonal projector onto a subspace
$\mathcal V$. We write
$a\wedge b\coloneqq\min\{a,b\}$ and
$a\vee b\coloneqq\max\{a,b\}$. For nonnegative quantities $x$ and $y$,
the relation $x\lesssim y$ means that $x\le Cy$ for a universal constant
$C>0$, and $x\gtrsim y$ denotes the reverse inequality. We write
$x\asymp y$ when both $x\lesssim y$ and $x\gtrsim y$ hold. For a matrix
$\bm Z$, set
$\|\bm Z\|_{L^q}\coloneqq(\mathbb E\|\bm Z\|_{\mathrm{op}}^q)^{1/q}$.
For an event $\cE$, let $\ind_{\cE}$ denote its indicator.

\section{Model and bias-corrected subspace intersection}
\label{sec:estimator}

\subsection{Model and problem parameters}
\label{sec:formulation}
We now give a concrete parametrization of the two-view JIVE model introduced in
\cref{sec:introduction}. The parametrization separates the shared and
view-specific components and makes explicit the signal strengths and subspace
separation that govern the difficulty of estimating the shared subspace.

Recall that
$\cM_k=\operatorname{col}(\bm X_k)$ and
$\cU=\cM_1\cap\cM_2$.
For $k=1,2$, we observe
\begin{equation}
\label{eq:model}
  \bm Y_k=\bm X_k+\bm E_k,
  \qquad
  \bm X_k=[\bm U,\bm V_k]\bm B_k^\trans
  \in\mathbb R^{n\times d_k}.
\end{equation}
Here
$\bm U\in\St(n,r)$,
$\bm V_k\in\St(n,r_k)$, and
$[\bm U,\bm V_k]\in\St(n,p_k)$, where
$p_k\coloneqq r+r_k$.
The matrix
$\bm B_k\in\mathbb R^{d_k\times p_k}$
contains the right loadings.
The entries of $\bm E_1$ and $\bm E_2$ are mutually independent
$\mathcal N(0,\sigma^2)$ random variables.

We assume
$\sigma_{p_k}(\bm B_k)\ge s_k>0$.
Since $[\bm U,\bm V_k]$ has orthonormal columns, the nonzero singular
values of $\bm X_k$ and $\bm B_k$ coincide. In particular,
$\bm B_k$ has full column rank and
\[
  \cM_k
  =
  \operatorname{col}(\bm X_k)
  =
  \operatorname{col}(\bm U)
  \oplus
  \operatorname{col}(\bm V_k).
\]
Apart from the lower bound on its smallest singular value, we impose no
restriction on $\bm B_k$. Thus the right loadings may arbitrarily mix the
shared and view-specific coordinates. We use the corresponding effective
strength $\gamma_k$ defined in \cref{eq:effective-strength}.

It remains to ensure that the common block $\bm U$ spans the entire
intersection of the two signal spaces. For a separation parameter
$\theta\in(0,1/2]$, assume
\begin{equation}
\label{eq:misalign}
  \|\bm V_1^\trans\bm V_2\|_{\mathrm{op}}
  \le 1-2\theta.
\end{equation}
The left-hand side is the cosine of the smallest principal angle between
the two view-specific subspaces. Hence any $\theta>0$ implies
$\operatorname{col}(\bm V_1)\cap
\operatorname{col}(\bm V_2)=\{\bm 0\}$.
Consequently,
$\cU=\cM_1\cap\cM_2=\operatorname{col}(\bm U)$.
Smaller $\theta$ allows the two view-specific subspaces to be more closely
aligned.

Our target is the shared projector
$\bm P_{\cU}=\bm U\bm U^\trans$.
For an estimator $\widehat{\bm P}$, we measure error by
$\mathbb{E}[\|\widehat{\bm P}-\bm P_{\cU}\|_{\mathrm{op}}]$.

\subsection{The BCSI estimator}
\label{sec:bcsi-estimator}

We now give the formal construction of the bias-corrected subspace
intersection estimator (BCSI). It is also summarized in Algorithm~\ref{alg:bcsi}.
For each view, recall that
$a_k=n-p_k$ and $b_k=d_k-p_k$.
Let
$\widehat s_{k,1}\ge\cdots\ge\widehat s_{k,p_k}$ be the leading
$p_k$ singular values of $\bm Y_k$, with corresponding left singular vectors
$\widehat{\bm\phi}_{k,1},\ldots,\widehat{\bm\phi}_{k,p_k}$, and write
$\widehat{\bm\Phi}_k
=[\widehat{\bm\phi}_{k,1},\ldots,\widehat{\bm\phi}_{k,p_k}]
\in\mathbb R^{n\times p_k}$.

\paragraph{Leakage bias correction.}
The directionwise correction motivated in
\cref{sec:intro-leakage-correction} depends on the unknown population
singular values. We therefore estimate them from the empirical singular
values using the spike-location relation in
\cref{eq:intro-spike-location}. For an empirical singular value $y$ above
the spectral edge, solving this relation for the squared population
singular value gives the upper root
\begin{equation}
\label{eq:tau}
  \tau_k(y)
  \coloneqq
  \begin{cases}
    \displaystyle
    \frac{
      y^2-(a_k+b_k)\sigma^2
      +
      \sqrt{
        \bigl\{y^2-(a_k+b_k)\sigma^2\bigr\}^2
        -4a_kb_k\sigma^4
      }
    }{2},
    &
    y>
    \sigma\bigl(\sqrt{a_k}+\sqrt{b_k}\bigr),
    \\[12pt]
    \sigma^2\sqrt{a_kb_k},
    &
    \text{otherwise}.
  \end{cases}
\end{equation}
At the spectral edge, the two roots coincide at
$\sigma^2\sqrt{a_kb_k}$, and we keep $\tau_k(y)$ at this boundary value
below the edge.

Substituting $\tau_k(y)$ into the oracle rescaling from
\cref{sec:intro-leakage-correction} gives the weight
\begin{equation}
\label{eq:omega}
  w_k(y)
  \coloneqq
  \left\{
    \frac{
      \tau_k(y)\bigl(\tau_k(y)+a_k\sigma^2\bigr)
    }{
      \tau_k(y)+b_k\sigma^2
    }
  \right\}^{1/2}.
\end{equation}
We then form the bias-corrected empirical factor
\begin{equation}
\label{eq:C}
  \bm C_k
  \coloneqq
  \widehat{\bm\Phi}_k
  \diag\!\left\{
    w_k(\widehat s_{k,1}),
    \ldots,
    w_k(\widehat s_{k,p_k})
  \right\}
  \in\mathbb R^{n\times p_k}.
\end{equation}
The columns of $\bm C_k$ span the same empirical signal space as
$\widehat{\bm\Phi}_k$, but the rescaling is chosen so that the leakage Gram
of $\bm C_k$ is approximately $a_k\sigma^2\bm I_{p_k}$.

We therefore subtract this common leakage contribution from the within-view
Gram blocks and define
\begin{equation}
\label{eq:G}
  \bm G
  \coloneqq
  \begin{pmatrix}
    \bm C_1^\trans\bm C_1-a_1\sigma^2\bm I_{p_1}
    &
    \bm C_1^\trans\bm C_2
    \\
    \bm C_2^\trans\bm C_1
    &
    \bm C_2^\trans\bm C_2-a_2\sigma^2\bm I_{p_2}
  \end{pmatrix}.
\end{equation}
After this correction, $\bm G$ is intended to approximate the Gram matrix
obtained by retaining only the components of $\bm C_1$ and $\bm C_2$
inside their population signal spaces. The latter has an $r$-dimensional
nullspace corresponding to the shared subspace.

Let
$\bm Q=(\bm Q_1^\trans,\bm Q_2^\trans)^\trans
\in\mathbb R^{(p_1+p_2)\times r}$
have orthonormal columns spanning the invariant subspace associated with the
$r$ eigenvalues of $\bm G$ nearest zero, where
$\bm Q_k\in\mathbb R^{p_k\times r}$.

\paragraph{Combining views of unequal strength.}
To account for unequal signal strengths, we follow the weighted aggregation
used in prior work on shared-subspace estimation
\citep{baharav2025stacked,li2025heterojive}.
Specifically, we set
\begin{equation}
\label{eq:M}
  \bm M
  \coloneqq
  \alpha_1\bm C_1\bm Q_1
  -
  \alpha_2\bm C_2\bm Q_2
  \in\mathbb R^{n\times r}, 
  \qquad \text{where }   \alpha_k
  \coloneqq
  \frac{\gamma_k^2}{\gamma_1^2+\gamma_2^2},
  \qquad k=1,2.
\end{equation}
Thus the stronger view contributes more to the reconstruction.
Define $ \bm J \coloneqq \diag(\alpha_1\bm I_{p_1},-\alpha_2\bm I_{p_2}), $ so that the BCSI reconstruction can be written as $\bm M=\bm C\bm J\bm Q$. 
Finally, define
$\widehat{\bm P}_{\rm BCSI}
\coloneqq
\bm P_{\operatorname{col}(\bm M)}$.

\begin{algorithm}[t]
\caption{Bias-corrected subspace intersection (BCSI)}
\label{alg:bcsi}

\textbf{Input:}
Data matrices $\bm Y_1,\bm Y_2$; ranks $r,r_1,r_2$;
noise level $\sigma$; signal strength lower bounds $s_1,s_2$.
\begin{itemize}[leftmargin=1.5em]

\item \textbf{Reweight the empirical signal directions.}
For each $k=1,2$, let $p_k=r+r_k$, $a_k=n-p_k$, and
$b_k=d_k-p_k$. Compute the leading $p_k$ left singular vectors
$\widehat{\bm\Phi}_k$ and singular values
$\widehat s_{k,1},\ldots,\widehat s_{k,p_k}$ of $\bm Y_k$, and set
\[
  \bm C_k
  =
  \widehat{\bm\Phi}_k
  \diag\!\left\{
    w_k(\widehat s_{k,1}),
    \ldots,
    w_k(\widehat s_{k,p_k})
  \right\},
\]
where $w_k(\cdot)$ is defined in \eqref{eq:omega}.

\item \textbf{Form the corrected intersection.}
Construct the centered Gram matrix
\[
\bm G=
\begin{pmatrix}
\bm C_1^\trans\bm C_1-a_1\sigma^2\bm I_{p_1}
&
\bm C_1^\trans\bm C_2
\\
\bm C_2^\trans\bm C_1
&
\bm C_2^\trans\bm C_2-a_2\sigma^2\bm I_{p_2}
\end{pmatrix}.
\]
Let $\bm Q_1$ and $\bm Q_2$ denote the two row blocks of an
orthonormal basis for the invariant subspace associated with the
$r$ eigenvalues of $\bm G$ nearest zero.

\item \textbf{Combine the two views.}
With
$\alpha_k=\gamma_k^2/(\gamma_1^2+\gamma_2^2)$, form
\[
\bm M
=
\alpha_1\bm C_1\bm Q_1
-
\alpha_2\bm C_2\bm Q_2.
\]
\item \textbf{Output.}
Return the projector
$\widehat{\bm P}_{\rm BCSI}
=\bm P_{\operatorname{col}(\bm M)}$.

\end{itemize}
\end{algorithm}

\section{Statistical guarantees}

\subsection{Upper bound}
\label{sec:upper-bound}

We first state the performance guarantee for BCSI.
\begin{theorem}
\label{thm:upper}
For $k=1,2$, suppose that $p_k\le(n\wedge d_k)/3$. If $\gamma_{\min} \ge\nu\sigma\sqrt n$ for some sufficiently large constant $\nu$, 
then the BCSI estimator obeys
\begin{align}
\label{eq:main}
  &\mathbb E
  \opn{
    \widehat{\bm P}_{\rm BCSI}-\bm P_{\cU}
  } \notag \\
  &\qquad 
  \le
  C\min\left\{
    1,\,
    \frac{\sigma\sqrt n}{\gamma_{\max}}
    +
    \frac{\sigma\sqrt{r_1+r_2}}{\gamma_{\min}}
    +
    \frac{
      \sigma\sqrt{r+ (r_1\wedge r_2)}
    }{
      \gamma_{\min}\sqrt\theta
    }
    +
    \frac{
      \sigma^2\sqrt{n\{r+(r_1\vee r_2)\}}
    }{
      \gamma_{\min}^2\theta
    }
  \right\}
\end{align}
for some universal constant $C>0$.
\end{theorem}

\paragraph{Comparison with AJIVE.}
Consider the homogeneous two-view regime with 
$\gamma_1\asymp\gamma_2\asymp\gamma$ and
$r\asymp r_1\asymp r_2$.
Specializing the recent high-probability bound for equal-weight AJIVE
in \citet{li2025heterojive} to this setting and suppressing logarithmic
factors gives
\[
\begin{aligned}
\text{\rm AJIVE:}\quad&
  \frac{\sigma\sqrt{nr}}{\gamma}
  +
  \frac{\sigma r}{\gamma\sqrt\theta}
  +
  \frac{\sigma^2 n}{\gamma^2\theta},
\qquad \text{and} \qquad 
\text{\rm BCSI:}\quad&
  \frac{\sigma\sqrt n}{\gamma}
  +
  \frac{\sigma\sqrt r}{\gamma\sqrt\theta}
  +
  \frac{\sigma^2\sqrt{nr}}{\gamma^2\theta}.
\end{aligned}
\]
The main difference is the quadratic term: bias correction
improves this term by a factor $\sqrt{n/r}$.
This gain is most consequential when $\theta$ is small, as the quadratic
term can then dominate the estimation error; we illustrate this regime in
the numerical experiments.
The BCSI bound also has sharper rank dependence in the first-order terms.

\paragraph{Signal strength assumption.}
The condition
$\gamma_k\gtrsim\sigma\sqrt n$ ensures that the population signal space $\cM_k$ in
each view can be estimated with nontrivial accuracy. The ratio
$\frac{\sigma \sqrt{n}}{\gamma_k}
\asymp 
\frac{\sigma \sqrt{n}}{s_k}+\frac{\sigma^2 \sqrt{nd_k}}{s_k^2}$
is also the minimax optimal rate for rectangular column-space estimation
\citep{cai2018rate,Cai2021Unbalanced}. This condition is natural for BCSI,
which first estimates each signal space before estimating their intersection.
It should not be interpreted as a necessary condition for estimating the
shared subspace; see \cref{sec:discussion}.

\paragraph{No condition-number dependence.}
An important feature of \cref{thm:upper} is that the bound depends only on the
smallest signal scale through $\gamma_1$ and $\gamma_2$. It imposes no upper
bound on the singular values of $\bm X_1$ or $\bm X_2$, and hence no
condition-number assumption on either signal matrix. This differs from
existing analyses of AJIVE~\citep{yang2025estimating,li2025heterojive} that require the signal matrices to be
well-conditioned.

\subsection{Lower bound and minimax optimality}
\label{sec:lower-bound}

To determine whether the dependence in \cref{eq:main} is unavoidable, we
consider the full class of signals satisfying the same rank, strength, and
separation conditions. For fixed $s_1,s_2>0$, $r,r_1,r_2\ge1$, and
$\theta\in(0,1/2]$, define
\begin{align}
\mathfrak P
\coloneqq
\Bigl\{
  &(\bm U,\bm V_1,\bm V_2,\bm B_1,\bm B_2):
  \bm U\in\mathbb R^{n\times r},
  \quad
  \bm V_k\in\mathbb R^{n\times r_k},
  \notag\\
  &[\bm U,\bm V_k]\in\operatorname{St}(n,p_k),
  \quad
  \bm B_k\in\mathbb R^{d_k\times p_k},
  \quad
  \sigma_{p_k}(\bm B_k)\ge s_k,
  \quad k=1,2,\notag\\
  &\|\bm V_1^\trans\bm V_2\|_{\mathrm{op}}
  \le1-2\theta
\Bigr\}.
\label{eq:parameter-class}
\end{align}

\begin{theorem}
\label{thm:lower}
For $k=1,2$, suppose that $p_k\le(n\wedge d_k)/3$. If $\gamma_{\min} \ge\nu\sigma\sqrt n$ for some sufficiently large constant $\nu$, then  
\begin{align}
\label{eq:minimax-lower}
  &\inf_{\widehat{\bm P}}
  \sup_{\substack{
    (\bm U,\bm V_1,\bm V_2,\bm B_1,\bm B_2)
    \in
    \mathfrak P
  }}
  \mathbb E
  \opn{
    \widehat{\bm P}-\bm U\bm U^\trans
  } \notag\\
  &\qquad\ge
  c\min\left\{
    1,\,
    \frac{\sigma\sqrt n}{\gmax}
    +
    \frac{
      \sigma\sqrt{r+ (r_1\wedge r_2)}
    }{
      \gamma_{\min}\sqrt\theta
    }
    +
    \frac{
      \sigma^2\sqrt{n\{{r+ (r_1\wedge r_2)}\} }
    }{
      \gamma_{\min}^2\theta
    }
  \right\}
\end{align}
for some universal constant $c>0$.
\end{theorem}

When $r\asymp r_1\asymp r_2$,\footnote{In fact, the upper and lower bounds match in a broader setting when $r+(r_1\vee r_2)\asymp r+(r_1\wedge r_2)$.} \cref{thm:upper,thm:lower} match up to
numerical constants and yield the three-term rate in
\cref{eq:intro-upper-rate}.  Thus, in the signal-strength regime considered
here, the $\theta$-independent term, the $\theta^{-1/2}$ term, and the
$\theta^{-1}$ second-order term are all intrinsic.  For highly unbalanced
view-specific ranks, a gap remains in the rank dependence: the upper bound
depends on the larger individual rank, whereas the lower bound depends on
the smaller one~$r_1\wedge r_2$.

\paragraph{Comparison with existing lower bounds.}
The closest prior information-theoretic result is the refined minimax lower
bound of \citet[Theorem~5]{yang2025estimating}. Their result considers a
homogeneous $K$-view model with a common column dimension and signal strength and
with equal shared and view-specific ranks. Specializing their bound to $K=2$, and
writing $s$ for the common signal strength and $d$ for the common column dimension, gives, up to numerical constants,
\[
  \frac{\sigma\sqrt n}{\gamma}
  +
  \frac{\sigma\sqrt r}{\gamma\sqrt\theta},
  \qquad \text{where} \qquad
  \gamma
  =
  \frac{s^2}{\sqrt{s^2+d\sigma^2}}.
\]
Thus the prior lower bound captures the ambient and $\theta^{-1/2}$
difficulties, including the rectangular-matrix effect encoded by the effective
strength $\gamma$. Our result recovers these scales while allowing unequal
view dimensions, effective strengths, and view-specific ranks, and additionally
identifies the $\theta^{-1}$ obstruction
$
  \frac{\sigma^2\sqrt{nr}}{\gamma^2\theta}
$
in the balanced-rank regime.

The proof of \cref{thm:lower} is technically involved in general ranks, but the three statistical obstructions are already visible when
$r=r_1=r_2=1$.  We therefore use the rank-one case to explain where the three terms come from.  
The general principle behind the lower bound is to choose the shared direction $\bm u$, the individual directions $\bm v_1,\bm v_2$, and the right loadings so that two or more signal instances have well-separated shared directions but nearly indistinguishable observation laws.

\paragraph{Understanding the $\theta$-independent term.}
For the first term, suppose that the individual directions $\bm v_1$ and
$\bm v_2$ are fixed and are even revealed to the estimator.  We then vary only
the shared direction $\bm u$ inside
$\operatorname{span}\{\bm v_1,\bm v_2\}^{\perp}$, a space whose
dimension is of order $n$.  There is no difficulty here in deciding which
directions are shared and which are individual: the only task is to locate
$\bm u$ from two noisy views.  Since both views change with $\bm u$, their
information adds, giving the combined effective strength
$\sqrt{\gamma_1^2+\gamma_2^2}\asymp\gamma_{\max}$.  Packing the possible
shared directions over an $n$-dimensional space therefore gives the lower-bound
scale $\frac{\sigma\sqrt n}
       {\sqrt{\gamma_1^2+\gamma_2^2}}
  \asymp
  \frac{\sigma\sqrt n}{\gamma_{\max}}.$ 
This term is independent of $\theta$ because the individual directions are
already known.  It represents the unavoidable cost of estimating the
shared direction itself. 
The proof follows a similar strategy in the single-matrix subspace estimation literature~\citep{Cai2021Unbalanced}.

\paragraph{Understanding the $1/\sqrt{\theta}$ term.}
The second construction isolates a different difficulty: even if the signal
space of one view is perfectly known, its decomposition into shared and
individual directions is not. Suppose without loss of generality that
$\gamma_1\ge\gamma_2$, so that View~2 is the weaker view. Fix orthonormal
vectors $\bm e_1,\bm e_2,\bm z_0\in\mathbb R^n$. Take
$\bm e_1$ to be the shared direction and $\bm e_2$ to be the individual
direction in View~1, so that its signal space is the fixed plane
$\operatorname{span}\{\bm e_1,\bm e_2\}$.

To vary this decomposition without changing the first-view signal space,
rotate the shared and individual directions inside the same plane by an angle $t$:
\begin{equation}
\label{eq:lower-rotation-path}
  \bm u_t
  =
  \cos(t)\bm e_1+\sin(t)\bm e_2,
  \qquad
  \bm v_{1,t}
  =
  -\sin(t)\bm e_1+\cos(t)\bm e_2.
\end{equation}
Since $[\bm u_t,\bm v_{1,t}]$ is merely a rotation of
$[\bm e_1,\bm e_2]$, the first-view signal matrix can be chosen to be exactly
the same for every $t$ and may therefore be revealed. For View~2, define
\[
  \bm v_{2,t}
  =
  (1-2\theta)\bm v_{1,t}
  +
  2\sqrt{\theta(1-\theta)}\,\bm z_0.
\]
Then
$\bm v_{1,t}^\trans\bm v_{2,t}=1-2\theta$, so every $t$ satisfies
the prescribed separation condition. Thus all information about the unknown
rotation $t$ must come from View~2.

When $\theta$ is small, $\bm v_{2,t}$ is nearly aligned with
$\bm v_{1,t}$. As $t$ changes, the shared direction $\bm u_t$ and
the individual direction $\bm v_{2,t}$ therefore rotate almost coherently.
The component that distinguishes different values of $t$ lies in the fixed
direction $\bm z_0$ and has size of order $\sqrt{\theta}$. Consequently,
View~2 can distinguish rotations only at the scale
$\frac{\sigma}{\gamma_2\sqrt\theta}
  =
  \frac{\sigma}{\gamma_{\min}\sqrt\theta}$. 
Thus the second term is the cost of identifying which direction inside an
otherwise known signal plane is shared.

\paragraph{Understanding the $1/\theta$ term.}
The third construction starts from exactly the same family and makes only one
change: the fixed direction $\bm z_0$ that helps View~2 resolve the rotation is
now replaced by an unknown direction
$\bm z\in\operatorname{span}\{\bm e_1,\bm e_2\}^{\perp}$:
\[
  \bm v_{2,t,\bm z}
  =
  (1-2\theta)\bm v_{1,t}
  +
  2\sqrt{\theta(1-\theta)}\,\bm z.
\]
For each fixed $t$, we randomize $\bm z$ over an $(n-2)$-dimensional family.
This changes only the individual direction in View~2 and leaves the target
$\bm u_t \bm u_t^\trans$ fixed. The unknown parameter of interest is still the
same scalar rotation $t$, but the reference direction that helped reveal
this rotation in the preceding construction is now itself hidden.

With the fixed $\bm z_0$, View~2 can compare the rotating signal plane against a
known direction outside that plane. After $\bm z_0$ is replaced by the hidden
direction $\bm z$, this reference must itself be learned from the data.
Averaging over the high-dimensional nuisance direction removes the
first-order directional information available in the preceding construction.
View~2 must now simultaneously infer the hidden direction and determine the
rotation $t$, leaving only a second-order amount of information about the
latter.  Hiding $\bm z$ in a space   whose dimension is of order $n$ leads to
indistinguishable rotations of size
$  \frac{\sigma^2\sqrt n}
       {\gamma_2^2\theta}
  =
  \frac{\sigma^2\sqrt n}
       {\gamma_{\min}^2\theta},$ which gives the third lower-bound term.

\section{Numerical experiments}
\label{sec:numerical-experiments}

In this section, we use synthetic and real data to compare BCSI with AJIVE.

\subsection{Synthetic data}
\label{sec:numerical-synthetic}

\paragraph{Data-generating process.}
We consider the model with $d_1=d_2=n$, $\sigma=1$, and
$r=r_1=r_2=1$. Thus each signal matrix has rank two, with one shared
direction and one view-specific direction. Let
$\bm e_1,\bm e_2,\bm e_3$ be the first three standard basis vectors in $\mathbb{R}^n$. 
Set the signal strength $s_n = n^{2/3}$, and the separation parameter $  \theta_n=\frac{n^{-1/3}}{4}$. We take the shared direction to be
$\bm u=\bm e_1$ and define
\[
  \bm v_1
  =
  \sqrt{1-\theta_n}\,\bm e_2+\sqrt{\theta_n}\,\bm e_3,
  \qquad
  \bm v_2
  =
  \sqrt{1-\theta_n}\,\bm e_2-\sqrt{\theta_n}\,\bm e_3.
\]
Then the two signal matrices and observations are generated
according to
\[
  \bm X_k
  =
  [\bm u,\bm v_k]\,
  \frac{1}{\sqrt{2}}
  \begin{pmatrix}
    1 & -1\\
    1 &  1
  \end{pmatrix}
  \begin{pmatrix}
    s_n & 0\\
    0 & 2s_n
  \end{pmatrix}
  [\bm e_1,\bm e_2]^\trans,
  \qquad
  \bm Y_k=\bm X_k+\bm E_k,
  \qquad k=1,2,
\]
where the entries of $\bm E_1$ and $\bm E_2$ are mutually independent
standard Gaussian random variables. The $2\times2$ matrix in this
display rotates $\bm u$ and $\bm v_k$ by $45^\circ$. Consequently, the
left singular vectors of $\bm X_k$ are
$(\bm u+\bm v_k)/\sqrt{2}$ and
$(-\bm u+\bm v_k)/\sqrt{2}$, with singular values $s_n$ and $2s_n$.
Moreover, $\bm v_1^\trans\bm v_2=1-2\theta_n$, so the two view-specific
subspaces become more closely aligned as $n$ increases. The rotation
ensures that the shared and view-specific directions are not themselves
singular vectors, allowing unequal leakage from the two empirical
singular directions to bias the intersection estimated by AJIVE. Under
this scaling, the leading term in the BCSI rate decreases as
$n^{-1/6}$, whereas the quadratic bias scale for AJIVE remains constant.

\paragraph{Estimators being compared.}
We compare four procedures. The first is AJIVE, obtained from the
leading rank-one eigenspace of the sum of the two empirical rank-two
signal-space projectors. The second is Expected AJIVE: at each value of
$n$, we average the final rank-one AJIVE projectors across Monte Carlo
replications and take the leading eigenspace of this average. Expected
AJIVE is included as a diagnostic for systematic bias and is not an
estimator computed from a single dataset. The third procedure is oracle
BCSI, which uses the true noise level and the true effective signal
strength
$\gamma_n=s_n^2/\sqrt{s_n^2+n}$. The fourth is estimated-parameter BCSI,
which estimates the noise level from the pooled rank-two residual energy
and estimates the smallest signal singular value separately in each
view. All data-dependent procedures use marginal rank two and joint
rank one.

\begin{figure}[t]
  \centering
  \includegraphics[width=0.76\textwidth]
    {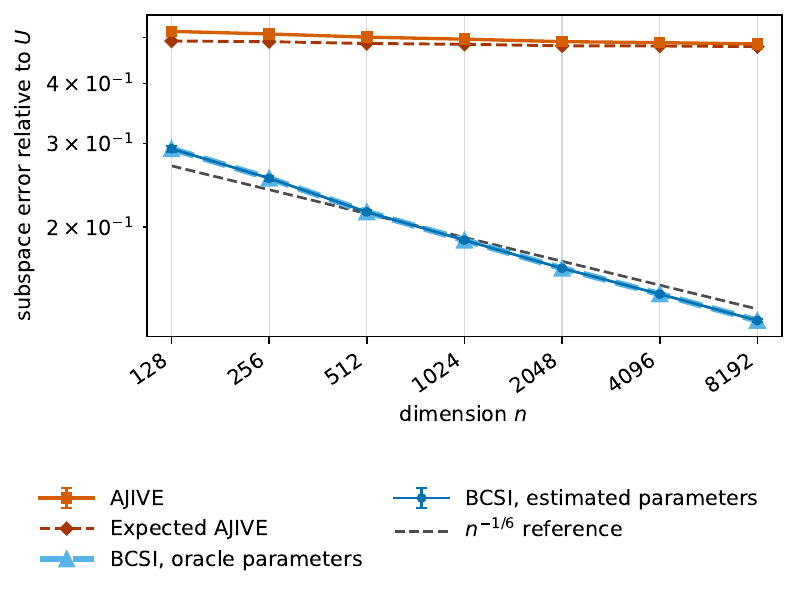}
  \caption{Comparison of AJIVE and BCSI as the dimension increases.
  The points for AJIVE and BCSI are mean projector errors, and the error
  bars are $\pm2$ Monte Carlo standard errors. Expected AJIVE is the
  leading eigenspace of the Monte Carlo mean of the final AJIVE
  projectors. Oracle BCSI uses the true noise level and effective signal
  strength, whereas estimated-parameter BCSI estimates them from the
  data. The dashed line is an $n^{-1/6}$ reference.}
  \label{fig:synthetic-ajive-bcsi}
\end{figure}

\paragraph{Experimental results.}
\Cref{fig:synthetic-ajive-bcsi} shows that the errors of both AJIVE and
Expected AJIVE remain nearly constant as the dimension increases. The
persistence of the Expected-AJIVE error indicates that this behavior is
caused by systematic bias rather than only by random fluctuations. In
contrast, the BCSI error decreases with the dimension and closely
follows the predicted $n^{-1/6}$ behavior. The oracle and
estimated-parameter BCSI curves are nearly indistinguishable, indicating
that estimating the required parameters has little effect in this
experiment. The results therefore support the claim that correcting the
direction-dependent leakage removes the persistent error suffered by
AJIVE in this regime.

\subsection{Real data}
\label{sec:numerical-real-data}

\paragraph{Data and scientific question.}
We apply AJIVE and BCSI to the Nutrimouse data
\citep{martin2007nutrigenomic}. The data contain gene-expression
measurements for 120 genes and concentrations of 21 liver fatty acids
from the same 40 mice. The mice belong to two genotypes and five diet
groups, with four mice in each genotype-by-diet group. Previous analyses
of these data suggest that genotype is represented in the variation
shared by the gene-expression and lipid views, whereas diet is primarily
represented in the lipid-specific variation
\citep{yuan2022exponential,sergazinov2026spectral}. We therefore examine
whether the estimated joint subspace retains variation associated with
genotype while excluding variation associated with diet.

\begin{figure}[t]
  \centering
  \includegraphics[width=\textwidth]
    {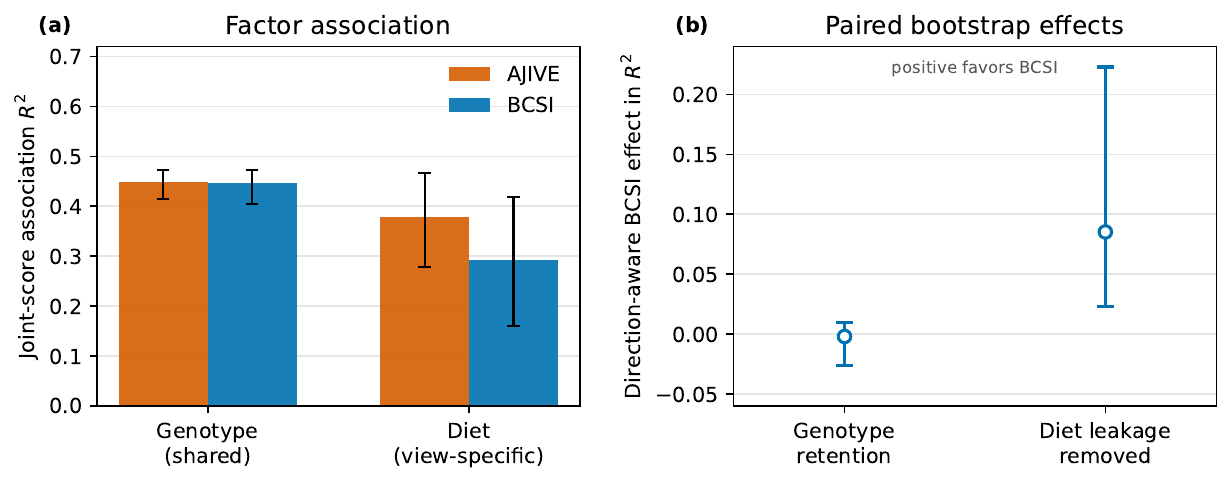}
  \caption{Matched-rank AJIVE and BCSI on the Nutrimouse data. Panel (a)
  reports the association of the estimated joint scores with genotype and
  diet. The error bars are 95\% percentile intervals from 2,000 paired
  bootstrap refits. Panel (b) reports the paired differences, with positive
  values favoring BCSI. For genotype, the difference is BCSI minus AJIVE;
  for diet, it is AJIVE minus BCSI.}
  \label{fig:nutrimouse-bcsi-ajive}
\end{figure}

\paragraph{Estimators and rank selection.}
We center and standardize every feature and then normalize each view by
its estimated residual noise level. We use marginal rank $p_1=3$ for
gene expression and $p_2=4$ for lipids. These ranks were used in the
earlier Nutrimouse analysis of \citet{yuan2022exponential}.
\citet{sergazinov2026spectral} observed that they agree with the elbows
of the two scree plots and give clearer spectral separation than the
larger ranks selected by an automatic singular-value threshold. With
these marginal ranks, AJIVE selects joint rank $r=2$. We supply the same
marginal and joint ranks to AJIVE and BCSI, so the comparison concerns
the estimated joint subspaces rather than differences in rank selection.
AJIVE is computed from the two empirical marginal subspaces. BCSI uses
the same empirical marginal subspaces but applies the proposed leakage
correction, with the noise level and signal strengths estimated from the
data.

\paragraph{Evaluation and uncertainty.}
For each estimated joint subspace, we measure the fraction of its score
variation explained by genotype and by diet. This quantity is one minus
the SWISS criterion of \citet{cabanski2010swiss}. A larger genotype
association indicates that the joint scores retain more of the factor
expected to be shared, whereas a smaller diet association indicates
that they retain less of the factor expected to be view-specific. The
genotype and diet labels are used only for this evaluation and do not
enter either estimator. To assess uncertainty, we use 2,000 paired
bootstrap refits. Within each bootstrap replication, we sample four mice
with replacement from each of the ten genotype-by-diet groups. Each
mouse is resampled together with both of its data views, and AJIVE and
BCSI are fitted to the same resampled dataset. This preserves the paired
multi-view structure and the size of every experimental group. We form
95\% percentile intervals from the bootstrap distributions of the
individual metrics and their paired differences.

\paragraph{Interpretation of the results.}
\Cref{fig:nutrimouse-bcsi-ajive} shows that the joint scores estimated
by AJIVE and BCSI have essentially the same association with genotype,
and the paired bootstrap interval does not indicate a meaningful
difference between the methods. In contrast, the BCSI joint scores have
a weaker association with diet, and the paired bootstrap interval favors
BCSI. This suggests that BCSI removes more variation associated with the
view-specific factor without a corresponding loss of variation associated
with the shared factor.

\section{Analysis of BCSI}
\label{sec:upper-analysis}

We now develop the main ingredients for analyzing BCSI. Throughout this
section, we view the corrected Gram matrix as a perturbation of a signal-only
Gram matrix that encodes the shared subspace.

 For each view $k=1,2$, decompose the bias-corrected empirical factor into its component inside the corresponding signal space and its orthogonal leakage:
\begin{equation}
\label{eq:SL}
  \bm C_k
  =
  \bm S_k+\bm L_k,
  \qquad
  \bm S_k
  \coloneqq
  \bm P_{\cM_k}\bm C_k,
  \qquad
  \bm L_k
  \coloneqq
  \bm P_{\cM_k^\perp}\bm C_k.
\end{equation}
Here $\bm S_k$ is the part of the bias-corrected empirical factor that remains inside the population signal space, whereas $\bm L_k$ records its leakage into the orthogonal complement. 

With $\bm C\coloneqq[\bm C_1,\bm C_2]$, 
$\bm S\coloneqq[\bm S_1,\bm S_2]$, 
  $\bm L\coloneqq[\bm L_1,\bm L_2],$ and 
$\bm D_\sigma
  \coloneqq
  \diag(
    a_1\sigma^2\bm I_{p_1},
    a_2\sigma^2\bm I_{p_2}
  )$, 
the bias-corrected Gram matrix in
\cref{eq:G} admits the exact decomposition
\begin{equation}
\label{eq:Delta}
  \bm G
  =
  \underbrace{\vphantom{\bigl(\bm L^\trans\bm L-\bm D_\sigma\bigr)}\bm S^\trans\bm S}_{\eqqcolon \displaystyle \bm G_0}
  +
  \underbrace{
    \bm S^\trans\bm L
    +
    \bm L^\trans\bm S
    +
    \bigl(\bm L^\trans\bm L-\bm D_\sigma\bigr)
  }_{\eqqcolon \displaystyle \bm\Delta}.
\end{equation}
Thus the natural oracle matrix is
$\bm G_0=\bm S^\trans\bm S$. The first two terms in
$\bm\Delta$ are linear in the leakage, while
$\bm L^\trans\bm L-\bm D_\sigma$ is the centered quadratic leakage.
This decomposition makes precise the role of the bias correction:
the weighting makes $\bm L^\trans\bm L$ close to the common leakage level
$\bm D_\sigma$, which is then subtracted so that $\bm G$ is close to the
signal-only Gram matrix $\bm G_0$.

\paragraph{Proof roadmap.} The proof proceeds in four steps. First, we show that the oracle matrix
$\bm G_0=\bm S^\trans\bm S$ has an $r$-dimensional nullspace that exactly
encodes the shared subspace, and that its positive spectrum is separated from
zero by a gap of order $\gmin^2\theta$. Second, we develop a relative
perturbation bound showing that the $r$ eigenvectors $\bm Q$ of
$\bm G=\bm G_0+\bm\Delta$ nearest zero approximately satisfy the oracle
nullspace relation, i.e., $\opn{\bm S\bm Q}$ is small.
Third, we translate this approximate relation into an error bound for the
reconstructed shared subspace and decompose the error into a leading term
linear in the leakage and a higher-order remainder. Finally, the stochastic
analysis controls these leakage terms, verifies the relative perturbation
condition, and yields the rate in \cref{thm:upper}. The first three steps are
developed in \cref{sec:deterministic-analysis}, and the final step in
\cref{sec:stochastic-analysis}; proofs of the main auxiliary results are
deferred to the appendix.

\subsection{Deterministic analysis}
\label{sec:deterministic-analysis}

Fix a sufficiently small numerical constant $\cemb>0$ and define
\begin{equation}
\label{eq:Eemb}
  \cE_{\rm emb}
  \coloneqq
  \bigcap_{k=1}^2
  \left\{
    \sigma_{\min}(\bm S_k)\ge \cemb\gamma_k
  \right\}.
\end{equation}
On this event, each $\bm S_k$ has full column rank and therefore
$\operatorname{col}(\bm S_k)=\cM_k$. We work under this deterministic
condition throughout the next several arguments; its probability is
controlled later.

\subsubsection{Oracle geometry}

Since BCSI selects the eigenvectors of $\bm G$ nearest zero, we first
understand the nullspace of the oracle matrix $\bm G_0=\bm S^\trans\bm S$. 
This oracle matrix should have an
$r$-dimensional nullspace encoding the coefficient relations that produce
the shared subspace. To exhibit these relations, define
\[
  \bm Q_\star
  \coloneqq
  \begin{pmatrix}
    \bm S_1^\dagger\bm U\\
    -\bm S_2^\dagger\bm U
  \end{pmatrix}
   \in\mathbb R^{(p_1+p_2)\times r}.
\]

\begin{lemma}
\label{lem:oracle-geometry}
On $\cE_{\rm emb}$,
\[
  \ker(\bm G_0)=\ker(\bm S),
  \qquad
  \dim\ker(\bm G_0)=r.
\]
Moreover,
\begin{equation}
\label{eq:Q_star}
  \bm S\bm Q_\star=\bm 0,
  \qquad
  \bm S\bm J\bm Q_\star=\bm U,
\end{equation}
and
\begin{equation}
\label{eq:oracle-curvature}
  \lambda_{\min}^{+}(\bm G_0)
  \ge
  2\cemb^2\gmin^2\theta.
\end{equation}
\end{lemma}

Thus the nullspace of $\bm G_0$ consists exactly of the coefficient
relations that reconstruct the shared subspace, while
\cref{eq:oracle-curvature} separates this nullspace from the positive
spectrum. These are the two properties needed to analyze the perturbed
near-zero eigenspace of $\bm G$.

\subsubsection{Perturbation of the nullspace}

The standard Davis--Kahan
argument tells us that if the perturbation size $\opn{\bm\Delta}$ is small,
then the near-zero eigenspace of $\bm G=\bm G_0+\bm\Delta$ stays close to
the oracle nullspace $\ker(\bm G_0)$.  However, this does not suffice for
two reasons.  First, by \cref{lem:oracle-geometry}, the oracle nullspace
$\ker(\bm G_0)=\ker(\bm S)$ exactly encodes the coefficient relations that
produce the shared subspace.  Thus what matters for reconstruction is not
the Euclidean distance between the two coefficient spaces, but how well the
empirical eigenspace continues to satisfy the oracle relation, as measured by
$\opn{\bm S\bm Q}$.  Second, a perturbation bound based only on
$\opn{\bm\Delta}$ is too crude.  Indeed, the linear part of $\bm\Delta$
contains $\bm S^\trans\bm L+\bm L^\trans\bm S$, whose operator norm can
scale with the largest singular value of $\bm S$.  Such a bound would
therefore introduce an unnecessary upper signal scale, even though
perturbations along stronger signal directions are accompanied by
correspondingly larger eigenvalues of $\bm G_0$.

To address these two issues, we use the following self-contained deterministic
perturbation result. Let $d\ge 2$, let
$\bm G_0\in\mathbb R^{d\times d}$ be symmetric and positive semidefinite, and
let $\bm\Delta\in\mathbb R^{d\times d}$ be symmetric. Set
$\bm G\coloneqq\bm G_0+\bm\Delta$ and define
\begin{equation}
\label{eq:relative-blocks}
\begin{aligned}
  \bm\Delta_{00}
  &\coloneqq
  \bm P_{\ker(\bm G_0)}\bm\Delta\bm P_{\ker(\bm G_0)},
  \\
  \bm\Delta_{+0}
  &\coloneqq
  (\bm G_0^\dagger)^{1/2}\bm\Delta\bm P_{\ker(\bm G_0)},
  \\
  \bm\Delta_{++}
  &\coloneqq
  (\bm G_0^\dagger)^{1/2}\bm\Delta(\bm G_0^\dagger)^{1/2}.
\end{aligned}
\end{equation}
These three matrices describe the perturbation relative to the
decomposition
$\ker(\bm G_0)\oplus\ker(\bm G_0)^\perp$, after scaling the
positive-eigenvalue directions by $(\bm G_0^\dagger)^{1/2}$.  Specifically,
$\bm\Delta_{00}$ acts within the oracle nullspace,
$\bm\Delta_{+0}$ maps the oracle nullspace into its orthogonal complement,
and $\bm\Delta_{++}$ acts within that orthogonal complement.

Define the perturbation size
\(\varepsilon_\Delta
\coloneqq
\opn{\bm\Delta_{00}}
+2\opn{\bm\Delta_{+0}}^2\).
\begin{lemma}
\label{lem:nullspace}
Assume that $\dim\ker(\bm G_0)=r$ for some $1\le r< d$, and $\bm G = \bm G_0 + \bm \Delta$ with
\[
  \opn{\bm\Delta_{++}}\le\frac14,
  \qquad
  \varepsilon_\Delta
  \le\frac{1}{16}\lambda_{\min}^{+}(\bm G_0).
\]
Then $\bm G$ has exactly $r$ eigenvalues in
$[-\varepsilon_\Delta,\varepsilon_\Delta]$,
while all remaining eigenvalues exceed $4\varepsilon_\Delta$.
If $\bm Q\in\operatorname{St}(d,r)$ spans the invariant subspace
associated with these $r$ eigenvalues, then
\begin{equation}
\label{eq:abstract-relative-nullspace-conclusion}
  \opn{\bm G_0^{1/2}\bm Q}
  \le
  2\opn{\bm\Delta_{+0}}.
\end{equation}
\end{lemma}

In our application, $d=p_1+p_2$ and $\bm G_0=\bm S^\trans\bm S$. Hence
$\bm S\bm Q$ and $\bm G_0^{1/2}\bm Q$ have the same Gram matrix, so
\cref{lem:nullspace} gives
\begin{equation}
\label{eq:relative-nullspace-conclusion}
  \opn{\bm S\bm Q}
  =
  \opn{\bm G_0^{1/2}\bm Q}
  \le
  2\opn{\bm\Delta_{+0}}.
\end{equation}
Moreover, $\bm S\bm P_{\ker(\bm G_0)}=\bm 0$. Consequently, the linear
leakage terms in $\bm\Delta$ vanish from $\bm\Delta_{00}$, leaving only the
centered quadratic leakage.

We therefore define the good event
\begin{equation}
\label{eq:Egood}
  \cE_{\rm good}
  \coloneqq
  \cE_{\rm emb}
  \cap
  \left\{
    \opn{\bm\Delta_{++}}\le\frac14,\quad
    \opn{\bm\Delta_{00}}
    +2\opn{\bm\Delta_{+0}}^2
    \le
    \frac18\cemb^2\gmin^2\theta
  \right\}.
\end{equation}
By \cref{lem:oracle-geometry,eq:oracle-curvature}, the conditions of
\cref{lem:nullspace} hold on $\cE_{\rm good}$. Therefore, the invariant subspace of $\bm G$ identified in \cref{lem:nullspace} is
precisely the near-zero eigenspace selected by BCSI and satisfies
\cref{eq:relative-nullspace-conclusion}. Since the oracle relation is
$\bm S\bm Q_\star=\bm 0$, this bound shows that the empirical eigenvectors
approximately satisfy the same nullspace relation.

\subsubsection{From the near-zero eigenspace perturbation to subspace error}

We now translate the bound on
$\|\bm S\bm Q\|_{\mathrm{op}}$ into a bound on the error of the estimated shared subspace.
The orthonormal normalization of $\bm Q$ is natural for the eigenspace
problem, whereas $\bm Q_\star$ is a basis of the oracle nullspace normalized so that
$\bm S\bm J\bm Q_\star=\bm U$.  Even when $\bm Q$ spans the exact oracle
kernel, $\bm S\bm J\bm Q$ need only equal $\bm U$ up to an invertible right
factor.  We remove this arbitrary factor before measuring subspace error.

Recall that $\bm M=\bm C\bm J\bm Q$. On $\cE_{\rm good}$, $\bm U^\trans\bm M$ is invertible (see
\cref{lem:calibration}).  Define
$
  \widetilde{\bm M}
  \coloneqq
  \bm M(\bm U^\trans\bm M)^{-1}.
$
This right transformation does not change the estimated column space, while
$\bm U^\trans\widetilde{\bm M}=\bm I_r$.  Consequently,
\begin{equation}
\label{eq:one-sided-reconstruction}
  \opn{
    \widehat{\bm P}_{\rm BCSI}-\bm P_{\cU}
  }
  \le
  \opn{
    \bm P_{\cU^\perp}\widetilde{\bm M}
  }.
\end{equation}

We next decompose the right-hand side into its leading and higher-order
terms.  To state the decomposition, define the following linear map.
Every $\bm x\in\cM_1+\cM_2$ has a unique representation
$\bm x=\bm U\bm a+\bm V_1\bm b_1+\bm V_2\bm b_2$.  Define
$\bm T:\mathbb R^n\to\mathbb R^n$ by
\begin{equation}
\label{eq:transfer-operator}
  \bm T\bm x
  \coloneqq
  \alpha_1\bm V_1\bm b_1-\alpha_2\bm V_2\bm b_2
  \quad\text{for }\bm x\in\cM_1+\cM_2,
  \qquad
  \bm T\bm x\coloneqq\bm 0
  \quad\text{for }\bm x\in(\cM_1+\cM_2)^\perp.
\end{equation}

\begin{lemma}
\label{lem:error-decomposition}
On $\cE_{\rm good}$, the error has the decomposition
\begin{equation}
\label{eq:error-decomposition}
  \bm P_{\cU^\perp}\widetilde{\bm M}
  =
  (\bm L\bm J-\bm T\bm L)\bm Q_\star
  +
  \bm{\mathsf{Rem}},
\end{equation}
where the higher-order remainder satisfies
\begin{equation}
\label{eq:higher-order-deterministic}
  \opn{\bm{\mathsf{Rem}}}
  \lesssim
  \frac{1}{\gmin^2\theta}
  \left[
    \opn{\bm L^\trans\bm L-\bm D_\sigma}
    +
    \left\{
      \sqrt\theta\,\opn{\bm L}
      +
      \opn{\bm P_{\col(\bm S)}\bm L}
      +
      \opn{\bm\Delta_{+0}}
    \right\}
    \opn{\bm\Delta_{+0}}
  \right].
\end{equation}
\end{lemma}

The two leading terms in \cref{eq:error-decomposition} are linear in the
leakage $\bm L$, while $\bm{\mathsf{Rem}}$ contains only second- and
higher-order terms.  The stochastic analysis controls these terms
separately.

\subsection{Stochastic analysis}
\label{sec:stochastic-analysis}

The deterministic analysis reduces the proof to controlling the first-order
term in \cref{eq:error-decomposition}, the remainder in
\cref{eq:higher-order-deterministic}, and the probability of
$\cE_{\rm good}$. 

The main probabilistic input is the following one-view
result, which quantifies both the size and the centered quadratic fluctuation
of the weighted leakage.\footnote{The results hold for general moments. We use $q = 32$ throughout for concreteness and to simplify notation. }

\begin{lemma}
\label{lem:one-view}
Fix $k\in\{1,2\}$. 
Suppose that $p_k \leq (n \wedge d_k) / 3$ and that $\gamma_k \geq \nu \sigma \sqrt{n}$ for a sufficiently large constant $\nu$. Then there exists an event $\cE_k$ such that
$\sigma_{\min}(\bm S_k)\ge \cemb\gamma_k$ on $\cE_k$ and
$\mathbb P(\cE_k^c)\lesssim
\sigma^2\sqrt{np_k}/\gamma_k^2$.
Moreover,
\begin{equation}
\label{eq:one-view-bounds}
  \Lqn{
    \bm L_k^\trans\bm L_k-a_k\sigma^2\bm I_{p_k}
  }{32}
  \lesssim
  \sigma^2\sqrt{a_kp_k},
  \qquad
  \Lqn{\bm L_k}{64}
  \lesssim
  \sigma\sqrt{a_k}.
\end{equation}
\end{lemma}

The lemma can be viewed as a nonasymptotic counterpart of singular-vector
overlap results from spiked random matrix theory. Classical rectangular spiked-matrix theory characterizes outlier singular
values and singular-vector overlaps in the fixed-rank asymptotic regime
\citep{benaych2012singular,gavish2017optimal,bao2021singular}. More recent
work permits the signal rank to grow with the matrix dimensions, including
sublinear- and proportional-rank regimes
\citep{liu2025asymptotic,landau2023singular,shmalo2026mesoscopic}. These
results are mainly asymptotic.  
The purpose of \cref{lem:one-view} is different. It gives finite-sample moment
control of the full, data-dependently weighted leakage Gram that enters the
analysis of BCSI. This conclusion is obtained under the comparatively strong
signal-strength condition $\gamma_{k}\ge\nu\sigma\sqrt n$. Thus the lemma
is complementary to the existing asymptotic theory. To our knowledge, this particular finite-sample estimate
has not appeared previously.

For BCSI, this result provides the finite-sample justification for subtracting
$a_k\sigma^2\bm I_{p_k}$ from the within-view Gram blocks. Together with the
rotational invariance of Gaussian noise, it also supplies the probabilistic
input needed to control the linear and higher-order leakage terms.  We defer these probabilistic calculations to the appendix and
record their consequences below.

\begin{lemma}
\label{lem:stochastic-master}
Under the assumptions of \cref{thm:upper}, suppose in addition that the sum
of the four rate terms on the right-hand side of \cref{eq:main}, before
truncation at $1$, is at most a sufficiently small numerical constant.
Then
\begin{subequations}
\begin{align}
  \mathbb E\left[
    \ind_{\cE_{\rm good}}
    \opn{(\bm L\bm J-\bm T\bm L)\bm Q_\star}
  \right]
  &\lesssim
  \frac{\sigma\sqrt n}{\gamma_{\max}}
  +
  \frac{\sigma\sqrt{r_1+r_2}}{\gamma_{\min}}
  +
  \frac{
    \sigma\sqrt{r+ (r_1\wedge r_2)}
  }{
    \gamma_{\min}\sqrt\theta
  },
  \label{eq:first-order-expectation}\\
  \mathbb E\left[
    \ind_{\cE_{\rm good}}
    \opn{\bm{\mathsf{Rem}}}
  \right]
  &\lesssim
  \frac{
    \sigma^2\sqrt{n\{r+(r_1\vee r_2)\}}
  }{
    \gamma_{\min}^2\theta
  },
  \label{eq:higher-order-expectation}\\
  \mathbb P(\cE_{\rm good}^c)
  &\lesssim
  \frac{\sigma\sqrt{r_1+r_2}}{\gamma_{\min}}
  +
  \frac{
    \sigma\sqrt{r+ (r_1\wedge r_2)}
  }{
    \gamma_{\min}\sqrt\theta
  }
  +
  \frac{
    \sigma^2\sqrt{n\{r+(r_1\vee r_2)\}}
  }{
    \gamma_{\min}^2\theta
  }. \label{eq:G-complement}
\end{align}
\end{subequations}

\end{lemma}

\paragraph{Completing the proof of~\cref{thm:upper}.}
Combining
\cref{eq:one-sided-reconstruction,eq:error-decomposition}
with \cref{lem:stochastic-master}, and using that the distance between two
orthogonal projectors is at most one on $\cE_{\rm good}^c$, proves
\cref{eq:main} when the untruncated rate is sufficiently small.  The
remaining case follows from the same trivial bound after enlarging the
universal constant.

\section{Discussion}
\label{sec:discussion}
We propose a bias-corrected subspace intersection estimator (BCSI) for the shared
subspace in the two-view JIVE model with Gaussian noise and establish its
statistical theory. In the regime $\gmin\gtrsim\sigma\sqrt n$ and with
comparable shared and view-specific ranks, BCSI attains the minimax rate and
identifies the distinct first- and second-order difficulties in estimating the
shared subspace. The theory also points to several unresolved questions, which
we discuss below.

\paragraph{What happens when the ranks are unbalanced?}
The upper and lower bounds match when
$r+(r_1\vee r_2)\asymp r+(r_1\wedge r_2)$,
but a gap remains when the two individual ranks are substantially different.
The upper bound contains terms that depend on the larger individual rank,
whereas the corresponding terms in the lower bound depend only on the smaller
individual rank. The principal-angle decomposition suggests that the rank
dependence in the upper bound may not be sharp: the blocks for the
$r_1\wedge r_2$ paired directions can have smallest eigenvalues of order
$\theta$, while each of the remaining $|r_1-r_2|$ unpaired directions has
eigenvalue $1$. The dependence on the larger rank may come from controlling the
leakage globally rather than separately over these two types of directions. An
upper-bound proof that controls the two types of principal-angle blocks
separately, or a lower-bound construction that captures the second-order
contribution of the unpaired directions, may sharpen the rank dependence.

\paragraph{How strong must the weaker view be?}
The condition $\gamma_{\min}\gtrsim\sigma\sqrt n$ is natural for BCSI because
its construction requires the empirical factor in each view to retain all
population signal directions. Thus the signal strength must be large enough
for both empirical factors to be reliable. The rate itself suggests that the
two views play different roles: the stronger view estimates the shared subspace
in $\mathbb R^n$, as reflected by the term involving $\gamma_{\max}$, whereas
the weaker view determines which directions within the stronger signal space
are shared, as reflected by the terms involving $\gamma_{\min}$. The weaker
view might identify the shared subspace through part of its signal even when it
cannot recover its entire signal space. The present parameter class does not
cover this possibility, but the rate still distinguishes between estimating the
shared subspace in $\mathbb R^n$ and separating the shared directions from the
individual directions. A full characterization below the current signal
threshold would be an interesting next step.

\paragraph{Beyond Gaussian noise.}
The present finite-sample analysis uses the Gaussian model in two essential
ways. First, the one-view leakage isometry relies on the precise behavior of
rectangular Gaussian spiked matrices. Second, rotational invariance allows the
leakage directions to be symmetrized by an independent Haar rotation, which
converts operator-norm bounds that depend on $n$ into bounds that depend on the
ranks for the linear leakage terms. The basic bias-correction principle may
extend beyond Gaussian noise, but the appropriate correction need not remain a
scalar multiple of the identity. Under heteroskedastic, correlated, or
anisotropic noise, different directions can have different population leakage
even after singular-value rescaling. In such settings, one may need to estimate
and subtract a non-scalar leakage operator.

\paragraph{More than two views.}
With multiple views, the oracle coefficient relations form a higher-dimensional
nullspace of a block Gram matrix, and the optimal reconstruction must account
for both the signal strengths of the views and the geometry of their individual
subspaces. Additional views may improve estimation of the shared subspace in
$\mathbb R^n$, but it is less clear how they change the error from separating
shared and individual directions or the second-order term caused by the
individual subspaces. Determining the minimax rate as a function of the number
of views, their strengths, and the alignment of their individual subspaces
would provide a more complete theory for heterogeneous multi-view integration.

\section*{Use of AI}
During the preparation of this manuscript, the authors used GPT~5.6~Sol to assist with paper editing, formatting, code writing, and proof writing. All its suggestions were reviewed and revised by the authors, who take full responsibility for the manuscript.
The authors also used GPT~5.6~Sol~to formalize the theoretical results in this paper using Lean 4.

\section*{Acknowledgement}
C.M.~was partially supported by the National Science Foundation via the CAREER
Award DMS-2443867.

\bibliographystyle{plainnat}
\bibliography{All-of-Bibs}

\newpage
\appendix

\section{Equivalence of plug-in subspace intersection with AJIVE}
\label{sec:AJIVE}

\begin{lemma}
\label{lem:ajive-equivalence}
Let
$\widehat{\bm\Phi}_k\in\operatorname{St}(n,p_k)$ for $k=1,2$, and set
$q\coloneqq p_1\wedge p_2$.  Write
\[
  \widehat{\bm\Phi}_1^\trans\widehat{\bm\Phi}_2
  =
  \bm O_1
  \diag(\zeta_1,\ldots,\zeta_q)
  \bm O_2^\trans,
  \qquad
  \zeta_1\ge\cdots\ge\zeta_q\ge0,
\]
where
$\bm O_1\in\operatorname{St}(p_1,q)$ and
$\bm O_2\in\operatorname{St}(p_2,q)$.
Suppose $\zeta_r>\zeta_{r+1}$, with $\zeta_{q+1}\coloneqq0$.

Let
$\widehat{\bm A}\coloneqq
[\widehat{\bm\Phi}_1,\widehat{\bm\Phi}_2]
\in\mathbb R^{n\times(p_1+p_2)}$,
and let
$\bm Q=(\bm Q_1^\trans,\bm Q_2^\trans)^\trans$
span its $r$ smallest right singular directions.  Then
\[
  \operatorname{col}
  \bigl(
    \widehat{\bm\Phi}_1\bm Q_1
    -
    \widehat{\bm\Phi}_2\bm Q_2
  \bigr)
\]
is exactly the leading $r$-dimensional eigenspace of 
$  \widehat{\bm\Phi}_1\widehat{\bm\Phi}_1^\trans
  +
  \widehat{\bm\Phi}_2\widehat{\bm\Phi}_2^\trans
$.
Hence the empirical subspace-intersection procedure used in the main text
returns the same subspace as two-view AJIVE.
\end{lemma}

\begin{proof}
We first identify the reconstruction from the $r$ smallest right singular
directions of $\widehat{\bm A}$, and then show that the resulting directions
are exactly the leading AJIVE directions.

Let $\bm o_{1,i}$ and $\bm o_{2,i}$ denote the $i$th columns of
$\bm O_1$ and $\bm O_2$, and set
$\widehat{\bm a}_i\coloneqq\widehat{\bm\Phi}_1 o_{1,i}$ and
$\widehat{\bm b}_i\coloneqq\widehat{\bm\Phi}_2\bm o_{2,i}$.
Then both vectors have unit norm and
$\widehat{\bm a}_i^\trans\widehat{\bm b}_i=\zeta_i$.

The Gram matrix of $\widehat{\bm A}$ is
\[
  \widehat{\bm A}^\trans\widehat{\bm A}
  =
  \begin{pmatrix}
    \bm I_{p_1}
    &
    \widehat{\bm\Phi}_1^\trans\widehat{\bm\Phi}_2
    \\
    \widehat{\bm\Phi}_2^\trans\widehat{\bm\Phi}_1
    &
    \bm I_{p_2}
  \end{pmatrix},
\]
and, for every $i\le q$,
\[
  \widehat{\bm A}^\trans\widehat{\bm A}
  \frac{1}{\sqrt2}
  \begin{pmatrix}
    \bm o_{1,i}\\
    -\bm o_{2,i}
  \end{pmatrix}
  =
  (1-\zeta_i)
  \frac{1}{\sqrt2}
  \begin{pmatrix}
    \bm o_{1,i}\\
    -\bm o_{2,i}
  \end{pmatrix}.
\]
Since $\zeta_1\ge\cdots\ge\zeta_q$ and
$\zeta_r>\zeta_{r+1}$, these vectors span the $r$ smallest right singular
directions of $\widehat{\bm A}$.  Since only the reconstructed column space
matters, we may therefore take
$\bm Q_1=2^{-1/2}[\bm o_{1,1},\ldots,\bm o_{1,r}]$ and
$\bm Q_2=-2^{-1/2}[\bm o_{2,1},\ldots,\bm o_{2,r}]$, which gives
\[
  \operatorname{col}
  \bigl(
    \widehat{\bm\Phi}_1\bm Q_1
    -
    \widehat{\bm\Phi}_2\bm Q_2
  \bigr)
  =
  \operatorname{span}
  \left\{
    \widehat{\bm a}_i+\widehat{\bm b}_i:
    i=1,\ldots,r
  \right\}.
\]

It remains to identify this span with the leading AJIVE eigenspace.  Since
$\widehat{\bm\Phi}_1^\trans\widehat{\bm b}_i
=\zeta_i\bm o_{1,i}$ and
$\widehat{\bm\Phi}_2^\trans\widehat{\bm a}_i
=\zeta_i\bm o_{2,i}$,
\[
  \left(
    \widehat{\bm\Phi}_1\widehat{\bm\Phi}_1^\trans
    +
    \widehat{\bm\Phi}_2\widehat{\bm\Phi}_2^\trans
  \right)
  (\widehat{\bm a}_i\pm\widehat{\bm b}_i)
  =
  (1\pm\zeta_i)
  (\widehat{\bm a}_i\pm\widehat{\bm b}_i).
\]
Any unpaired directions have eigenvalue $1$.  Hence
$\zeta_1\ge\cdots\ge\zeta_q$ and
$\zeta_r>\zeta_{r+1}$ imply that the leading $r$-dimensional eigenspace is
precisely
$\operatorname{span}\{
\widehat{\bm a}_i+\widehat{\bm b}_i:i=1,\ldots,r\}$.
Thus the two procedures return the same $r$-dimensional subspace.
\end{proof}

\section{Deterministic ingredients for the upper bound}
The deterministic argument proceeds in four stages. We first use the
population geometry to establish curvature of the oracle Gram matrix and to
construct the transfer operator that describes how individual signal
directions enter the reconstructed shared subspace. We then apply a relative
perturbation argument to isolate the $r$-dimensional near-zero eigenspace
without imposing an upper bound on the signal strengths. Since an orthonormal
basis of this eigenspace is determined only up to a right orthogonal factor,
the calibration step aligns it with the oracle normalization. Finally, the
calibrated eigenvalue equation separates the subspace error into its linear
leakage contribution and four higher-order remainder terms.
\subsection{Population and oracle geometry}

\subsubsection{Geometry of the joint signal spaces}
\label{app:joint-subspace-geometry}

The oracle bounds below require a lower bound on the nonzero spectrum of
$\bm P_{\cM_1}+\bm P_{\cM_2}$ and an upper bound on the trace of its
pseudoinverse. Both follow from the principal-angle decomposition of the two
individual signal spaces. Let
$\chi_j=\sigma_j(\bm V_1^\trans\bm V_2)$ for
$j=1,\ldots,r_1\wedge r_2$. These singular values satisfy
$1\ge\chi_j\ge0$, and \cref{eq:misalign} gives
$\chi_1\le1-2\theta$.

\begin{lemma}
\label{lem:joint-subspace-geometry}
The operator
$\bm P_{\cM_1}+\bm P_{\cM_2}$ has support
$\cM_1+\cM_2$ and rank $r + r_1 + r_2$.
Its nonzero eigenvalues are
\begin{equation}
\label{eq:projector-sum-spectrum}
  \underbrace{2,\ldots,2}_{r\text{ times}},
  \qquad
  \underbrace{1+\chi_j,\ 1-\chi_j}_{j =1, \ldots, r_1 \wedge r_2},
  \qquad
  \underbrace{1,\ldots,1}_{|r_1-r_2|\text{ times}}.
\end{equation}
Consequently,
\begin{equation}
\label{eq:projector-sum-floor}
  \lmin\bigl(
    \bm P_{\cM_1}+\bm P_{\cM_2}
  \bigr)
  \ge 2\theta,
\end{equation}
and
\begin{equation}
\label{eq:projector-sum-pseudoinverse-trace}
\begin{aligned}
  \tr\left\{
    \bigl(
      \bm P_{\cM_1}+\bm P_{\cM_2}
    \bigr)^\dagger
  \right\}
  \lesssim
  r+r_1+r_2+\frac{r_1 \wedge r_2}{\theta}.
\end{aligned}
\end{equation}
\end{lemma}

\begin{proof}
Define $q \coloneqq r_1 \wedge r_2$.  Choose
orthonormal principal-vector pairs
$\widetilde{\bm v}_{1,1},\ldots,\widetilde{\bm v}_{1,q}
\in\col(\bm V_1)$ and
$\widetilde{\bm v}_{2,1},\ldots,\widetilde{\bm v}_{2,q}
\in\col(\bm V_2)$ such that
\[
  \widetilde{\bm v}_{1,i}^\trans
  \widetilde{\bm v}_{2,j}
  =
  \chi_j\ind\{i=j\}.
\]
These vectors are obtained from the left and right singular vectors of
$\bm V_1^\trans\bm V_2$ and may be completed to orthonormal bases of the two
individual subspaces.

Since
$\cM_k=\cU\oplus\col(\bm V_k)$,
\[
  \bm P_{\cM_1}+\bm P_{\cM_2}
  =
  2\bm P_{\cU}
  +
  \bm V_1\bm V_1^\trans
  +
  \bm V_2\bm V_2^\trans.
\]
The operator therefore has eigenvalue $2$ on $\cU$.  For each $j\le q$, the
principal plane
$\operatorname{span}\{
\widetilde{\bm v}_{1,j},
\widetilde{\bm v}_{2,j}
\}$
is invariant, and
\[
  \bigl(
    \bm P_{\cM_1}+\bm P_{\cM_2}
  \bigr)
  \bigl(
    \widetilde{\bm v}_{1,j}
    \pm
    \widetilde{\bm v}_{2,j}
  \bigr)
  =
  (1\pm\chi_j)
  \bigl(
    \widetilde{\bm v}_{1,j}
    \pm
    \widetilde{\bm v}_{2,j}
  \bigr).
\]
If $r_1\ne r_2$, the remaining $|r_1-r_2|$ basis vectors in the larger
individual subspace are orthogonal to the other individual subspace, so each
has eigenvalue $1$.  All the displayed eigenvalues are positive because
$\chi_j\le1-2\theta<1$ for every $j$, while the operator vanishes on
$(\cM_1+\cM_2)^\perp$.  This proves \cref{eq:projector-sum-spectrum} and shows
that the support is $\cM_1+\cM_2$.  Its rank is therefore
$r+r_1+r_2=p_1+p_2-r$.

By \cref{eq:misalign},
$\chi_j\le1-2\theta$ for every $j$.  Hence
$1-\chi_j\ge2\theta$, while $1\ge2\theta$ and $2\ge2\theta$ because
$\theta\le1/2$.  This proves \cref{eq:projector-sum-floor}.

Finally, summing the reciprocals of the eigenvalues in
\cref{eq:projector-sum-spectrum} gives
\[
\begin{aligned}
  \tr\left\{
    \bigl(\bm P_{\cM_1}+\bm P_{\cM_2}\bigr)^\dagger
  \right\}
  &=
  \frac r2+|r_1-r_2|
  +\sum_{j=1}^{q}
  \left(
    \frac{1}{1+\chi_j}
    +
    \frac{1}{1-\chi_j}
  \right) \\
  &\le
  \frac r2+|r_1-r_2|+q+\frac{q}{2\theta}
  \lesssim
  r+r_1+r_2+\frac{r_1\wedge r_2}{\theta}.
\end{aligned}
\]
Here we used $1+\chi_j\ge1$, $1-\chi_j\ge2\theta$, and
$q=r_1\wedge r_2$. This proves
\cref{eq:projector-sum-pseudoinverse-trace}.
\end{proof}

\subsubsection{Geometry of the oracle Gram matrix}
\label{app:oracle-geometry}

\begin{proof}[Proof of \cref{lem:oracle-geometry}]

The proof has three parts. We first identify the nullspace of the oracle Gram
matrix, then verify the normalization identities for $\bm Q_\star$, and
finally establish the curvature bound on the orthogonal complement of the
nullspace.

\paragraph{Nullspace of the oracle Gram.}
Since $\bm G_0=\bm S^\trans\bm S$, we have
$\bm z^\trans\bm G_0\bm z=\|\bm S\bm z\|_2^2$, and hence
$\ker(\bm G_0)=\ker(\bm S)$.
On $\cE_{\rm emb}$, $\col(\bm S)=\cM_1+\cM_2$. Since
$\cM_1\cap\cM_2=\cU$, we have
$\rank(\bm S)=\dim(\cM_1+\cM_2)=p_1+p_2-r$.
As $\bm S$ has $p_1+p_2$ columns, rank--nullity gives
$\dim\ker(\bm G_0)=\dim\ker(\bm S)=r$.

\paragraph{Proof of~\cref{eq:Q_star}.}
Next consider $\bm Q_\star$.  Since
$\cU\subseteq\col(\bm S_k)$, we have
$\bm S_k\bm S_k^\dagger\bm U=\bm U$ for $k=1,2$.  Therefore
$\bm S\bm Q_\star=\bm 0$ and
$\bm S\bm J\bm Q_\star=(\alpha_1+\alpha_2)\bm U=\bm U$, proving
\cref{eq:Q_star}.

\paragraph{Proof of~\cref{eq:oracle-curvature}.}
It remains to lower-bound the positive spectrum of $\bm G_0$.
On $\cE_{\rm emb}$,
$\col(\bm S_k)=\cM_k$ and
$\sigma_{\min}(\bm S_k)\ge\cemb\gamma_k\ge\cemb\gmin$.
Thus, for every $\bm x\in\R^n$,
\[
  \bm x^\trans\bm S_k\bm S_k^\trans\bm x
  =
  \|\bm S_k^\trans\bm P_{\cM_k}\bm x\|_2^2
  \ge
  (\cemb\gmin)^2
  \|\bm P_{\cM_k}\bm x\|_2^2,
\]
and hence
\begin{equation}
\label{eq:SS-projector-lower}
  \bm S\bm S^\trans
  \succeq
  (\cemb\gmin)^2
  \bigl(
    \bm P_{\cM_1}+\bm P_{\cM_2}
  \bigr).
\end{equation}
Since $\bm G_0=\bm S^\trans\bm S$ and
$\bm S\bm S^\trans$ have the same nonzero eigenvalues,
\cref{eq:SS-projector-lower,eq:projector-sum-floor} yield
\[
  \lmin(\bm G_0)
  =
  \lmin(\bm S\bm S^\trans)
  \ge
  2\cemb^2\gmin^2\theta,
\]
which proves \cref{eq:oracle-curvature}.

The nullspace calculation, \cref{eq:Q_star}, and the curvature bound therefore
establish all three claims in \cref{lem:oracle-geometry}. This completes the
proof.
\end{proof}

\subsubsection{Population norm bounds}
\label{app:population-norm-bounds}

We collect several consequences of the oracle geometry that will be used
repeatedly in the probabilistic analysis.  Write
\[
  p_{\min}\coloneqq p_1\wedge p_2
  =r+(r_1\wedge r_2),
  \qquad
  p_{\max}\coloneqq p_1\vee p_2
  =r+(r_1\vee r_2).
\]
\begin{lemma}
\label{lem:population-norm-bounds}
For $k=1,2$, define $\bm R_k\coloneqq\bm S_k^\dagger\bm U\in\mathbb R^{p_k\times r}$. On $\cE_{\rm emb}$, the following bounds hold.
\begin{subequations}
\label{eq:population-geometry-bounds}
\begin{align}
\label{eq:population-Rk-bounds}
  \opn{\bm R_k}
  &\lesssim
  \frac{1}{\gamma_k},
  &
  \frn{\bm R_k}
  &\lesssim
  \frac{\sqrt r}{\gamma_k},
  \qquad k=1,2,
\\
\label{eq:population-SGamma-bounds}
  \opn{\bm S(\bm G_0^\dagger)^{1/2}}
  &\le 1,
  &
  \frn{\bm S(\bm G_0^\dagger)^{1/2}}
  &=
  \sqrt{p_1+p_2-r}
  \lesssim
  \sqrt{p_{\max}},
\\
\label{eq:population-Gamma-bounds}
  \opn{(\bm G_0^\dagger)^{1/2}}
  &\lesssim
  \frac{1}{\gmin\sqrt\theta},
  &
  \frn{(\bm G_0^\dagger)^{1/2}}
  &\lesssim
  \frac{1}{\gmin}
  \left(
    \sqrt{r_1+r_2}
    +
    \frac{\sqrt{p_{\min}}}{\sqrt\theta}
  \right).
\end{align}
\end{subequations}
\end{lemma}

\begin{proof}
For \cref{eq:population-Rk-bounds}, on $\cE_{\rm emb}$ we have
$\col(\bm S_k)=\cM_k$, $\cU\subseteq\cM_k$, and
$\sigma_{\min}(\bm S_k)\ge\cemb\gamma_k$.  Hence
$\|\bm R_k\|_{\mathrm{op}}
\le\|\bm S_k^\dagger\|_{\mathrm{op}}
\lesssim\gamma_k^{-1}$, and the Frobenius bound follows from
$\rank(\bm R_k)\le r$.

For~\cref{eq:population-SGamma-bounds}, observe that
\[
  \bigl(\bm S(\bm G_0^\dagger)^{1/2}\bigr)^\trans\bm S(\bm G_0^\dagger)^{1/2}
  =
  (\bm G_0^\dagger)^{1/2}\bm G_0(\bm G_0^\dagger)^{1/2}
  =
  \bm I-\bm P_{\ker(\bm G_0)}.
\]
By \cref{lem:oracle-geometry}, $\rank(\bm P_{\ker(\bm G_0)})=r$, and therefore
$\rank(\bm I-\bm P_{\ker(\bm G_0)})=p_1+p_2-r$.  This proves
\cref{eq:population-SGamma-bounds}.

For $(\bm G_0^\dagger)^{1/2}$ in~\cref{eq:population-Gamma-bounds}, we have $\opn{(\bm G_0^\dagger)^{1/2}} = 1/ \sqrt{\lambda_{\min}^+(\bm G_0)}$. 
This together with \cref{lem:oracle-geometry} yields the operator norm bound. 
To control the Frobenius norm, observe that, on $\cE_{\rm emb}$, both
$\bm S\bm S^\trans$ and
$\bm P_{\cM_1}+\bm P_{\cM_2}$ have support
$\cM_1+\cM_2$. We may therefore restrict
\cref{eq:SS-projector-lower} to this common support, where both operators are
positive definite, and invert the inequality. Extending the inverses by zero
on $(\cM_1+\cM_2)^\perp$ gives
\[
  (\bm S\bm S^\trans)^\dagger
  \preceq
  \frac{1}{\cemb^2\gmin^2}
  \bigl(\bm P_{\cM_1}+\bm P_{\cM_2}\bigr)^\dagger.
\]
Since $\bm G_0=\bm S^\trans\bm S$ and
$\bm S\bm S^\trans$ have the same nonzero eigenvalues,
\cref{eq:projector-sum-pseudoinverse-trace} now yields
\[
\begin{aligned}
  \frn{(\bm G_0^\dagger)^{1/2}}^2
  &=
  \tr(\bm G_0^\dagger)
  =
  \tr\{(\bm S\bm S^\trans)^\dagger\} \\
  &\lesssim
  \frac{1}{\gmin^2}
  \left(
    r+r_1+r_2+\frac{r_1\wedge r_2}{\theta}
  \right)
  \lesssim
  \frac{1}{\gmin^2}
  \left(
    r_1+r_2+\frac{p_{\min}}{\theta}
  \right).
\end{aligned}
\]
The last inequality uses
$p_{\min}=r+(r_1\wedge r_2)$ and $\theta\le1/2$. Taking square roots and using
$\sqrt{x+y}\le\sqrt x+\sqrt y$ proves the Frobenius bound in
\cref{eq:population-Gamma-bounds}.
\end{proof}

\subsection{Relative perturbation of the nullspace}
\label{app:nullspace-perturbation}

\begin{proof}[Proof of \cref{lem:nullspace}]
The proof has three parts.  We first obtain a lower bound for the entire
spectrum of $\bm G$.  We then use the min--max principle to separate the
$r$ eigenvalues near zero from the remaining spectrum.  Finally, we project
the invariant-subspace equation onto $\ker(\bm G_0)^\perp$ to control the
relative residual.

Choose orthonormal bases
$\bm Q_0\in\St(d,r)$ and
$\bm Q_+\in\St(d,d-r)$
for $\ker(\bm G_0)$ and $\ker(\bm G_0)^\perp$, respectively, and define the
positive spectral block
$\bm K\coloneqq\bm Q_+^\trans\bm G_0\bm Q_+$.
Then $\bm K\succ\bze$ and
$\lambda_{\min}(\bm K)=\lambda_{\min}^{+}(\bm G_0)$.
We identify the relative perturbation operators with their nonzero coordinate
blocks,
\[
  \bm\Delta_{00}
  =
  \bm Q_0^\trans\bm\Delta\bm Q_0,
  \qquad
  \bm\Delta_{+0}
  =
  \bm K^{-1/2}\bm Q_+^\trans\bm\Delta\bm Q_0,
  \qquad
  \bm\Delta_{++}
  =
  \bm K^{-1/2}
  \bm Q_+^\trans\bm\Delta\bm Q_+
  \bm K^{-1/2}.
\]
This identification preserves their operator norms.  The assumptions of
\cref{lem:nullspace} imply
\begin{equation}
\label{eq:null-smallness}
  \opn{\bm\Delta_{++}}\le\frac14,
  \qquad
  \varepsilon_\Delta
  \le\frac{1}{16}\lambda_{\min}^{+}(\bm G_0).
\end{equation}

\paragraph{Lower bound for the spectrum.}
Write any $\bm w\in\R^d$ uniquely as
$\bm w=\bm Q_0\bm u+\bm Q_+\bm v$, and set
$\bm t\coloneqq\bm K^{1/2}\bm v$.  In these coordinates,
\begin{equation}
\label{eq:null-quadratic-form}
  \bm w^\trans\bm G\bm w
  =
  \|\bm t\|_2^2
  +
  \bm t^\trans\bm\Delta_{++}\bm t
  +
  2\bm t^\trans\bm\Delta_{+0}\bm u
  +
  \bm u^\trans\bm\Delta_{00}\bm u.
\end{equation}
By \cref{eq:null-smallness}, the first two terms are at least
$\frac12\|\bm t\|_2^2$.  Young's inequality gives
$2|\bm t^\trans\bm\Delta_{+0}\bm u|
\le\frac12\|\bm t\|_2^2+
2\|\bm\Delta_{+0}\|_{\mathrm{op}}^2\|\bm u\|_2^2$.
Consequently,
\[
  \bm w^\trans\bm G\bm w
  \ge
  -\left(
    \opn{\bm\Delta_{00}}
    +
    2\opn{\bm\Delta_{+0}}^2
  \right)
  \|\bm u\|_2^2
  \ge
  -\varepsilon_\Delta\|\bm w\|_2^2.
\]
It follows that $\lambda_{\min}(\bm G)\ge-\varepsilon_\Delta$.

\paragraph{Separation of the near-zero spectrum.}
Order the eigenvalues as
$\lambda_1(\bm G)\ge\cdots\ge\lambda_d(\bm G)$.
On $\ker(\bm G_0)=\col(\bm Q_0)$, the quadratic form of $\bm G$ is
$\bm u^\trans\bm\Delta_{00}\bm u$.  Courant--Fischer therefore gives
\[
  \lambda_{d-r+1}(\bm G)
  \le
  \opn{\bm\Delta_{00}}
  \le
  \varepsilon_\Delta.
\]
Thus at least $r$ eigenvalues are at most $\varepsilon_\Delta$.

To control the remaining spectrum, apply Courant--Fischer on
$\ker(\bm G_0)^\perp=\col(\bm Q_+)$.  In this case $\bm u=\bze$, and
\cref{eq:null-quadratic-form} gives
\[
  \bm w^\trans\bm G\bm w
  \ge
  \lambda_{\min}^{+}(\bm G_0)
  \bigl(
    1-\opn{\bm\Delta_{++}}
  \bigr)
  \|\bm w\|_2^2.
\]
By \cref{eq:null-smallness}, the coefficient on the right is at least
$12\varepsilon_\Delta$, which is strictly larger than
$4\varepsilon_\Delta$.  Hence
$\lambda_{d-r}(\bm G)>4\varepsilon_\Delta$.
Combining this separation with the preceding lower bound shows that exactly
$r$ eigenvalues lie in $[-\varepsilon_\Delta,\varepsilon_\Delta]$, while all remaining eigenvalues
exceed $4\varepsilon_\Delta$.

\paragraph{Bound for $\bm G_0^{1/2}\bm Q$.}
Let $\bm Q\in\St(d,r)$ span the invariant subspace associated with
the $r$ eigenvalues in $[-\varepsilon_\Delta,\varepsilon_\Delta]$, and set
$\bm\Lambda\coloneqq\bm Q^\trans\bm G\bm Q$.
Then $\bm G\bm Q=\bm Q\bm\Lambda$ and
$\|\bm\Lambda\|_{\mathrm{op}}\le \varepsilon_\Delta$.

To express the desired residual in the positive spectral coordinates, define
$\bm H\coloneqq\bm K^{1/2}\bm Q_+^\trans\bm Q$.
Projecting the invariant-subspace equation onto $\col(\bm Q_+)$ and
multiplying by $\bm K^{-1/2}$ gives
\begin{equation}
\label{eq:nullspace-H-equation}
  \bm H
  =
  -\bm\Delta_{+0}\bm Q_0^\trans\bm Q
  -\bm\Delta_{++}\bm H
  +\bm K^{-1}\bm H\bm\Lambda.
\end{equation}
Taking operator norms, using
$\|\bm Q_0^\trans\bm Q\|_{\mathrm{op}}\le1$, and recalling that
$\|\bm K^{-1}\|_{\mathrm{op}}
=1/\lambda_{\min}^{+}(\bm G_0)$, we obtain
\[
  \opn{\bm H}
  \le
  \opn{\bm\Delta_{+0}}
  +
  \left(
    \opn{\bm\Delta_{++}}
    +
    \frac{\opn{\bm\Lambda}}
         {\lambda_{\min}^{+}(\bm G_0)}
  \right)
  \opn{\bm H}.
\]
Since $\|\bm\Lambda\|_{\mathrm{op}}\le \varepsilon_\Delta$, the factor multiplying
$\|\bm H\|_{\mathrm{op}}$ is at most $1/2$ by
\cref{eq:null-smallness}.  Therefore
$\|\bm H\|_{\mathrm{op}}
\le2\|\bm\Delta_{+0}\|_{\mathrm{op}}$.
Finally,
\[
  \opn{\bm G_0^{1/2}\bm Q}
  =
  \opn{\bm K^{1/2}\bm Q_+^\trans\bm Q}
  =
  \opn{\bm H}
  \le
  2\opn{\bm\Delta_{+0}}.
\]
This is exactly the bound in
\cref{eq:abstract-relative-nullspace-conclusion}, so the result follows.
\end{proof}

\subsection{Calibration and deterministic error decomposition}

\subsubsection{Calibration}

The relative perturbation argument controls the residual
$\bm S\bm Q$. To calibrate $\bm Q$, we first show that this residual
controls the coefficient error once the shared component has been fixed.

\begin{lemma}
\label{lem:coefficient-stability}
On $\cE_{\rm emb}$, let $\bm Z$ be any matrix satisfying
$\bm U^\trans\bm S\bm J\bm Z=\bze$. Then
\begin{equation}
\label{eq:coefficient-stability}
  \opn{\bm Z}
  \le
  \frac{\opn{\bm S\bm Z}}
       {\cemb\gmin\sqrt{2\theta}}.
\end{equation}
\end{lemma}

\begin{proof}
It suffices to prove the corresponding vector inequality. Let
$\bm z=(\bm z_1^\trans,\bm z_2^\trans)^\trans$ satisfy
$\bm U^\trans\bm S\bm J\bm z=\bze$. Since
$\col(\bm S_k)=\cM_k$ on $\cE_{\rm emb}$, write
$\bm S_k\bm z_k=\bm U\bm a_k+\bm V_k\bm b_k$. The constraint implies
$\alpha_1\bm a_1=\alpha_2\bm a_2$, so there is a vector $\bm a$ such that
$\bm a_1=\alpha_2\bm a$ and $\bm a_2=\alpha_1\bm a$. Consequently,
$\|\bm a_1+\bm a_2\|_2^2=\|\bm a\|_2^2$ and, since
$\alpha_1^2+\alpha_2^2\le 1$ and $2\theta\le 1$,
\[
  \|\bm a_1+\bm a_2\|_2^2
  \ge
  2\theta
  \bigl(
    \|\bm a_1\|_2^2+\|\bm a_2\|_2^2
  \bigr).
\]
Restricting \cref{eq:projector-sum-floor} to $\cU^\perp$ gives the same
coercive principal-angle bound on the two individual signal spaces. Applied
to the coordinate pair $(\bm b_1,\bm b_2)$, it yields
$\|\bm V_1\bm b_1+\bm V_2\bm b_2\|_2^2
\ge2\theta(\|\bm b_1\|_2^2+\|\bm b_2\|_2^2)$.

Because $\bm U$ is orthogonal to $\bm V_1$ and $\bm V_2$, it follows that
\begin{align*}
  \|\bm S\bm z\|_2^2
  &=
  \|\bm a_1+\bm a_2\|_2^2
  +
  \|\bm V_1\bm b_1+\bm V_2\bm b_2\|_2^2 \\
  &\ge
  2\theta
  \sum_{k=1}^2
  \bigl(
    \|\bm a_k\|_2^2+\|\bm b_k\|_2^2
  \bigr) \\
  &=
  2\theta
  \sum_{k=1}^2
  \|\bm S_k\bm z_k\|_2^2
  \ge
  2\theta(\cemb\gmin)^2\|\bm z\|_2^2.
\end{align*}
Thus
$\|\bm z\|_2
\le
\|\bm S\bm z\|_2/(\cemb\gmin\sqrt{2\theta})$.
Applying this inequality to $\bm z=\bm Z\bm x$ and taking the supremum
over $\|\bm x\|_2=1$ proves \cref{eq:coefficient-stability}.
\end{proof}

We now use this stability estimate to align the empirical near-zero
eigenspace with the oracle normalization. Define
\[
  \bm\Xi
  \coloneqq \bm U^\trans \bm M = 
  \bm U^\trans\bm S\bm J\bm Q.
\]
This matrix is the change of basis required to impose the normalization
$\bm S\bm J\bm Q_\star=\bm U$. Whenever $\bm\Xi$ is invertible, define
\[
  \widetilde{\bm Q}
  \coloneqq
  \bm Q\bm\Xi^{-1},
  \qquad
  \bm\Delta_{\bm Q}
  \coloneqq
  \widetilde{\bm Q}-\bm Q_\star.
\]

\begin{lemma}
\label{lem:calibration}
On $\cE_{\rm good}$, the matrix $\bm\Xi$ is invertible, and
\[
  \opn{\bm\Xi^{-1}}
  \lesssim
  \frac{1}{\gmin},
  \qquad
  \opn{\bm S\bm\Delta_{\bm Q}}
  \lesssim
  \frac{\opn{\bm\Delta_{+0}}}{\gmin},
  \qquad
  \opn{\bm\Delta_{\bm Q}}
  \lesssim
  \frac{\opn{\bm\Delta_{+0}}}
       {\gmin^2\sqrt\theta}.
\]
\end{lemma}

\begin{proof}
Set
$\bm Z\coloneqq\bm Q-\bm Q_\star\bm\Xi$. The identities
$\bm S\bm Q_\star=\bze$ and
$\bm S\bm J\bm Q_\star=\bm U$ imply
$\bm U^\trans\bm S\bm J\bm Z=\bze$ and
$\bm S\bm Z=\bm S\bm Q$. Hence
\cref{eq:coefficient-stability,eq:relative-nullspace-conclusion} and the
definition of $\cE_{\rm good}$ give
\[
  \opn{\bm Z}
  \le
  \frac{2\opn{\bm\Delta_{+0}}}
       {\cemb\gmin\sqrt{2\theta}}
  \le
  \frac12.
\]

This bound implies that $\bm\Xi$ is invertible. Indeed, if
$\bm\Xi\bm x=\bze$, then $\bm Q\bm x=\bm Z\bm x$. Since $\bm Q$ has
orthonormal columns, and $\opn{\bm Z} \leq \frac12$, we have
$\|\bm x\|_2\le\frac12\|\bm x\|_2$, and therefore $\bm x=\bze$.

The identity
$\bm Q=\bm Q_\star\bm\Xi+\bm Z$ now gives
$\bm Q\bm\Xi^{-1}
=\bm Q_\star+\bm Z\bm\Xi^{-1}$. Since $\bm Q$ is an isometry,
\[
  \opn{\bm\Xi^{-1}}
  =
  \opn{\bm Q\bm\Xi^{-1}}
  \le
  \opn{\bm Q_\star}
  +
  \frac12\opn{\bm\Xi^{-1}}.
\]
By the block definition of $\bm Q_\star$ and
\cref{eq:population-Rk-bounds},
\[
  \opn{\bm Q_\star}^2
  =
  \opn{\bm R_1^\trans\bm R_1+\bm R_2^\trans\bm R_2}
  \le
  \opn{\bm R_1}^2+\opn{\bm R_2}^2
  \lesssim
  \frac{1}{\gmin^2}.
\]
Thus
$\opn{\bm\Xi^{-1}}\le2\opn{\bm Q_\star}\lesssim\gmin^{-1}$.

Finally, $\bm\Delta_{\bm Q}=\bm Z\bm\Xi^{-1}$. Since
$\bm S\bm Z=\bm S\bm Q$, the relative nullspace bound yields
\[
  \opn{\bm S\bm\Delta_{\bm Q}}
  \le
  \opn{\bm S\bm Q}\opn{\bm\Xi^{-1}}
  \lesssim
  \frac{\opn{\bm\Delta_{+0}}}{\gmin}.
\]
The first display in this proof already gives
$\opn{\bm Z}\lesssim
\opn{\bm\Delta_{+0}}/(\gmin\sqrt\theta)$.
Combining this estimate with
$\bm\Delta_{\bm Q}=\bm Z\bm\Xi^{-1}$ and
$\opn{\bm\Xi^{-1}}\lesssim\gmin^{-1}$ proves the final bound in
\cref{lem:calibration}. This completes the proof.
\end{proof}

\subsubsection{The transfer operator}
\label{subsec:transfer}

Recall the transfer operator defined in \cref{eq:transfer-operator}.  The next
lemma records the identity that connects this population operator to the
coefficient reconstruction, together with the two norm bounds needed in the
probabilistic analysis.

\begin{lemma}
\label{lem:transfer}
For every coefficient vector $\bm z$,
\begin{equation}
\label{eq:transfer-identity}
  \bm P_{\cU^\perp}\bm S\bm J\bm z
  =
  \bm T(\bm S\bm z).
\end{equation}
Moreover,
\begin{equation}
\label{eq:transfer-op}
  \opn{\bm T}
  \le
  \frac{1}{\sqrt{2\theta}},
  \qquad
  \frn{\bm T}
  \lesssim
  \sqrt{r_1+r_2}
  +
  \frac{\sqrt{r_1\wedge r_2}}{\sqrt\theta}.
\end{equation}
\end{lemma}

\begin{proof}
We first verify the identity from the coordinate definition of $\bm T$.  We
then factor $\bm T$ through the two individual signal spaces and use their
principal-angle spectrum to obtain the norm bounds.

Write $\bm z=(\bm z_1^\trans,\bm z_2^\trans)^\trans$ and decompose
$\bm S_k\bm z_k=\bm U\bm a_k+\bm V_k\bm b_k$, where
$\bm a_k\in\R^r$ and $\bm b_k\in\R^{r_k}$.  Then
$\bm S\bm z=\bm U(\bm a_1+\bm a_2)+\bm V_1\bm b_1+\bm V_2\bm b_2$, whereas
\[
  \bm P_{\cU^\perp}\bm S\bm J\bm z
  =
  \alpha_1\bm V_1\bm b_1
  -
  \alpha_2\bm V_2\bm b_2
  =
  \bm T(\bm S\bm z).
\]
The last equality is precisely the definition of $\bm T$, and proves
\cref{eq:transfer-identity}.

For the norm bounds, set
$\bm V\coloneqq[\bm V_1,\bm V_2]$ and
$\bm V_\alpha\coloneqq[\alpha_1\bm V_1,-\alpha_2\bm V_2]$.
The spectrum in \cref{lem:joint-subspace-geometry} shows that $\bm V$ has
full column rank. Thus $\bm V^\dagger$ recovers the unique individual-space
coordinates of every vector in
$\col(\bm V_1)+\col(\bm V_2)$ and vanishes on its orthogonal complement,
which contains both $\cU$ and $(\cM_1+\cM_2)^\perp$. Comparing this action
with \cref{eq:transfer-operator} shows that
$\bm T=\bm V_\alpha\bm V^\dagger$ on all of $\mathbb R^n$.

For $\bm b=(\bm b_1^\trans,\bm b_2^\trans)^\trans$, the triangle and
Cauchy--Schwarz inequalities give
$\|\bm V_\alpha\bm b\|_2
\le\alpha_1\|\bm b_1\|_2+\alpha_2\|\bm b_2\|_2
\le\|\bm b\|_2$.
Hence $\opn{\bm V_\alpha}\le1$.

The same spectrum, together with
$\chi_j\le1-2\theta$, shows that
$\lambda_{\min}(\bm V^\trans\bm V)\ge2\theta$.  Therefore
$\|\bm V^\dagger\|_{\mathrm{op}}\le(2\theta)^{-1/2}$, and hence
$\|\bm T\|_{\mathrm{op}}\le(2\theta)^{-1/2}$.

For the Frobenius norm, the complete spectrum of
$\bm V^\trans\bm V$ gives
\[
  \frn{\bm T}^2
  \le
  \frn{\bm V^\dagger}^2
  =
  |r_1-r_2|
  +
  \sum_{j=1}^{r_1 \wedge r_2}
  \left(
    \frac{1}{1+\chi_j}
    +
    \frac{1}{1-\chi_j}
  \right)
  \lesssim
  r_1+r_2+\frac{r_1 \wedge r_2}{\theta}.
\]
Indeed, each summand is
$2/(1-\chi_j^2)\le1/\theta$.  Taking square roots and using
$\sqrt{x+y}\le\sqrt x+\sqrt y$ proves \cref{eq:transfer-op}.
\end{proof}

\subsubsection{Deterministic error decomposition}
\label{app:error-decomposition}

The decomposition isolates the part of the subspace error that is linear in
the leakage $\bm L$. The direct contribution is already visible in the
reconstruction formula. The remaining linear contribution comes from the
displacement of the near-zero eigenspace and is contained in
$\bm T\bm S\bm\Delta_{\bm Q}$. We extract this term using the eigenvalue
equation for $\bm Q$ and then bound the resulting higher-order remainder.

\begin{proof}[Proof of \cref{lem:error-decomposition}]
Write
$\bm\Delta_{\rm quad}\coloneqq\bm L^\trans\bm L-\bm D_\sigma$ and
$\bm\Lambda\coloneqq\bm Q^\trans\bm G\bm Q$. Since $\bm Q$ has
orthonormal columns spanning an invariant subspace of $\bm G$, we have
$\bm G\bm Q=\bm Q\bm\Lambda$. On $\cE_{\rm good}$,
\cref{lem:calibration} gives
$\widetilde{\bm Q}
=\bm Q\bm\Xi^{-1}
=\bm Q_\star+\bm\Delta_{\bm Q}$.

Using the reconstruction identity preceding
\cref{lem:error-decomposition} and
$\bm S\bm Q_\star=\bze$, we obtain
\begin{equation}
\label{eq:error-proof-start}
  \bm P_{\cU^\perp}\widetilde{\bm M}
  =
  \bm T\bm S\bm\Delta_{\bm Q}
  +
  \bm L\bm J\bm Q_\star
  +
  \bm L\bm J\bm\Delta_{\bm Q}.
\end{equation}
The second term is linear in $\bm L$, whereas the third is higher order.
It remains to identify the linear part of
$\bm T\bm S\bm\Delta_{\bm Q}$.

The relations
$\bm G\bm Q=\bm Q\bm\Lambda$ and
$\widetilde{\bm Q}=\bm Q\bm\Xi^{-1}$ imply
$\bm G\widetilde{\bm Q}
=\bm Q\bm\Lambda\bm\Xi^{-1}$.
Substituting
$\bm G=\bm G_0+\bm\Delta$ and
$\widetilde{\bm Q}=\bm Q_\star+\bm\Delta_{\bm Q}$, and using
$\bm G_0\bm Q_\star=\bze$, gives
\begin{equation}
\label{eq:aligned-invariance}
  \bm G_0\bm\Delta_{\bm Q}
  +
  \bm\Delta\bm Q_\star
  +
  \bm\Delta\bm\Delta_{\bm Q}
  =
  \bm Q\bm\Lambda\bm\Xi^{-1}.
\end{equation}

Since $\bm G_0=\bm S^\trans\bm S$, we have
$\bm S\bm G_0^\dagger\bm G_0=\bm S$ and
$\bm S\bm G_0^\dagger\bm S^\trans
=\bm P_{\col(\bm S)}$. Moreover,
$\bm T\bm P_{\col(\bm S)}=\bm T$ on $\cE_{\rm emb}$. Applying
$\bm T\bm S\bm G_0^\dagger$ to
\cref{eq:aligned-invariance}, expanding
$\bm\Delta
=\bm S^\trans\bm L+\bm L^\trans\bm S+\bm\Delta_{\rm quad}$, and
simplifying yields
\begin{align}
  \bm T\bm S\bm\Delta_{\bm Q}
  ={}&
  -\bm T\bm S\bm G_0^\dagger
    \bm\Delta\bm Q_\star
  -\bm T\bm S\bm G_0^\dagger
    \bm\Delta\bm\Delta_{\bm Q}
  +
  \bm T\bm S\bm G_0^\dagger
    \bm Q\bm\Lambda\bm\Xi^{-1}
  \notag\\
  ={}&
  -\bm T\bm L\bm Q_\star
  -
  \bm T\bm S\bm G_0^\dagger
    \bm\Delta_{\rm quad}\bm Q_\star
  -
  \bm T\bm S\bm G_0^\dagger
    \bm\Delta\bm\Delta_{\bm Q}
  +
  \bm T\bm S\bm G_0^\dagger
    \bm Q\bm\Lambda\bm\Xi^{-1}.
\label{eq:TSDeltaQ-expansion}
\end{align}
For the second equality,
$\bm L^\trans\bm S\bm Q_\star=\bze$, while
\[
  \bm T\bm S\bm G_0^\dagger
  \bm S^\trans\bm L\bm Q_\star
  =
  \bm T\bm P_{\col(\bm S)}\bm L\bm Q_\star
  =
  \bm T\bm L\bm Q_\star.
\]
Thus $-\bm T\bm L\bm Q_\star$ is precisely the linear contribution caused
by the displacement of the near-zero eigenspace.

Substituting \cref{eq:TSDeltaQ-expansion} into
\cref{eq:error-proof-start} proves \cref{eq:error-decomposition}. The
higher-order remainder in that decomposition is the sum of the following four
terms:
\begin{equation}
\label{eq:error-higher-order-appendix}
\begin{aligned}
  \bm{\mathsf{Rem}}
  ={}&
  \underbrace{
    -\bm T\bm S\bm G_0^\dagger
    \bm\Delta_{\rm quad}\bm Q_\star
  }_{\displaystyle \eqqcolon\bm A_1}
  +
  \underbrace{
    \bm L\bm J\bm\Delta_{\bm Q}
  }_{\displaystyle \eqqcolon\bm A_2}
  +
  \underbrace{\left(
    -\bm T\bm S\bm G_0^\dagger
    \bm\Delta\bm\Delta_{\bm Q}
  \right)}_{\displaystyle \eqqcolon\bm A_3}
  +
  \underbrace{
    \bm T\bm S\bm G_0^\dagger
    \bm Q\bm\Lambda\bm\Xi^{-1}
  }_{\displaystyle \eqqcolon\bm A_4}.
\end{aligned}
\end{equation}
It remains to establish the bound in
\cref{eq:higher-order-deterministic}. The four estimates repeatedly use
\[
  \opn{\bm T\bm S\bm G_0^\dagger}
  \le
  \opn{\bm T}
  \opn{\bm S(\bm G_0^\dagger)^{1/2}}
  \opn{(\bm G_0^\dagger)^{1/2}}
  \lesssim
  \frac{1}{\gmin\theta},
\]
where the last inequality follows from
\cref{eq:transfer-op,eq:population-SGamma-bounds,eq:population-Gamma-bounds}.

\paragraph{Bound $\opn{\bm A_1}$.}
By submultiplicativity, the preceding estimate, and the block definition of
$\bm Q_\star$ together with \cref{eq:population-Rk-bounds}, we obtain
$\opn{\bm A_1}
\lesssim\opn{\bm\Delta_{\rm quad}}/(\gmin^2\theta)$.

\paragraph{Bound $\opn{\bm A_2}$.}
Using $\opn{\bm J}\le1$ and the coefficient bound in
\cref{lem:calibration}, submultiplicativity gives
$\opn{\bm A_2}
\lesssim
\opn{\bm L}\opn{\bm\Delta_{+0}}/
(\gmin^2\sqrt\theta)$.

\paragraph{Bound $\opn{\bm A_3}$.} Expanding
$\bm\Delta
=\bm S^\trans\bm L+\bm L^\trans\bm S+\bm\Delta_{\rm quad}$ gives
\begin{align*}
  \opn{\bm A_3}
  \le{}&
  \opn{
    \bm T\bm S\bm G_0^\dagger
    \bm S^\trans\bm L\bm\Delta_{\bm Q}
  }
  +
  \opn{
    \bm T\bm S\bm G_0^\dagger
    \bm L^\trans\bm S\bm\Delta_{\bm Q}
  } +
  \opn{
    \bm T\bm S\bm G_0^\dagger
    \bm\Delta_{\rm quad}\bm\Delta_{\bm Q}
  }.
\end{align*}

For the first term, use
$\bm S\bm G_0^\dagger\bm S^\trans
=\bm P_{\col(\bm S)}$ and
$\bm T=\bm T\bm P_{\col(\bm S)}$ to obtain
\[
  \opn{
    \bm T\bm S\bm G_0^\dagger
    \bm S^\trans\bm L\bm\Delta_{\bm Q}
  }
  \le
  \opn{\bm T}
  \opn{\bm P_{\col(\bm S)}\bm L}
  \opn{\bm\Delta_{\bm Q}}
  \lesssim
  \frac{
    \opn{\bm P_{\col(\bm S)}\bm L}
    \opn{\bm\Delta_{+0}}
  }{
    \gmin^2\theta
  }.
\]
Here the last inequality follows from
$\|\bm T\|_{\mathrm{op}}\lesssim\theta^{-1/2}$ and
$\|\bm\Delta_{\bm Q}\|_{\mathrm{op}}
\lesssim
\|\bm\Delta_{+0}\|_{\mathrm{op}}/
(\gmin^2\sqrt\theta)$.

For the second term,
$\bm L^\trans\bm S
=(\bm P_{\col(\bm S)}\bm L)^\trans\bm S$, and therefore
\[
  \opn{
    \bm T\bm S\bm G_0^\dagger
    \bm L^\trans\bm S\bm\Delta_{\bm Q}
  }
  \le
  \opn{\bm T\bm S\bm G_0^\dagger}
  \opn{\bm P_{\col(\bm S)}\bm L}
  \opn{\bm S\bm\Delta_{\bm Q}}
  \lesssim
  \frac{
    \opn{\bm P_{\col(\bm S)}\bm L}
    \opn{\bm\Delta_{+0}}
  }{
    \gmin^2\theta
  }.
\]
The last inequality uses
$\|\bm T\bm S\bm G_0^\dagger\|_{\mathrm{op}}
\lesssim(\gmin\theta)^{-1}$ and
$\|\bm S\bm\Delta_{\bm Q}\|_{\mathrm{op}}
\lesssim\|\bm\Delta_{+0}\|_{\mathrm{op}}/\gmin$.

For the quadratic term,
\[
  \opn{
    \bm T\bm S\bm G_0^\dagger
    \bm\Delta_{\rm quad}\bm\Delta_{\bm Q}
  }
  \lesssim
  \frac{
    \opn{\bm\Delta_{\rm quad}}
    \opn{\bm\Delta_{+0}}
  }{
    \gmin^3\theta^{3/2}
  }
  \lesssim
  \frac{\opn{\bm\Delta_{\rm quad}}}
       {\gmin^2\theta}.
\]
The final inequality follows from
$\|\bm\Delta_{+0}\|_{\mathrm{op}}
\lesssim\gmin\sqrt\theta$ on $\cE_{\rm good}$.

Combining the three bounds gives
\[
  \opn{\bm A_3}
  \lesssim
  \frac{
    \opn{\bm P_{\col(\bm S)}\bm L}
    \opn{\bm\Delta_{+0}}
  }{
    \gmin^2\theta
  }
  +
  \frac{\opn{\bm\Delta_{\rm quad}}}
       {\gmin^2\theta}.
\]

\paragraph{Bound $\opn{\bm A_4}$.}
We first control $\bm S\bm G_0^\dagger\bm Q$. Write the polar decomposition
as $\bm S=\bm U_S\bm G_0^{1/2}$, where $\bm U_S$ is a partial isometry from
$\ker(\bm G_0)^\perp$ onto $\col(\bm S)$. Since the range of
$(\bm G_0^\dagger)^{1/2}\bm Q$ lies in $\ker(\bm G_0)^\perp$, the map
$\bm U_S$ preserves its norm. Therefore
\[
\begin{aligned}
  \opn{\bm S\bm G_0^\dagger\bm Q}
  &=
  \opn{\bm U_S(\bm G_0^\dagger)^{1/2}\bm Q}
  =
  \opn{(\bm G_0^\dagger)^{1/2}\bm Q} \\
  &\le
  \opn{\bm G_0^\dagger}
  \opn{\bm G_0^{1/2}\bm Q}
  \lesssim
  \frac{\opn{\bm\Delta_{+0}}}{\gmin^2\theta}.
\end{aligned}
\]
The inequality uses
$(\bm G_0^\dagger)^{1/2}
=\bm G_0^\dagger\bm G_0^{1/2}$.
The final bound follows from \cref{eq:oracle-curvature} and
\cref{eq:relative-nullspace-conclusion}.

Next, $\bm\Lambda=\bm Q^\trans\bm G\bm Q$ represents the restriction of
$\bm G$ to the selected invariant subspace. The eigenvalue localization in
\cref{lem:nullspace} therefore gives
\[
  \opn{\bm\Lambda}
  \le
  \opn{\bm\Delta_{00}}
  +
  2\opn{\bm\Delta_{+0}}^2.
\]
Combining these estimates with
$\|\bm T\|_{\mathrm{op}}\lesssim\theta^{-1/2}$ and
$\|\bm\Xi^{-1}\|_{\mathrm{op}}\lesssim\gmin^{-1}$ yields
\begin{align*}
  \opn{\bm A_4}
  &\le
  \opn{\bm T}
  \opn{\bm S\bm G_0^\dagger\bm Q}
  \opn{\bm\Lambda}
  \opn{\bm\Xi^{-1}} \\
  &\lesssim
  \frac{
    \bigl(
      \opn{\bm\Delta_{00}}
      +
      \opn{\bm\Delta_{+0}}^2
    \bigr)
    \opn{\bm\Delta_{+0}}
  }{
    \gmin^3\theta^{3/2}
  } \\
  &\lesssim
  \frac{\opn{\bm\Delta_{\rm quad}}}
       {\gmin^2\theta}
  +
  \frac{\opn{\bm\Delta_{+0}}^2}
       {\gmin^2\theta}.
\end{align*}
For the last inequality, observe that $\bm\Delta_{00}
  =
  \bm P_{\ker(\bm G_0)}
  \bm\Delta_{\rm quad}
  \bm P_{\ker(\bm G_0)},$
so
$\|\bm\Delta_{00}\|_{\mathrm{op}}
\le\|\bm\Delta_{\rm quad}\|_{\mathrm{op}}$. Moreover,
$\|\bm\Delta_{+0}\|_{\mathrm{op}}
\lesssim\gmin\sqrt\theta$ on $\cE_{\rm good}$. Consequently,
\[
  \frac{
    \opn{\bm\Delta_{00}}
    \opn{\bm\Delta_{+0}}
  }{
    \gmin^3\theta^{3/2}
  }
  \lesssim
  \frac{\opn{\bm\Delta_{\rm quad}}}
       {\gmin^2\theta},
  \qquad
  \frac{\opn{\bm\Delta_{+0}}^3}
       {\gmin^3\theta^{3/2}}
  \lesssim
  \frac{\opn{\bm\Delta_{+0}}^2}
       {\gmin^2\theta}.
\]

Combining the bounds for $\bm A_1,\ldots,\bm A_4$ in
\cref{eq:error-higher-order-appendix} proves
\cref{eq:higher-order-deterministic} and completes the proof.
\end{proof}

\section{Probabilistic tools for the upper bound}
\label{sec:Gaussian}

\subsection{Gaussian concentration and matrix bounds}

We begin with the standard Gaussian concentration inequality. It will be
applied below to a Lipschitz function of the bulk noise matrix.

\begin{lemma}[Gaussian Lipschitz concentration]
\label{lem:gaussian-concentration}
Let $\bm g\sim\mathcal N(\bm 0,\bm I_N)$, and let
$f:\mathbb R^N\to\mathbb C$ be $L$-Lipschitz with respect to the Euclidean
norm. Then, for every $q\ge2$,
\begin{equation}
\label{eq:gaussian-lipschitz-moment}
  \|f(\bm g)-\mathbb Ef(\bm g)\|_{L^q}
  \le C\sqrt q\,L.
\end{equation}
\end{lemma}
The following bounds are consequences of the standard extreme-singular-value
estimates for Gaussian matrices; see
\cite[Theorem~5.32 and Corollary~5.35]{vershynin2010introduction}.
The moment bounds follow by integrating the corresponding tail inequalities.

\begin{lemma}[Gaussian matrix deviations and moments]
\label{lem:gaussian-matrix-deviations}
Let $\bm G\in\mathbb R^{u\times v}$ have independent
$\mathcal N(0,b^{-1})$ entries. Then, for every $x\ge0$,
\begin{equation}
\label{eq:gaussian-matrix-op-tail}
  \mathbb P\left\{
    \|\bm G\|_{\mathrm{op}}
    >
    \frac{\sqrt u+\sqrt v}{\sqrt b}+x
  \right\}
  \le 2e^{-cbx^2}.
\end{equation}
If $v=b$, then, for every $s\ge0$,
\begin{equation}
\label{eq:gaussian-wishart-tail}
  \mathbb P\left\{
    \|\bm G\bm G^\trans-\bm I_u\|_{\mathrm{op}}
    >
    C\left(
      \sqrt{\frac{u+s}{b}}+\frac{u+s}{b}
    \right)
  \right\}
  \le 2e^{-cs}.
\end{equation}
Moreover, for every $q\ge2$,
\begin{equation}
\label{eq:gaussian-matrix-op-moment}
  \|\bm G\|_{L^q}
  \le
  C\left(
    \sqrt{\frac ub}
    +
    \sqrt{\frac vb}
    +
    \sqrt{\frac qb}
  \right),
\end{equation}
and, when $v=b$,
\begin{equation}
\label{eq:gaussian-wishart-moment}
  \|\bm G\bm G^\trans-\bm I_u\|_{L^q}
  \le
  C\left(
    \sqrt{\frac{u+q}{b}}+\frac{u+q}{b}
  \right).
\end{equation}
\end{lemma}

We next control the Gaussian coefficient matrices that arise in the one-view
expansion. The result covers both independent bilinear forms and centered
same-sample quadratic forms.

\begin{lemma}[Gaussian coefficient moments]
\label{lem:gaussian-coefficient-moments}
Let $q\ge2$, and let
$\bm G_1\in\mathbb R^{r_1\times p}$ and
$\bm G_2\in\mathbb R^{r_2\times p}$ be independent matrices with independent
$\mathcal N(0,b^{-1})$ entries. For every deterministic
$\bm A\in\mathbb R^{r_1\times r_2}$,
\begin{equation}
\label{eq:gaussian-coefficient-moments}
  \|\bm G_1^\trans\bm A\bm G_2\|_{L^q}
  \le
  \frac{C}{b}
  \left\{
    (\sqrt p+\sqrt q)\|\bm A\|_{\mathrm F}
    +(p+q)\|\bm A\|_{\mathrm{op}}
  \right\}.
\end{equation}
If $r_1=r_2$ and $\bm A\in\mathbb R^{r_1\times r_1}$, then
\begin{equation}
\label{eq:gaussian-quadratic-moment}
  \left\|
    \bm G_1^\trans\bm A\bm G_1
    -b^{-1}\operatorname{tr}(\bm A)\bm I_p
  \right\|_{L^q}
  \le
  \frac{C}{b}
  \left\{
    (\sqrt p+\sqrt q)\|\bm A\|_{\mathrm F}
    +(p+q)\|\bm A\|_{\mathrm{op}}
  \right\}.
\end{equation}
Both conclusions remain valid conditionally when $\bm A$ is measurable with
respect to a sigma-field independent of the displayed Gaussian matrices.
\end{lemma}

\begin{proof}
We first record the Gaussian matrix estimate used for both conclusions.
Let $\bm H\in\mathbb R^{k\times d}$ have independent standard Gaussian
entries, and let $\bm B\in\mathbb R^{d\times m}$ be deterministic. Apply
Chevet's inequality \cite[Section~8.7]{vershynin2018high} with
$T=\bm B\mathbb S^{m-1}$ and $S=\mathbb S^{k-1}$. Since
\[
  \operatorname{rad}(T)=\|\bm B\|_{\mathrm{op}},
  \qquad
  w(T)
  =
  \mathbb E\|\bm B^\trans\bm g_d\|_2
  \le
  \|\bm B\|_{\mathrm F},
\]
whereas
$\operatorname{rad}(S)=1$ and
$w(S)=\mathbb E\|\bm g_k\|_2\le\sqrt k$, we obtain
\[
  \mathbb E\|\bm H\bm B\|_{\mathrm{op}}
  \le
  \|\bm B\|_{\mathrm F}
  +
  \sqrt k\,\|\bm B\|_{\mathrm{op}}.
\]

It remains to control the fluctuation around this mean. The map
$\bm H\mapsto\|\bm H\bm B\|_{\mathrm{op}}$ is
$\|\bm B\|_{\mathrm{op}}$-Lipschitz with respect to the Frobenius norm,
because
\[
  \left|
    \|\bm H_1\bm B\|_{\mathrm{op}}
    -
    \|\bm H_0\bm B\|_{\mathrm{op}}
  \right|
  \le
  \|(\bm H_1-\bm H_0)\bm B\|_{\mathrm{op}}
  \le
  \|\bm H_1-\bm H_0\|_{\mathrm F}
  \|\bm B\|_{\mathrm{op}}.
\]
Therefore, \cref{eq:gaussian-lipschitz-moment} and the preceding expectation
bound give
\begin{equation}
\label{eq:gaussian-fixed-coefficient-moment}
  \|\bm H\bm B\|_{L^q}
  \le
  C\left\{
    \|\bm B\|_{\mathrm F}
    +
    (\sqrt k+\sqrt q)\|\bm B\|_{\mathrm{op}}
  \right\}.
\end{equation}

We now prove the bilinear estimate. Set
$\bm Z_i=\sqrt b\,\bm G_i$, so that
$\bm Z_1$ and $\bm Z_2$ have independent standard Gaussian entries.
Conditional on $\bm Z_2$, apply
\cref{eq:gaussian-fixed-coefficient-moment} with
$\bm H=\bm Z_1^\trans$ and $\bm B=\bm A\bm Z_2$. Taking the \(L^q\) norm
over $\bm Z_2$ then gives
\begin{equation}
\label{eq:gaussian-bilinear-conditional}
  \|\bm Z_1^\trans\bm A\bm Z_2\|_{L^q}
  \le
  C\left\{
    \bigl\|\|\bm A\bm Z_2\|_{\mathrm F}\bigr\|_{L^q}
    +
    (\sqrt p+\sqrt q)
    \bigl\|\|\bm A\bm Z_2\|_{\mathrm{op}}\bigr\|_{L^q}
  \right\}.
\end{equation}

The two random norms on the right are controlled by the same two ingredients.
First, the map
$\bm Z_2\mapsto\|\bm A\bm Z_2\|_{\mathrm F}$ is
$\|\bm A\|_{\mathrm{op}}$-Lipschitz, and
\[
  \mathbb E\|\bm A\bm Z_2\|_{\mathrm F}
  \le
  \left(
    \mathbb E\|\bm A\bm Z_2\|_{\mathrm F}^2
  \right)^{1/2}
  =
  \sqrt p\,\|\bm A\|_{\mathrm F}.
\]
Gaussian concentration therefore gives
\[
  \bigl\|\|\bm A\bm Z_2\|_{\mathrm F}\bigr\|_{L^q}
  \le
  \sqrt p\,\|\bm A\|_{\mathrm F}
  +
  C\sqrt q\,\|\bm A\|_{\mathrm{op}}.
\]
Second, applying \cref{eq:gaussian-fixed-coefficient-moment} to
$\bm Z_2^\trans\bm A^\trans$ yields
\[
  \bigl\|\|\bm A\bm Z_2\|_{\mathrm{op}}\bigr\|_{L^q}
  \le
  C\left\{
    \|\bm A\|_{\mathrm F}
    +
    (\sqrt p+\sqrt q)\|\bm A\|_{\mathrm{op}}
  \right\}.
\]
Substituting these two bounds into
\cref{eq:gaussian-bilinear-conditional} and using
$(\sqrt p+\sqrt q)^2\le2(p+q)$ gives
\[
  \|\bm Z_1^\trans\bm A\bm Z_2\|_{L^q}
  \le
  C\left\{
    (\sqrt p+\sqrt q)\|\bm A\|_{\mathrm F}
    +(p+q)\|\bm A\|_{\mathrm{op}}
  \right\}.
\]
Since
$\bm G_1^\trans\bm A\bm G_2
=b^{-1}\bm Z_1^\trans\bm A\bm Z_2$, this proves
\cref{eq:gaussian-coefficient-moments}.

For the same-sample estimate, let $\bm Z'$ be an independent copy of
$\bm Z_1$ and define
\[
  \bm U\coloneqq\frac{\bm Z_1+\bm Z'}{\sqrt2},
  \qquad
  \bm V\coloneqq\frac{\bm Z_1-\bm Z'}{\sqrt2}.
\]
The matrices $\bm U$ and $\bm V$ are independent standard Gaussian matrices.
By symmetrization and the identity
\[
  \bm Z_1^\trans\bm A\bm Z_1
  -
  \bm Z'^\trans\bm A\bm Z'
  =
  \bm U^\trans\bm A\bm V
  +
  \bm V^\trans\bm A\bm U,
\]
the bilinear estimate gives
\[
\begin{aligned}
  \left\|
    \bm Z_1^\trans\bm A\bm Z_1
    -
    \operatorname{tr}(\bm A)\bm I_p
  \right\|_{L^q}
  &\le
  \left\|
    \bm U^\trans\bm A\bm V
    +
    \bm V^\trans\bm A\bm U
  \right\|_{L^q} \\
  &\le
  C\left\{
    (\sqrt p+\sqrt q)\|\bm A\|_{\mathrm F}
    +(p+q)\|\bm A\|_{\mathrm{op}}
  \right\}.
\end{aligned}
\]
Rescaling by $b^{-1}$ proves
\cref{eq:gaussian-quadratic-moment}. The conditional versions follow by
applying the deterministic bounds conditionally on the sigma-field with
respect to which $\bm A$ is measurable and then taking the outer
expectation.
\end{proof}

\section{One-view weighted leakage isometry}
\label{app:one-view-isometry}

This section proves the one-view result used in
\cref{lem:one-view}.  Throughout the section, $C,c>0$ denote numerical
constants whose values may change from line to line.


\subsection{Normalized model and theorem}
\label{app:one-view-statement}

The leakage of an empirical singular vector depends on the corresponding
population singular value.  The purpose of the weighting is to remove this
dependence: after weighting, the leakage Gram matrix should be close to the
same scalar matrix for every signal direction.  We formulate this statement
first in normalized noise units.

\paragraph{Normalized observation.}

Let $\bm X\in\R^{n\times d}$ have rank $p$ and compact SVD
\[
  \bm X
  =
  \bm U_\star\bm\Sigma_\star\bm V_\star^\trans,
  \qquad
  \bm\Sigma_\star
  =
  \diag(s_1^\star,\ldots,s_p^\star),
\]
where $\bm U_\star\in\St(n,p)$, $\bm V_\star\in\St(d,p)$, and
$s_1^\star\ge\cdots\ge s_p^\star\ge s_\star>0$.  Set
$a\coloneqq n-p$, $b\coloneqq d-p$, and $\beta\coloneqq a/b$.  We observe
\begin{equation}
\label{eq:one-view-model}
  \bm Y=\bm X+\bm Z,
  \qquad
  Z_{ij}\stackrel{\mathrm{iid}}{\sim}N(0,b^{-1}).
\end{equation}
Let $\bm U\bm\Sigma\bm V^\trans$ be a leading rank-$p$ SVD of $\bm Y$,
where
\[
  \bm\Sigma
  =
  \diag(\widehat s_1,\ldots,\widehat s_p),
  \qquad
  \widehat s_1\ge\cdots\ge\widehat s_p.
\]
Choose an orthonormal complement $\bm U_\star^\perp$ of $\bm U_\star$.
Then $(\bm U_\star^\perp)^\trans\bm U$ contains the components of the
empirical left singular vectors outside the population signal space.

\paragraph{Bias-correcting weights.}

As discussed in \cref{sec:ajive-bias}, the Gram matrix of these components
has a diagonal bias that varies with the signal singular values.  To correct
it, define
\begin{equation}
\label{eq:one-view-t-coordinate}
  \tau_\beta(s)
  \coloneqq
  \begin{cases}
  \displaystyle
  \frac{
    s^2-(1+\beta)
    +
    \sqrt{\{s^2-(1+\beta)\}^2-4\beta}
  }{2},
  &
  s>1+\sqrt\beta,
  \\[8pt]
  \sqrt\beta,
  &
  s\le1+\sqrt\beta.
  \end{cases}
\end{equation}
For each empirical singular value $\widehat s_i$, set
\[
  \tau_i\coloneqq\tau_\beta(\widehat s_i),
  \qquad
  \ell_i
  \coloneqq
  \sqrt{\frac{\tau_i(\tau_i+\beta)}{\tau_i+1}},
  \qquad
  \bm W_L\coloneqq\diag(\ell_1,\ldots,\ell_p).
\]
When $\tau_i>\sqrt\beta$, the first branch of
\cref{eq:one-view-t-coordinate} is equivalent to
\begin{equation}
\label{eq:one-view-upper-branch-relation}
  \widehat s_i^2
  =
  1+\beta+\tau_i+\frac{\beta}{\tau_i}.
\end{equation}
Below the threshold $1+\sqrt\beta$, the definition fixes
$\tau_i=\sqrt\beta$ and hence $\ell_i^2=\beta$.

The weighted leakage matrix is
\begin{equation}
\label{eq:one-view-NL}
  \bm N_L
  \coloneqq
  (\bm U_\star^\perp)^\trans\bm U\bm W_L.
\end{equation}
The choice of $\bm W_L$ is designed so that
$\bm N_L^\trans\bm N_L$ is centered at $\beta\bm I_p$.  The theorem below
makes this statement quantitative.

\paragraph{One-view theorem.}

The relevant signal scale is
$s_\star^2\wedge s_\star^4$.  For $s_\star\le1$, this equals
$s_\star^4$, while for $s_\star\ge1$, it equals $s_\star^2$.  Thus the
assumption below covers the two signal regimes without imposing an upper
bound on $s_1^\star$.

\begin{theorem}[One-view weighted leakage isometry]
\label{thm:one-view-isometry}
Suppose that
\[
  p\le\frac13\min\{n,d\},
  \qquad
  s_\star^2\wedge s_\star^4\ge\nu^2\beta,
\]
where $\nu>0$ is a sufficiently large numerical constant.  Then
\begin{subequations}
\label{eq:one-view-theorem-bounds}
\begin{align}
  \Lqn{
    \bm N_L^\trans\bm N_L-\beta\bm I_p
  }{32}
  &\le
  C\beta\sqrt{\frac pa},
  \label{eq:one-view-theorem-centered}
  \\
  \Lqn{\bm N_L}{64}
  &\le
  C\sqrt\beta.
  \label{eq:one-view-theorem-size}
\end{align}
\end{subequations}
Moreover, there is an event $\cE_{\rm 1v}$ such that
\begin{equation}
\label{eq:one-view-theorem-good-event-prob}
  \Pp(\cE_{\rm 1v}^c)
  \le
  C\exp\left\{
    -\frac{b(s_\star^2\wedge s_\star^4)}{C}
  \right\},
\end{equation}
and, on $\cE_{\rm 1v}$,
\begin{equation}
\label{eq:one-view-theorem-weight-floor}
  \bm W_L^2-\beta\bm I_p
  \succeq
  c(s_\star^2\wedge s_\star^4)\bm I_p.
\end{equation}
Here $C,c>0$ are numerical constants.
\end{theorem}

\subsection{Proof of \cref{lem:one-view}}

Fix $k\in\{1,2\}$.  The normalized theorem supplies both the moment bounds
for $\bm L_k$ and a lower bound for the weighted empirical Gram matrix.  We
first translate its conclusions back to the original noise units and then combine
the latter lower bound with the centered leakage estimate to control
$\bm S_k^\trans\bm S_k$.

Let
$\bm X_k=\bm U_{\star,k}\bm\Sigma_{\star,k}\bm V_{\star,k}^\trans$
be a compact SVD, and choose an orthonormal complement
$\bm U_{\star,k}^\perp$ of $\bm U_{\star,k}$.  Apply
\cref{thm:one-view-isometry} to
$\widetilde{\bm Y}_k\coloneqq
\bm Y_k/(\sigma\sqrt{b_k})$, with
$a=a_k$, $b=b_k$, $\beta=a_k/b_k$, and
$s_\star=s_k/(\sigma\sqrt{b_k})$.  The noise entries of
$\widetilde{\bm Y}_k$ are independent $N(0,b_k^{-1})$, and its empirical
left singular vectors are the columns of $\widehat{\bm\Phi}_k$.

The signal assumption follows directly from the definition of $\gamma_k$.
Indeed, $b_k=d_k-p_k$ and $p_k\le d_k/3$ imply
$b_k\le d_k\le3b_k/2$, so
$s_\star^2\wedge s_\star^4
\asymp\gamma_k^2/(b_k\sigma^2)$.
The standing condition $\gamma_k\ge\nu\sigma\sqrt n$ therefore gives
$s_\star^2\wedge s_\star^4
\gtrsim\nu^2a_k/b_k$ because $a_k\le n$.
Thus the signal condition in \cref{thm:one-view-isometry} holds after
increasing the numerical constant $\nu$ if necessary.

Let $\bm W_L$ and $\bm N_L$ be the matrices in the normalized theorem.
Substitution into the definitions of $\tau_k$ and $w_k$ gives
$w_k(\widehat s_{k,i})=\sigma\sqrt{b_k}\,\ell_i$.  Hence
$\bm C_k=\sigma\sqrt{b_k}\,\widehat{\bm\Phi}_k\bm W_L$ and
$\bm L_k=\sigma\sqrt{b_k}\,\bm U_{\star,k}^\perp\bm N_L$.
Rescaling the two conclusions of \cref{thm:one-view-isometry} gives
\begin{equation}
\label{eq:one-view-bounds-proof}
  \Lqn{
    \bm L_k^\trans\bm L_k-a_k\sigma^2\bm I_{p_k}
  }{32}
  \lesssim
  \sigma^2\sqrt{a_kp_k},
  \qquad
  \Lqn{\bm L_k}{64}
  \lesssim
  \sigma\sqrt{a_k}.
\end{equation}
This proves \cref{eq:one-view-bounds}.

Let $\widetilde{\cE}_k$ denote the event supplied by
\cref{thm:one-view-isometry} for the normalized $k$th view.  Since
$\widehat{\bm\Phi}_k$ has orthonormal columns, the lower bound for
$\bm W_L^2-\frac{a_k}{b_k}\bm I_{p_k}$ and the preceding comparison of
signal scales imply, on $\widetilde{\cE}_k$,
\[
  \bm C_k^\trans\bm C_k-a_k\sigma^2\bm I_{p_k}
  \succeq
  c\gamma_k^2\bm I_{p_k}.
\]
Define
\begin{equation}
\label{eq:one-view-view-event}
  \cE_k
  \coloneqq
  \widetilde{\cE}_k
  \cap
  \left\{
    \opn{
      \bm L_k^\trans\bm L_k-a_k\sigma^2\bm I_{p_k}
    }
    \le
    \frac c2\gamma_k^2
  \right\}.
\end{equation}
Because $\bm C_k=\bm S_k+\bm L_k$ and
$\bm S_k^\trans\bm L_k=\bm 0$, on $\cE_k$ we have
\begin{equation}
\label{eq:one-view-signal-Gram-lower}
  \bm S_k^\trans\bm S_k
  =
  \bigl(
    \bm C_k^\trans\bm C_k-a_k\sigma^2\bm I_{p_k}
  \bigr)
  -
  \bigl(
    \bm L_k^\trans\bm L_k-a_k\sigma^2\bm I_{p_k}
  \bigr)
  \succeq
  \frac c2\gamma_k^2\bm I_{p_k}.
\end{equation}
Choosing $\cemb\le\sqrt{c/2}$ proves
$\sigma_{\min}(\bm S_k)\ge\cemb\gamma_k$ on $\cE_k$.

Finally, the probability bound in \cref{thm:one-view-isometry} and Markov's
inequality applied to \cref{eq:one-view-bounds-proof} give
\[
  \mathbb P(\cE_k^c)
  \lesssim
  \exp\left(-\frac{\gamma_k^2}{C\sigma^2}\right)
  +
  \left(
    \frac{\sigma^2\sqrt{a_kp_k}}{\gamma_k^2}
  \right)^{32}
  \lesssim
  \frac{\sigma^2\sqrt{np_k}}{\gamma_k^2}.
\]
The last inequality uses $a_k\le n$ and
$\gamma_k^2\gtrsim n\sigma^2$.  This proves the asserted probability bound
and completes the proof.

\subsection{Proof of \cref{thm:one-view-isometry}}
\label{app:proof-one-view-isometry}

We first prove the result in the wide case $a\le b$, so that $\beta\le1$.
Afterwards, we derive a corresponding right-leakage bound under the stronger
signal condition $s_\star^2\ge\nu^2$ and use it to handle the tall case by
transposition.

By rotational invariance of the Gaussian noise, we may work in the
population singular-vector coordinates, in which
\begin{equation}
\label{eq:one-view-block-model}
  \bm Y
  =
  \begin{pmatrix}
    \bm\Sigma_\star+\bm Z_{11} & \bm Z_{12}\\
    \bm Z_{21} & \bm Z_{22}
  \end{pmatrix},
\end{equation}
where the four noise blocks have independent $N(0,b^{-1})$ entries and
$\bm Z_{22}\in\R^{a\times b}$. Partition the empirical singular vectors
conformably as
\begin{equation}
\label{eq:one-view-singular-blocks}
  \bm U
  =
  \begin{pmatrix}
    \bm U_\parallel\\
    \bm U_\perp
  \end{pmatrix},
  \qquad
  \bm V
  =
  \begin{pmatrix}
    \bm V_\parallel\\
    \bm V_\perp
  \end{pmatrix}.
\end{equation}
In these coordinates, $\bm N_L=\bm U_\perp\bm W_L$. Thus
\cref{eq:one-view-theorem-centered} amounts to controlling
\begin{equation}
\label{eq:one-view-XL}
  \bm X_L
  \coloneqq
  \bm W_L\bm U_\perp^\trans\bm U_\perp\bm W_L
  -\beta\bm I_p.
\end{equation}

We first work on the localization event
\begin{equation}
\label{eq:one-view-localization-event}
  \cE_{\rm loc}
  \coloneqq
  \left\{
    \min_i\tau_i\ge\frac{s_\star^2}{16},
    \quad
    \opn{\bm Z_{22}\bm Z_{22}^\trans-\bm I_a}
    \le\frac{s_\star^2}{128}
  \right\}.
\end{equation}
The signal condition implies $s_\star^2/16>\sqrt\beta$ for sufficiently
large $\nu$.  Thus the first condition in $\cE_{\rm loc}$ places every
selected singular value on the upper branch, where
\cref{eq:one-view-upper-branch-relation} applies.  The second condition will
make the singular-vector expansion contractive.

\begin{lemma}
\label{lem:one-view-localization}
Under $s_\star^2\wedge s_\star^4\ge\nu^2\beta$,
\begin{equation}
\label{eq:one-view-tail-under-left}
  \Pp(\cE_{\rm loc}^c)
  \lesssim
  \exp\left\{
    -c b(s_\star^2\wedge s_\star^4)
  \right\}.
\end{equation}
Under the stronger condition $s_\star^2\ge\nu^2$,
$\Pp(\cE_{\rm loc}^c)\lesssim\exp(-cbs_\star^2)$.
\end{lemma}

\paragraph{Step 1: control $\bm X_L$ on $\cE_{\rm loc}$.}
The main difficulty is that the left and right leakage Grams
$\bm U_\perp^\trans\bm U_\perp$ and
$\bm V_\perp^\trans\bm V_\perp$ are coupled. The two leakage Grams satisfy a pair of approximate linear relations. To state them, define the following three entrywise operators 
\begin{subequations}
\label{eq:one-view-population-operators}
\begin{align}
  [\mathcal A_U(\bm M)]_{ij}
  &\coloneqq
  \frac{\beta r_ir_j}{\tau_i\tau_j-\beta}M_{ij},
  \label{eq:one-view-population-operator-U}\\
  [\mathcal B(\bm M)]_{ij}
  &\coloneqq
  \frac{\beta}{\tau_i\tau_j-\beta}M_{ij},
  \label{eq:one-view-population-operator-B}\\
  [\mathcal A_V(\bm M)]_{ij}
  &\coloneqq
  \frac{\ell_i\ell_j}{\tau_i\tau_j-\beta}M_{ij},
  \label{eq:one-view-population-operator-V}
\end{align}
\end{subequations}
where $r_i\coloneqq\sqrt{\frac{\tau_i(\tau_i+1)}{\tau_i+\beta}}$ and we define $\bm W_R\coloneqq\diag(r_1,\ldots,r_p)$. 
These operators are well defined on $\cE_{\rm loc}$ because
$\tau_i>\sqrt\beta$ for every $i$. Define the residuals
\begin{subequations}
\label{eq:one-view-eps}
\begin{align}
  \bm\varepsilon_U
  &\coloneqq
  \bm U_\perp^\trans\bm U_\perp
  -
  \mathcal A_U(\bm V_\parallel^\trans\bm V_\parallel)
  -
  \mathcal B(\bm U_\parallel^\trans\bm U_\parallel),
  \label{eq:one-view-eps-U}\\
  \bm\varepsilon_V
  &\coloneqq
  \bm V_\perp^\trans\bm V_\perp
  -
  \mathcal B(\bm V_\parallel^\trans\bm V_\parallel)
  -
  \mathcal A_V(\bm U_\parallel^\trans\bm U_\parallel).
  \label{eq:one-view-eps-V}
\end{align}
\end{subequations}

Since the empirical singular-vector matrices have orthonormal columns,
$\bm U_\parallel^\trans\bm U_\parallel+
\bm U_\perp^\trans\bm U_\perp=\bm I_p$ and similarly for $\bm V$.
Thus \cref{eq:one-view-eps} forms a linear system for the two leakage
Grams. To solve it, define
$\bm X_R\coloneqq
\bm W_R\bm V_\perp^\trans\bm V_\perp\bm W_R-\bm I_p$.
Fix $i,j$ and abbreviate
$u=\tau_i\tau_j$, $\ell=\ell_i\ell_j$, $r=r_ir_j$,
$x=(\bm X_L)_{ij}$, and $y=(\bm X_R)_{ij}$.
Since $\ell_ir_i=\tau_i$, we have $\ell r=u$, and the two equations reduce to
\begin{equation}
\label{eq:one-view-entrywise-system}
  \begin{pmatrix}
    \dfrac{u}{u-\beta}
    &
    \dfrac{\beta\ell}{u-\beta}
    \\[5pt]
    \dfrac{r}{u-\beta}
    &
    \dfrac{u}{u-\beta}
  \end{pmatrix}
  \begin{pmatrix}
    x\\y
  \end{pmatrix}
  =
  \begin{pmatrix}
    \ell(\bm\varepsilon_U)_{ij}\\
    r(\bm\varepsilon_V)_{ij}
  \end{pmatrix}.
\end{equation}
Direct inversion gives the cancellation identities
\begin{equation}
\label{eq:one-view-cancellation}
  {
  \bm X_L
  =
  \bm W_L\bm\varepsilon_U\bm W_L-\beta\bm\varepsilon_V,
  \qquad
  \bm X_R
  =
  \bm W_R\bm\varepsilon_V\bm W_R-\bm\varepsilon_U.}
\end{equation}

The following lemma controls the two residuals in precisely the combinations
appearing here.

\begin{lemma}
\label{lem:one-view-gram-approximation}
Under $s_\star^2\wedge s_\star^4\ge\nu^2\beta$,
\begin{equation}
\label{eq:one-view-left-error-bound}
  \Lqn{
    \ind_{\cE_{\rm loc}}
    \bm W_L\bm\varepsilon_U\bm W_L
  }{32}
  +
  \beta
  \Lqn{
    \ind_{\cE_{\rm loc}}\bm\varepsilon_V
  }{32}
  \lesssim
  \frac{\sqrt{ap}}b.
\end{equation}
Under the stronger condition $s_\star^2\ge\nu^2$,
\begin{equation}
\label{eq:one-view-right-error-bound}
  \Lqn{
    \ind_{\cE_{\rm loc}}
    \bm W_R\bm\varepsilon_V\bm W_R
  }{32}
  +
  \Lqn{
    \ind_{\cE_{\rm loc}}\bm\varepsilon_U
  }{32}
  \lesssim
  \sqrt{\frac pb}.
\end{equation}
\end{lemma}

By \cref{eq:one-view-cancellation,eq:one-view-left-error-bound},
$\|\ind_{\cE_{\rm loc}}\bm X_L\|_{L^{32}}
\lesssim\sqrt{ap}/b$.

\paragraph{Step 2: control $\bm X_L$ on $\cE_{\rm loc}^c$.}
From the definition of $\bm X_L$,
$\|\bm X_L\|_{\mathrm{op}}
\le\|\bm U_\perp\bm W_L\|_{\mathrm{op}}^2+\beta$.
We use the following radial estimate.

\begin{lemma}
\label{lem:one-view-radial}
In the wide case $a\le b$,
\begin{equation}
\label{eq:one-view-radial-bound}
  \Lqn{\bm U_\perp\bm W_L}{128}
  \lesssim
  \sqrt\beta,
  \qquad
  \Lqn{\bm V_\perp\bm W_R}{128}
  \lesssim
  1.
\end{equation}
\end{lemma}

H\"older's inequality,
\cref{eq:one-view-tail-under-left,eq:one-view-radial-bound}, and
$b(s_\star^2\wedge s_\star^4)\ge\nu^2a$ give
$\|\ind_{\cE_{\rm loc}^c}\bm X_L\|_{L^{32}}
\lesssim
\beta e^{-ca}\lesssim\sqrt{ap}/b$.
Together with the bound on $\cE_{\rm loc}$, this proves
\cref{eq:one-view-theorem-centered} in the wide case.

\paragraph{Completing the wide-case proof.}
The radial bound \cref{eq:one-view-theorem-size} and the localization
probability \cref{eq:one-view-theorem-good-event-prob} follow directly
from \cref{lem:one-view-radial,lem:one-view-localization}.
Finally, on $\cE_{\rm loc}$,
$\tau_i\ge s_\star^2/16$ and
\[
  \ell_i^2-\beta
  =
  \frac{\tau_i^2-\beta}{\tau_i+1}
  \ge
  \frac{(s_\star^2/16)^2-\beta}{s_\star^2/16+1}
  \gtrsim
  s_\star^2\wedge s_\star^4,
\]
where we used the monotonicity of
$t\mapsto(t^2-\beta)/(t+1)$ on $[\sqrt\beta,\infty)$.
This proves \cref{eq:one-view-theorem-weight-floor} and completes the
wide-case proof.

\paragraph{A right-leakage bound for wide matrices.}
To pass to the tall case, we need the corresponding control of the
\emph{right} leakage for a wide matrix. Choose an orthonormal complement
$\bm V_\star^\perp$ of $\bm V_\star$ and define
$\bm N_R\coloneqq(\bm V_\star^\perp)^\trans\bm V\bm W_R$.
In the population coordinates,
$\bm N_R^\trans\bm N_R-\bm I_p=\bm X_R$.

Suppose now that $s_\star^2\ge\nu^2$. On $\cE_{\rm loc}$, the second
identity in \cref{eq:one-view-cancellation} and
\cref{eq:one-view-right-error-bound} give
$\|\ind_{\cE_{\rm loc}}\bm X_R\|_{L^{32}}
\lesssim\sqrt{p/b}$.
On the complement,
$\|\bm X_R\|_{\mathrm{op}}
\le\|\bm V_\perp\bm W_R\|_{\mathrm{op}}^2+1$, so
\cref{lem:one-view-localization,lem:one-view-radial} and H\"older's
inequality give the same bound. Hence
\begin{align}
  \Lqn{\bm N_R^\trans\bm N_R-\bm I_p}{32}
  &\lesssim
  \sqrt{\frac pb},
  &
  \Lqn{\bm N_R}{64}
  &\lesssim
  1,
  \notag\\
  \Pp(\cE_{\rm loc}^c)
  &\lesssim
  e^{-cbs_\star^2},
  &
  \min_i(r_i^2-1)
  &\gtrsim
  s_\star^2
  \quad\text{on }\cE_{\rm loc}.
\label{eq:one-view-wide-right-bounds}
\end{align}
For the last bound, use
$r_i^2-1=(\tau_i^2-\beta)/(\tau_i+\beta)$,
$\tau_i\ge s_\star^2/16$, and $s_\star^2\ge\nu^2$.

\paragraph{The tall case.}
Suppose now that $a>b$, so $\beta>1$, and consider the rescaled transpose
\[
  \widetilde{\bm Y}
  \coloneqq
  \beta^{-1/2}\bm Y^\trans.
\]
This is a wide matrix with
$\widetilde a=b$, $\widetilde b=a$,
$\widetilde\beta=\beta^{-1}$, and
$\widetilde s_\star=\beta^{-1/2}s_\star$.
The original signal condition implies
$s_\star^2\ge\nu^2\beta$, and hence
$\widetilde s_\star^2=s_\star^2/\beta\ge\nu^2$.
Thus \cref{eq:one-view-wide-right-bounds} applies to
$\widetilde{\bm Y}$.

It remains only to identify its right-leakage quantities with the original
left-leakage quantities. The empirical singular values satisfy
$\widehat{\widetilde s}_i=\widehat s_i/\sqrt\beta$, and the definition of the bias-correction
map gives
$\widetilde\tau_i=\tau_i/\beta$. Consequently,
\begin{equation}
\label{eq:one-view-transpose-dictionary}
  \widetilde{\bm W}_R
  =
  \beta^{-1/2}\bm W_L,
  \qquad
  \widetilde{\bm N}_R
  =
  \beta^{-1/2}\bm N_L,
  \qquad
  \widetilde{\bm N}_R^\trans\widetilde{\bm N}_R-\bm I_p
  =
  \beta^{-1}
  \bigl(
    \bm N_L^\trans\bm N_L-\beta\bm I_p
  \bigr).
\end{equation}
Indeed, the right singular vectors of $\widetilde{\bm Y}$ are the left
singular vectors of $\bm Y$, and
$\widetilde r_i^2=\ell_i^2/\beta$.

The centered right-leakage bound for $\widetilde{\bm Y}$ therefore gives
\[
  \Lqn{
    \bm N_L^\trans\bm N_L-\beta\bm I_p
  }{32}
  \lesssim
  \beta\sqrt{\frac pa}
  =
  \frac{\sqrt{ap}}b,
\]
while the radial bound gives
$\|\bm N_L\|_{L^{64}}\lesssim\sqrt\beta$.
Moreover,
$\widetilde b\,\widetilde s_\star^2
=a s_\star^2/\beta=bs_\star^2$.
Since the tall-case signal condition forces $s_\star^2>1$, we have
$s_\star^2\wedge s_\star^4=s_\star^2$, so the right-side localization
bound gives \cref{eq:one-view-theorem-good-event-prob}.
Finally,
$\widetilde r_i^2-1=(\ell_i^2-\beta)/\beta$, and the right-side weight
floor gives
$\ell_i^2-\beta\gtrsim s_\star^2
=s_\star^2\wedge s_\star^4$.
This proves all the claims in the tall case and completes the proof.

\subsection{Proof of \cref{lem:one-view-localization}}
\label{app:proof-one-view-localization}

We first prove the claim under the left signal condition
$s_\star^2\wedge s_\star^4\ge\nu^2\beta$.
Introduce the auxiliary event
\begin{equation}
\label{eq:one-view-localization-aux}
  \cE_0
  \coloneqq
  \left\{
    \opn{\bm Z_{11}}\le\frac{s_\star}4,
    \quad
    \opn{
      \bm Z_{12}\bm Z_{12}^\trans-\bm I_p
    }
    \le\frac{s_\star^2}{8},
    \quad
    \opn{
      \bm Z_{22}\bm Z_{22}^\trans-\bm I_a
    }
    \le\frac{s_\star^2}{128}
  \right\}.
\end{equation}

\paragraph{Step 1: show that $\cE_0\subseteq\cE_{\rm loc}$.}
Consider the first $p$ rows of $\bm Y$,
$\bm M\coloneqq(\bm\Sigma_\star+\bm Z_{11},\bm Z_{12})$.
By the definition of singular values,
$\widehat s_p^2\ge\lambda_{\min}(\bm M\bm M^\trans)$.
On $\cE_0$, Weyl's inequality gives
$\sigma_{\min}(\bm\Sigma_\star+\bm Z_{11})\ge3s_\star/4$ and
$\bm Z_{12}\bm Z_{12}^\trans
\succeq(1-s_\star^2/8)\bm I_p$.
Consequently,
$\widehat s_p^2\ge1+7s_\star^2/16$.

We next translate this bound to the $\tau$-coordinate. The signal condition
implies $s_\star^2\ge\nu\sqrt\beta$, so for sufficiently large $\nu$,
$\widehat s_i\ge \widehat s_p>1+\sqrt\beta$ for every $i$. Hence every $\widehat s_i$ lies on the
upper branch of \cref{eq:one-view-t-coordinate}, so
$\tau_i>\sqrt\beta$ and \cref{eq:one-view-upper-branch-relation} applies.
Moreover, the signal condition gives
$1+\beta+s_\star^2/16+16\beta/s_\star^2
\le1+7s_\star^2/16$.
Since $t\mapsto1+\beta+t+\beta/t$ is increasing on
$(\sqrt\beta,\infty)$, we obtain
$\tau_i\ge s_\star^2/16$ for every $i$.

The second condition defining $\cE_{\rm loc}$ is already part of $\cE_0$,
so $\cE_0\subseteq\cE_{\rm loc}$.

\paragraph{Step 2: bound $\Pp(\cE_0^c)$.}
Since $p/b\le\beta/2$, $a/b=\beta$, and
$s_\star^2\wedge s_\star^4\ge\nu^2\beta$, we have
$\sqrt{p/b}\vee\sqrt{a/b}\le\sqrt\beta
\le(s_\star^2\wedge s_\star)/\nu$.

Apply \cref{eq:gaussian-matrix-op-tail} to $\bm Z_{11}$ with
$u=v=p$ and $x=s_\star/8$. For sufficiently large $\nu$,
$2\sqrt{p/b}\le s_\star/8$, and hence
$\Pp\{\|\bm Z_{11}\|_{\mathrm{op}}>s_\star/4\}
\lesssim
\exp\{-cb(s_\star^2\wedge s_\star^4)\}$.

The block dimensions in \cref{eq:one-view-block-model} are
$\bm Z_{12}\in\mathbb R^{p\times b}$ and
$\bm Z_{22}\in\mathbb R^{a\times b}$, and both matrices have independent
$N(0,b^{-1})$ entries. Thus \cref{eq:gaussian-wishart-tail} applies with
$u=p$ and $u=a$, respectively, while $v=b$ in both cases. Choose the tail
parameter in that inequality to be
$c_0b(s_\star^2\wedge s_\star^4)$, where $c_0>0$ is a sufficiently small
numerical constant. For either $u\in\{p,a\}$, the resulting deviation
threshold satisfies
\[
\begin{aligned}
  C\left[
    \sqrt{\frac{u}{b}+c_0(s_\star^2\wedge s_\star^4)}
    +\frac{u}{b}
    +c_0(s_\star^2\wedge s_\star^4)
  \right]
  &\le
  C\left[
    \sqrt{\frac{u}{b}}
    +\sqrt{c_0}(s_\star^2\wedge s_\star)
    +\frac{u}{b}
    +c_0(s_\star^2\wedge s_\star^4)
  \right] \le
  \frac{s_\star^2}{128}.
\end{aligned}
\]
The last inequality follows from the preceding bound
$\sqrt{u/b}\le(s_\star^2\wedge s_\star)/\nu$, first by choosing $c_0$
sufficiently small and then by choosing $\nu$ sufficiently large. Therefore
the failure probabilities for the second and third conditions in
\cref{eq:one-view-localization-aux} are each bounded by
$2\exp\{-cb(s_\star^2\wedge s_\star^4)\}$.

Therefore, a union bound and $\cE_0\subseteq\cE_{\rm loc}$ give
\[
  \Pp(\cE_{\rm loc}^c)
  \lesssim
  \exp\{-cb(s_\star^2\wedge s_\star^4)\},
\]
which proves the first claim.

\paragraph{The stronger right signal condition.}
Suppose now that $s_\star^2\ge\nu^2$.
Since $\nu>1$ and $\beta\le1$, the left signal condition also holds, so
the inclusion $\cE_0\subseteq\cE_{\rm loc}$ proved above remains valid.
In the preceding Gaussian estimates we may now use
$s_\star^2\wedge s_\star=s_\star$, which gives
$\Pp(\cE_{\rm loc}^c)\lesssim e^{-cbs_\star^2}$.
This proves the second claim.
 
\subsection{Proof of \cref{lem:one-view-gram-approximation}}
\label{app:proof-one-view-gram-approximation}
On $\cE_{\rm loc}$, the proof proceeds in four stages. We first solve the
lower singular-vector equations by an absolutely convergent series and use
that series to derive exact expansions for the left and right leakage Gram
matrices. We then condition on $\bm Z_{22}$, replace the Gaussian coefficient
matrices involving $\bm Z_{12}$ and $\bm Z_{21}$ by their conditional means,
and collect the centered fluctuations. The remaining random coefficients are
normalized traces of functions of $\bm Z_{22}$; we replace them by the
deterministic Marchenko--Pastur transform. Finally, we control the centering
and bulk-replacement errors in a common balanced norm and convert that bound
into the left and right conclusions of
\cref{lem:one-view-gram-approximation}.

The lower block rows of the singular-vector equations are
\begin{equation}
\label{eq:one-view-lower-equations}
  \bm Z_{21}\bm V_\parallel+\bm Z_{22}\bm V_\perp
  =
  \bm U_\perp\bm\Sigma,
  \qquad
  \bm Z_{12}^\trans\bm U_\parallel+\bm Z_{22}^\trans\bm U_\perp
  =
  \bm V_\perp\bm\Sigma.
\end{equation}
Eliminating $\bm V_\perp$ gives
\begin{equation}
\label{eq:one-view-eliminated}
  \bm U_\perp(\bm\Sigma^2-\bm I_p)
  =
  (\bm Z_{22}\bm Z_{22}^\trans-\bm I_a)\bm U_\perp
  +
  \bm Z_{21}\bm V_\parallel\bm\Sigma
  +
  \bm Z_{22}\bm Z_{12}^\trans\bm U_\parallel.
\end{equation}
To simplify notation in this section, we write
\[
\bm H_{\rm bulk} \coloneqq \bm Z_{22}\bm Z_{22}^\trans-\bm I_a, \qquad \bm D_\Sigma \coloneqq \bm\Sigma^2-\bm I_p.
\]
On $\cE_{\rm loc}$, we have 
\begin{equation}\label{eq:one-view-contraction}
   \opn{\bm H_{\rm bulk}}\opn{\bm D_\Sigma^{-1}}
  \le 
  \frac18.   
\end{equation}

Indeed, \cref{eq:one-view-upper-branch-relation} and the definition of
$\cE_{\rm loc}$ give
$\widehat s_i^2-1\ge\tau_i\ge s_\star^2/16$ for every $i$.  Consequently,
$\bm\Sigma^2-\bm I_p$ is positive definite and $  \opn{\bm D_\Sigma^{-1}}
  =
  \frac{1}{\min_i(\widehat s_i^2-1)}
  \le
  \frac{16}{s_\star^2}.$

Thus, on $\cE_{\rm loc}$, \cref{eq:one-view-eliminated} can be solved by a
convergent Neumann series: 
\begin{equation}
\label{eq:one-view-U-finite-expansion}
  \bm U_\perp
  =
  \sum_{j \geq 0}
    \bm H_{\rm bulk}^j (\bm Z_{21}\bm V_\parallel\bm\Sigma
  +
  \bm Z_{22}\bm Z_{12}^\trans\bm U_\parallel) \bm D_\Sigma^{-j-1}.
\end{equation}
Taking the Gram of this expansion gives the four-term expansion
\begin{align}
\label{eq:one-view-U-Gram-double-expansion}
  \bm U_\perp^\trans\bm U_\perp
  =
  \sum_{r,s\geq 0}
  \bm D_\Sigma^{-r-1}
  \Bigl(&
    \bm\Sigma\bm V_\parallel^\trans
    \bm Z_{21}^\trans\bm H_{\rm bulk}^{r+s}\bm Z_{21}
    \bm V_\parallel\bm\Sigma
    \notag\\
  &+
    \bm\Sigma\bm V_\parallel^\trans
    \bm Z_{21}^\trans\bm H_{\rm bulk}^{r+s}\bm Z_{22}
    \bm Z_{12}^\trans\bm U_\parallel
    \notag\\
  &+
    \bm U_\parallel^\trans\bm Z_{12}\bm Z_{22}^\trans
    \bm H_{\rm bulk}^{r+s}\bm Z_{21}
    \bm V_\parallel\bm\Sigma
    \notag\\
  &+
    \bm U_\parallel^\trans\bm Z_{12}\bm Z_{22}^\trans
    \bm H_{\rm bulk}^{r+s}\bm Z_{22}\bm Z_{12}^\trans
    \bm U_\parallel
  \Bigr)
  \bm D_\Sigma^{-s-1}.
\end{align}
Note that~\cref{eq:one-view-U-Gram-double-expansion} can be viewed as an equation linking the left leakage Gram matrix $\bm U_\perp^\trans\bm U_\perp$ with the two in-signal components $\bm U_\parallel$ and $\bm V_\parallel$. 
While $\bm D_\Sigma$ and $\bm \Sigma$ are simple diagonal matrices, it is the terms like  $\bm Z_{21}^\trans\bm H_{\rm bulk}^{r+s}\bm Z_{21}$ that is problematic as they depend on the noise matrices in a complicated way. 

In what follows, we replace the terms like $\bm Z_{21}^\trans\bm H_{\rm bulk}^{r+s}\bm Z_{21}$ by simpler deterministic matrices in two steps: 
in the first step, we remove the non-bulk noises $\bm Z_{12}, \bm Z_{21}$ by centering them, while in the second step, we remove the bulk noise $\bm Z_{22}$.

\subsubsection{Centering the non-bulk noise}
The Gram expansion above contains four Gaussian coefficient matrices.
The goal of this subsection is to split each coefficient into its
conditional mean given $\bm Z_{22}$ and a centered fluctuation.  The conditional
means will form the empirical Gram operators, while all centered
fluctuations will be collected into a remainder.

For $m\ge0$, define the centered Gaussian coefficient matrices
\begin{subequations}
\label{eq:one-view-U-centered-coefficients}
\begin{align}
  \bm C_{U,11}^{(m)}
  &\coloneqq
  \bm Z_{21}^\trans\bm H_{\rm bulk}^m\bm Z_{21}
  -
  \frac{\tr(\bm H_{\rm bulk}^m)}{b}\bm I_p,
  \label{eq:one-view-U-core-11}\\
  \bm C_{U,22}^{(m)}
  &\coloneqq
  \bm Z_{12}\bm Z_{22}^\trans\bm H_{\rm bulk}^m\bm Z_{22}\bm Z_{12}^\trans
  -
  \frac{\tr(\bm Z_{22}^\trans\bm H_{\rm bulk}^m\bm Z_{22})}{b}\bm I_p,
  \label{eq:one-view-U-core-22}\\
  \bm C_{U,12}^{(m)}
  &\coloneqq
  \bm Z_{21}^\trans\bm H_{\rm bulk}^m\bm Z_{22}\bm Z_{12}^\trans,
  \label{eq:one-view-U-core-12}\\
  \bm C_{U,21}^{(m)}
  &\coloneqq
  \bm Z_{12}\bm Z_{22}^\trans\bm H_{\rm bulk}^m\bm Z_{21}.
  \label{eq:one-view-U-core-21}
\end{align}
\end{subequations}
Conditional on $\bm Z_{22}$, all four matrices are centered.  Collecting the
corresponding conditional means in
\cref{eq:one-view-U-Gram-double-expansion} motivates the empirical operators
\begin{subequations}
\begin{align}
  \widehat{\mathcal A}_U(\bm M)
  &\coloneqq
  \sum_{r,s\ge0}
  \frac{\tr(\bm H_{\rm bulk}^{r+s})}{b}\,
  \bm D_\Sigma^{-r-1}
  \bm\Sigma\bm M\bm\Sigma
  \bm D_\Sigma^{-s-1},
  \label{eq:one-view-empirical-AU}\\
  \widehat{\mathcal B}(\bm M)
  &\coloneqq
  \sum_{r,s\ge0}
  \frac{\tr(\bm Z_{22}^\trans\bm H_{\rm bulk}^{r+s}\bm Z_{22})}{b}\,
  \bm D_\Sigma^{-r-1}\bm M\bm D_\Sigma^{-s-1}.
  \label{eq:one-view-empirical-B}
\end{align}
\end{subequations}
With these definitions, the left leakage Gram admits the exact decomposition
\begin{equation}
\label{eq:one-view-empirical-U-Gram}
  {\bm U_\perp^\trans\bm U_\perp
  =
  \widehat{\mathcal A}_U
  \bigl(
    \bm V_\parallel^\trans\bm V_\parallel
  \bigr)
  +
  \widehat{\mathcal B}
  \bigl(
    \bm U_\parallel^\trans\bm U_\parallel
  \bigr)
  +
  \bm R_{U,\mathrm{coef}},}
\end{equation}
where
\begin{equation}
\label{eq:one-view-RU-coef}
\begin{aligned}
  \bm R_{U,\mathrm{coef}}
  \coloneqq{}&
  \sum_{r,s\ge0}
  \bm D_\Sigma^{-r-1}
  \Bigl\{
    \bm\Sigma\bm V_\parallel^\trans
    \bm C_{U,11}^{(r+s)}
    \bm V_\parallel\bm\Sigma
    +
    \bm U_\parallel^\trans
    \bm C_{U,22}^{(r+s)}
    \bm U_\parallel
  \\
  &\qquad\qquad
    +
    \bm\Sigma\bm V_\parallel^\trans
    \bm C_{U,12}^{(r+s)}
    \bm U_\parallel
    +
    \bm U_\parallel^\trans
    \bm C_{U,21}^{(r+s)}
    \bm V_\parallel\bm\Sigma
  \Bigr\}
  \bm D_\Sigma^{-s-1}.
\end{aligned}
\end{equation}

\paragraph{The right subspace leakage.}
We next carry out the same decomposition for
$\bm V_\perp^\trans\bm V_\perp$.  The only additional bookkeeping comes from
the base term
$\bm Z_{12}^\trans\bm U_\parallel\bm\Sigma^{-1}$ in the expansion of
$\bm V_\perp$.  We keep this term separate and, for $q\ge0$, set
\(\bm G_q \coloneqq \bm Z_{22}^\trans\bm H_{\rm bulk}^q\bm Z_{22}$ and $\bm E_q \coloneqq \bm D_\Sigma^{-q-1}\bm\Sigma^{-1}.\)
Substituting the convergent expansion of $\bm U_\perp$ into the second
equation in \cref{eq:one-view-lower-equations} gives
\begin{equation}
\label{eq:one-view-V-series}
  \bm V_\perp
  =
  \sum_{r\ge0}
    \bm Z_{22}^\trans\bm H_{\rm bulk}^r
    \bm Z_{21}\bm V_\parallel
    \bm D_\Sigma^{-r-1}
  +
  \bm Z_{12}^\trans\bm U_\parallel\bm\Sigma^{-1}
  +
  \sum_{q \geq 0}
    \bm G_q\bm Z_{12}^\trans
    \bm U_\parallel\bm E_q.
\end{equation}
We now center the four families of quadratic coefficients arising when the Gram of
\cref{eq:one-view-V-series} is expanded.

For $r,s,q,u\ge0$, define four centered terms as 
\begin{subequations}
\label{eq:one-view-V-centered-coefficients}
\begin{align}
  \bm C_{V,11}^{(r,s)}
  &\coloneqq
  \bm Z_{21}^\trans
  \bm H_{\rm bulk}^r(\bm I_a+\bm H_{\rm bulk})\bm H_{\rm bulk}^s
  \bm Z_{21}
  -
  \frac{
    \tr\{\bm H_{\rm bulk}^r(\bm I_a+\bm H_{\rm bulk})\bm H_{\rm bulk}^s\}
  }{b}\bm I_p,
  \label{eq:one-view-V-core-11}\\
  \bm C_{V,22}^{(q,u)}
  &\coloneqq
  \bm Z_{12}\bm G_q\bm G_u\bm Z_{12}^\trans
  -
  \frac{\tr(\bm G_q\bm G_u)}{b}\bm I_p,
  \label{eq:one-view-V-core-22}\\
  \bm C_{V,12}^{(r,u)}
  &\coloneqq
  \bm Z_{21}^\trans
  \bm H_{\rm bulk}^r\bm Z_{22}\bm G_u\bm Z_{12}^\trans,
  \label{eq:one-view-V-core-12}\\
  \bm C_{V,21}^{(q,s)}
  &\coloneqq
  \bm Z_{12}\bm G_q\bm Z_{22}^\trans
  \bm H_{\rm bulk}^s\bm Z_{21}.
  \label{eq:one-view-V-core-21}
\end{align}
\end{subequations}
Again, these matrices are centered conditional on $\bm Z_{22}$.  Collecting their
conditional means leads to the remaining empirical operator
\begin{equation}
\label{eq:one-view-empirical-AV}
\begin{aligned}
  \widehat{\mathcal A}_V(\bm M)
  \coloneqq{}&
  \bm\Sigma^{-1}\bm M\bm\Sigma^{-1}
  +
  \sum_{q\ge0}
  \frac{\tr(\bm G_q)}{b}
  \left(
    \bm E_q\bm M\bm\Sigma^{-1}
    +
    \bm\Sigma^{-1}\bm M\bm E_q
  \right)
  \\
  &+
  \sum_{q,u\ge0}
  \frac{\tr(\bm G_q\bm G_u)}{b}\,
  \bm E_q\bm M\bm E_u.
\end{aligned}
\end{equation}
Expanding the Gram in \cref{eq:one-view-V-series} and collecting the centered
terms therefore gives
\begin{equation}
\label{eq:one-view-empirical-V-Gram}
  {\bm V_\perp^\trans\bm V_\perp
  =
  \widehat{\mathcal B}
  \bigl(
    \bm V_\parallel^\trans\bm V_\parallel
  \bigr)
  +
  \widehat{\mathcal A}_V
  \bigl(
    \bm U_\parallel^\trans\bm U_\parallel
  \bigr)
  +
  \bm R_{V,\mathrm{coef}},}
\end{equation}
where
\begin{align}
\label{eq:one-view-RV-coef}
  \bm R_{V,\mathrm{coef}}
  \coloneqq{}&
  \sum_{r,s\ge0}
    \bm D_\Sigma^{-r-1}
    \bm V_\parallel^\trans
    \bm C_{V,11}^{(r,s)}
    \bm V_\parallel
    \bm D_\Sigma^{-s-1}
  +
  \sum_{q,u\ge0}
    \bm E_q
    \bm U_\parallel^\trans
    \bm C_{V,22}^{(q,u)}
    \bm U_\parallel
    \bm E_u
  \notag\\
  &+
  \sum_{q\ge0}
  \bm E_q\bm U_\parallel^\trans
  \left(
    \bm Z_{12}\bm G_q\bm Z_{12}^\trans
    -
    \frac{\tr(\bm G_q)}{b}\bm I_p
  \right)
  \bm U_\parallel\bm\Sigma^{-1}
  \notag\\
  &+
  \sum_{u\ge0}
  \bm\Sigma^{-1}\bm U_\parallel^\trans
  \left(
    \bm Z_{12}\bm G_u\bm Z_{12}^\trans
    -
    \frac{\tr(\bm G_u)}{b}\bm I_p
  \right)
  \bm U_\parallel\bm E_u
  \notag\\
  &+
  \bm\Sigma^{-1}\bm U_\parallel^\trans
  \left(
    \bm Z_{12}\bm Z_{12}^\trans-\bm I_p
  \right)
  \bm U_\parallel\bm\Sigma^{-1}
  \notag\\
  &+
  \sum_{r,u\ge0}
    \bm D_\Sigma^{-r-1}
    \bm V_\parallel^\trans
    \bm C_{V,12}^{(r,u)}
    \bm U_\parallel
    \bm E_u
  +
  \sum_{r\ge0}
    \bm D_\Sigma^{-r-1}
    \bm V_\parallel^\trans
    \bm Z_{21}^\trans\bm H_{\rm bulk}^r
    \bm Z_{22}\bm Z_{12}^\trans
    \bm U_\parallel\bm\Sigma^{-1}
  \notag\\
  &+
  \sum_{q,s\ge0}
    \bm E_q
    \bm U_\parallel^\trans
    \bm C_{V,21}^{(q,s)}
    \bm V_\parallel
    \bm D_\Sigma^{-s-1}
  \notag\\
  &+
  \sum_{s\ge0}
    \bm\Sigma^{-1}\bm U_\parallel^\trans
    \bm Z_{12}\bm Z_{22}^\trans
    \bm H_{\rm bulk}^s\bm Z_{21}
    \bm V_\parallel\bm D_\Sigma^{-s-1}.
\end{align}

\subsubsection{Replacing the bulk traces}

The preceding exact Gram equations have isolated the centered Gaussian
coefficient fluctuations in
$\bm R_{U,\mathrm{coef}}$ and $\bm R_{V,\mathrm{coef}}$.
The remaining randomness in the empirical operators enters only through
normalized traces of functions of the bulk matrix $\bm Z_{22}$.  The goal of this
subsection is to replace these random traces by their deterministic high-dimensional limits and to isolate the resulting replacement error.

Instead of working with the trace of bulk noise such as $\tr(\bm H_{\rm bulk}^{m})$ in~\cref{eq:one-view-empirical-AU} separately for each order $m$, we rely on the following resolvent 
\[
\widehat m_{\rm bulk}(z) \coloneqq \frac1a \tr(\bm I_a-z\bm H_{\rm bulk})^{-1},
\]
which is defined on the contour
\[
\mathcal C_\star
\coloneqq
\{z\in\C:|z|=32/s_\star^2\}.
\]
Also set
$\bm K_0(z) \coloneqq \bm D_\Sigma^{-1} \left( \bm I_p-z^{-1}\bm D_\Sigma^{-1} \right)^{-1},$ and $ \bm K_1(z) \coloneqq (1+z^{-1})\bm K_0(z).$
On $\cE_{\rm loc}$,
$\|z^{-1}\bm D_\Sigma^{-1}\|_{\mathrm{op}}\le1/2$, so these matrices admit
convergent geometric expansions.
These matrices package the trace coefficients appearing in the empirical
operators.  Cauchy's coefficient formula yields the following contour
representations:
\begin{subequations}
\label{eq:one-view-empirical-contours}
\begin{align}
  \widehat{\mathcal A}_U(\bm M)
  &=
  \frac{\beta}{2\pi\mathrm i}
  \oint_{\mathcal C_\star}
  \frac{\widehat m_{\rm bulk}(z)}{z}\,
  \bm\Sigma\bm K_0(z)
  \bm M
  \bm K_0(z)\bm\Sigma\,dz,
  \label{eq:one-view-empirical-AU-contour}\\
  \widehat{\mathcal B}(\bm M)
  &=
  \frac{\beta}{2\pi\mathrm i}
  \oint_{\mathcal C_\star}
  \frac{\widehat m_{\rm bulk}(z)}{z}\,
  \bm K_1(z)\bm M\bm K_0(z)\,dz,
  \label{eq:one-view-empirical-B-contour}\\
  \widehat{\mathcal A}_V(\bm M)
  &=
  \bm\Sigma^{-1}\bm M\bm\Sigma^{-1}
  +
  \frac{\beta}{2\pi\mathrm i}
  \oint_{\mathcal C_\star}
  \frac{\widehat m_{\rm bulk}(z)}{z}\,
  \bm\Sigma^{-1}
  \bigl(
    \bm K_1(z)\bm M
    +
    \bm M\bm K_1(z)
    +
    \bm K_1(z)\bm M\bm K_1(z)
  \bigr)
  \bm\Sigma^{-1}\,dz.
  \label{eq:one-view-empirical-AV-contour}
\end{align}
\end{subequations}

We next define the deterministic object that replaces these traces.
Let $m_\beta$ be the unique solution analytic at zero, with
$m_\beta(0)=1$, of
\begin{equation}
\label{eq:one-view-MP-equation}
  m_\beta(z)
  =
  1
  +
  \beta z\,m_\beta(z)
  \bigl\{
    (1+z)m_\beta(z)-1
  \bigr\}.
\end{equation}
Equivalently, $m_\beta$ is the moment transform of the centered
Marchenko--Pastur law with aspect ratio~$\beta$.  

Replacing $\widehat m_{\rm bulk}$ by $m_\beta$ in the contour formulas above produces
deterministic candidate operators.  The next proposition identifies these
operators with the population operators
$\mathcal A_U$, $\mathcal B$, and $\mathcal A_V$ defined earlier in
\cref{eq:one-view-population-operators}.

\begin{proposition}\label{prop:population-operator}
    Recall the population operators defined in \cref{eq:one-view-population-operators}. On the event $\cE_{\rm loc}$, we have 
    \begin{subequations}
\label{eq:one-view-population-contours}
\begin{align}
  \mathcal A_U(\bm M)
  &=
  \frac{\beta}{2\pi\mathrm i}
  \oint_{\mathcal C_\star}
  \frac{m_\beta(z)}{z}\,
  \bm\Sigma\bm K_0(z)
  \bm M
  \bm K_0(z)\bm\Sigma\,dz,
  \label{eq:one-view-population-AU-contour}\\
  \mathcal B(\bm M)
  &=
  \frac{\beta}{2\pi\mathrm i}
  \oint_{\mathcal C_\star}
  \frac{m_\beta(z)}{z}\,
  \bm K_1(z)\bm M\bm K_0(z)\,dz,
  \label{eq:one-view-population-B-contour}\\
  \mathcal A_V(\bm M)
  &=
  \bm\Sigma^{-1}\bm M\bm\Sigma^{-1}
  +
  \frac{\beta}{2\pi\mathrm i}
  \oint_{\mathcal C_\star}
  \frac{m_\beta(z)}{z}\,
  \bm\Sigma^{-1}
  \bigl(
    \bm K_1(z)\bm M
    +
    \bm M\bm K_1(z)
    +
    \bm K_1(z)\bm M\bm K_1(z)
  \bigr)
  \bm\Sigma^{-1}\,dz.
  \label{eq:one-view-population-AV-contour}
\end{align}
\end{subequations}
\end{proposition}

Thus the difference between the exact empirical Gram operators and their
population counterparts is entirely due to the replacement
$\widehat m_{\rm bulk}\mapsto m_\beta$.  We collect this difference in the bulk
remainders
\begin{subequations}
\label{eq:one-view-bulk-residuals}
\begin{align}
  \bm R_{U,\mathrm{bulk}}
  &\coloneqq
  (\widehat{\mathcal A}_U-\mathcal A_U)
  \bigl(
    \bm V_\parallel^\trans\bm V_\parallel
  \bigr)
  +
  (\widehat{\mathcal B}-\mathcal B)
  \bigl(
    \bm U_\parallel^\trans\bm U_\parallel
  \bigr),
  \label{eq:one-view-RU-bulk}\\
  \bm R_{V,\mathrm{bulk}}
  &\coloneqq
  (\widehat{\mathcal B}-\mathcal B)
  \bigl(
    \bm V_\parallel^\trans\bm V_\parallel
  \bigr)
  +
  (\widehat{\mathcal A}_V-\mathcal A_V)
  \bigl(
    \bm U_\parallel^\trans\bm U_\parallel
  \bigr).
  \label{eq:one-view-RV-bulk}
\end{align}
\end{subequations}
Combining these bulk remainders with the exact coefficient-centered Gram
equations above, and recalling the definitions of
$\bm\varepsilon_U$ and $\bm\varepsilon_V$, we obtain the desired
decomposition
\[
{\bm\varepsilon_U = \bm R_{U,\mathrm{coef}} + \bm R_{U,\mathrm{bulk}}, \qquad \bm\varepsilon_V = \bm R_{V,\mathrm{coef}} + \bm R_{V,\mathrm{bulk}}.}
\]

\subsubsection{Bounding the two errors}

It remains to control these two sources of error
separately.  We first bound the centered Gaussian coefficient fluctuations,
and then the error from replacing the empirical bulk traces by their
high-dimensional limits.

It is convenient to control both errors first in a common balanced norm.
Introduce
\[
\bm W_{\rm bal} \coloneqq \diag\bigl( \beta^{1/4}\sqrt{\tau_i} \bigr).
\]
We will obtain bounds for
$\bm W_L\bm\varepsilon_U\bm W_L$ and
$\bm W_{\rm bal}\bm\varepsilon_V\bm W_{\rm bal}$.
At the end of the subsection, deterministic comparisons between
$\bm W_{\rm bal}$, $\bm W_L$, and $\bm W_R$ will convert this balanced estimate
into the two bounds stated in \cref{lem:one-view-gram-approximation}.

\paragraph{Step 1: Bound the centering error. }

We first control the remainders
$\bm R_{U,\mathrm{coef}}$ and $\bm R_{V,\mathrm{coef}}$, which collect the
centered Gaussian coefficient fluctuations from the exact Gram expansions.

\begin{lemma}
\label{lem:one-view-coefficient-error}
On $\cE_{\rm loc}$,
\begin{equation}
\label{eq:one-view-coefficient-error-master}
  \Lqn{
    \ind_{\cE_{\rm loc}}
    \bm W_L\bm R_{U,\mathrm{coef}}\bm W_L
  }{32}
  +
  \Lqn{
    \ind_{\cE_{\rm loc}}
    \bm W_{\rm bal}\bm R_{V,\mathrm{coef}}\bm W_{\rm bal}
  }{32}
  \le
  C\beta \sqrt{\frac{p}{a}}.
\end{equation}
\end{lemma}

\paragraph{Step 2: Bound the bulk replacement error}

We next control
$\bm R_{U,\mathrm{bulk}}$ and $\bm R_{V,\mathrm{bulk}}$,
which arise from replacing the empirical bulk transform
$\widehat m_{\rm bulk}$ by its deterministic counterpart $m_\beta$.

\begin{lemma}
\label{lem:one-view-resolvent}
Under the condition that 
$s_\star^2\wedge s_\star^4\ge\nu^2\beta$,
\begin{equation}
\label{eq:one-view-resolvent-bound}
  \sup_{z\in\mathcal C_\star}
  \left\|
    \ind_{\cE_{\rm loc}}
    \bigl(\widehat m_{\rm bulk}(z)-m_\beta(z)\bigr)
  \right\|_{L^{32}}
  \le
  \frac{C}{\sqrt a}.
\end{equation}
\end{lemma}

We use this estimate to control the two bulk remainders in the same
balanced norms as in the preceding subsection.

\begin{lemma}[Bulk replacement remainder]
\label{lem:one-view-bulk-error}
Under the condition that 
$s_\star^2\wedge s_\star^4\ge\nu^2\beta$,
\begin{equation}
\label{eq:one-view-bulk-error-master}
  \Lqn{
    \ind_{\cE_{\rm loc}}
    \bm W_L\bm R_{U,\mathrm{bulk}}\bm W_L
  }{32}
  +
  \Lqn{
    \ind_{\cE_{\rm loc}}
    \bm W_{\rm bal}\bm R_{V,\mathrm{bulk}}\bm W_{\rm bal}
  }{32}
  \le
  C\frac{\sqrt a}{b}.
\end{equation}
\end{lemma}

\paragraph{Step 3: Completing the Proof.}

We can now return to
$\bm\varepsilon_U$ and $\bm\varepsilon_V$.
Since
\(\bm\varepsilon_U = \bm R_{U,\mathrm{coef}} + \bm R_{U,\mathrm{bulk}}, \qquad \bm\varepsilon_V = \bm R_{V,\mathrm{coef}} + \bm R_{V,\mathrm{bulk}},\)
and $\sqrt a/b\le\sqrt{ap}/b$,
\cref{eq:one-view-coefficient-error-master,eq:one-view-bulk-error-master}
yield the balanced estimate
\begin{equation}
\label{eq:one-view-balanced-master-error}
  \Lqn{
    \ind_{\cE_{\rm loc}}
    \bm W_L\bm\varepsilon_U\bm W_L
  }{32}
  +
  \Lqn{
    \ind_{\cE_{\rm loc}}
    \bm W_{\rm bal}\bm\varepsilon_V\bm W_{\rm bal}
  }{32}
  \le
  C\frac{\sqrt{ap}}b.
\end{equation}

It remains only to translate this common balanced estimate into the two
forms stated in \cref{lem:one-view-gram-approximation}.

Under the left signal condition, since
$(\bm W_{\rm bal})_{ii}^2 = \sqrt\beta\,\tau_i \ge \beta,$
we have
$\beta\|\bm\varepsilon_V\|_{\mathrm{op}} \le \| \bm W_{\rm bal}\bm\varepsilon_V\bm W_{\rm bal} \|_{\mathrm{op}}.$
Together with \cref{eq:one-view-balanced-master-error}, this proves
\cref{eq:one-view-left-error-bound}.

Under the stronger right signal condition,
$\tau_i\ge s_\star^2/16\ge\nu^2/16$.  The definitions of the three weights
then give
\(\|\bm W_L^{-1}\|_{\mathrm{op}}\le C\) and
\(\|\bm W_R\bm W_{\rm bal}^{-1}\|_{\mathrm{op}}^2
\le C\beta^{-1/2}.\)
Consequently, we have
\[
\begin{aligned}
  &
  \Lqn{
    \ind_{\cE_{\rm loc}}
    \bm W_R\bm\varepsilon_V\bm W_R
  }{32}
  +
  \Lqn{
    \ind_{\cE_{\rm loc}}\bm\varepsilon_U
  }{32}\le
  C\beta^{-1/2}
  \frac{\sqrt{ap}}b
  \le
  C\sqrt{\frac pb},
\end{aligned}
\]
where the final step uses $a=\beta b$.
This proves \cref{eq:one-view-right-error-bound}.

\subsection{Proof of \cref{lem:one-view-radial}}
\label{app:proof-one-view-radial}

We first control $\|\bm U_\perp\bm W_L\|_{\mathrm{op}}$.
Partition the coordinates into
\begin{subequations}
\label{eq:one-view-radial-index-sets}
\begin{align}
  \cI_0
  &\coloneqq
  \{i:\tau_i=\sqrt\beta\},
  \\
  \cI_{\rm h}
  &\coloneqq
  \left\{
    i:\tau_i>\sqrt\beta,\;
    \widehat s_i^2-1
    \ge
    2\opn{\bm Z_{22}\bm Z_{22}^\trans-\bm I_a}
  \right\},
  \\
  \cI_{\rm \ell}
  &\coloneqq
  \left\{
    i:\tau_i>\sqrt\beta,\;
    \widehat s_i^2-1
    <
    2\opn{\bm Z_{22}\bm Z_{22}^\trans-\bm I_a}
  \right\}.
\end{align}
\end{subequations}
For $\cI\subseteq[p]$, use the corresponding subscript to denote the
associated column submatrices, and write
$\bm\Sigma_{\cI}=\diag(\widehat s_i:i\in\cI)$ and
$\bm W_{L,\cI}=\diag(\ell_i:i\in\cI)$.

\paragraph{Frozen coordinates.}
On $\cI_0$, $\ell_i=\sqrt\beta$, and hence
$\|\bm U_{\perp,\cI_0}\bm W_{L,\cI_0}\|_{\mathrm{op}}
\le\sqrt\beta$.

\paragraph{High coordinates.}
Restrict \cref{eq:one-view-eliminated} to $\cI_{\rm h}$.
By definition of this set,
$\|\bm Z_{22}\bm Z_{22}^\trans-\bm I_a\|_{\mathrm{op}}
\|(\bm\Sigma_{\cI_{\rm h}}^2-\bm I)^{-1}\|_{\mathrm{op}}
\le1/2$.
Moreover, \cref{eq:one-view-upper-branch-relation} and the definition of
$\ell_i$ give
$\widehat s_i\ell_i/(\widehat s_i^2-1)
=(\tau_i+\beta)/(\tau_i+\beta+\beta/\tau_i)\le1$, and therefore also
$\ell_i/(\widehat s_i^2-1)\le1$.
Multiplying the restricted equation by
$(\bm\Sigma_{\cI_{\rm h}}^2-\bm I)^{-1}\bm W_{L,\cI_{\rm h}}$
and taking operator norms gives
\[
  \opn{\bm U_{\perp,\cI_{\rm h}}\bm W_{L,\cI_{\rm h}}}
  \le
  2\left(
    \opn{\bm Z_{21}}
    +
    \opn{\bm Z_{22}\bm Z_{12}^\trans}
  \right).
\]

\paragraph{Low coordinates.}
Set
$\bm Z_{\rm low}\coloneqq[\bm Z_{21},\bm Z_{22}]$.
Restricting the first equation in
\cref{eq:one-view-lower-equations} to $\cI_{\rm \ell}$ and using the
orthonormality of the stacked right singular vectors gives
$\|\bm U_{\perp,\cI_{\rm \ell}}\bm W_{L,\cI_{\rm \ell}}\|_{\mathrm{op}}
\le
\|\bm Z_{\rm low}\|_{\mathrm{op}}
\max_{i\in\cI_{\rm \ell}}\ell_i/\widehat s_i$.
The same identities give
$\ell_i/\widehat s_i=\tau_i/(\tau_i+1)\le\tau_i\le \widehat s_i^2-1$; therefore, the
definition of $\cI_{\rm \ell}$ yields
\[
  \opn{\bm U_{\perp,\cI_{\rm \ell}}\bm W_{L,\cI_{\rm \ell}}}
  \le
  2
  \opn{\bm Z_{22}\bm Z_{22}^\trans-\bm I_a}
  \opn{\bm Z_{\rm low}}.
\]

\paragraph{Combining the three regimes.}
The preceding bounds give
\begin{equation}
\label{eq:one-view-left-radial-envelope}
  \opn{\bm U_\perp\bm W_L}
  \lesssim
  \sqrt\beta
  +
  \opn{\bm Z_{21}}
  +
  \opn{\bm Z_{22}\bm Z_{12}^\trans}
  +
  \opn{\bm Z_{22}\bm Z_{22}^\trans-\bm I_a}
  \opn{\bm Z_{\rm low}}.
\end{equation}
By \cref{eq:gaussian-matrix-op-moment},
$\|\bm Z_{21}\|_{L^{128}}\lesssim\sqrt\beta$ and
$\|\bm Z_{\rm low}\|_{L^{256}}\lesssim1$.
Applying \cref{eq:gaussian-wishart-moment} to the stacked Gaussian matrix
formed from $\bm Z_{12}$ and $\bm Z_{22}$ also gives
$\|\bm Z_{22}\bm Z_{12}^\trans\|_{L^{128}}\lesssim\sqrt\beta$ and
$\|\bm Z_{22}\bm Z_{22}^\trans-\bm I_a\|_{L^{256}}
\lesssim\sqrt\beta$.
Thus \cref{eq:one-view-left-radial-envelope} and H\"older's inequality
yield
$\|\bm U_\perp\bm W_L\|_{L^{128}}\lesssim\sqrt\beta$.

\paragraph{Right leakage.}
Let $\cI_{\rm up}\coloneqq\{i:\tau_i>\sqrt\beta\}$.
On $\cI_0$, $r_i=1$, so
$\|\bm V_{\perp,\cI_0}\bm W_{R,\cI_0}\|_{\mathrm{op}}\le1$.
On $\cI_{\rm up}$, restrict the second equation in
\cref{eq:one-view-lower-equations}.  By
\cref{eq:one-view-upper-branch-relation} and the definition of $r_i$,
we have $r_i/\widehat s_i=\tau_i/(\tau_i+\beta)\le1$, and hence
$\|\bm V_{\perp,\cI_{\rm up}}\bm W_{R,\cI_{\rm up}}\|_{\mathrm{op}}
\le
\|\bm Z_{12}\|_{\mathrm{op}}+\|\bm Z_{22}\|_{\mathrm{op}}$.
Consequently,
$\|\bm V_\perp\bm W_R\|_{\mathrm{op}}
\le
1+\|\bm Z_{12}\|_{\mathrm{op}}+\|\bm Z_{22}\|_{\mathrm{op}}$,
and \cref{eq:gaussian-matrix-op-moment} gives
$\|\bm V_\perp\bm W_R\|_{L^{128}}\lesssim1$.

\subsection{Proofs of the Gram-approximation inputs}
\subsubsection{Proof of~\cref{prop:population-operator}}
Since $\bm\Sigma$, $\bm{K}_0(z)$, and $\bm{K}_1(z)$ are diagonal, it
suffices to evaluate the contour formulas entrywise.

\paragraph{Pole locations and scalar identities.}
On $\cE_{\rm loc}$, we have $\tau_i>\sqrt\beta$, so the upper-branch
relation gives
$\widehat s_i^2-1=\tau_i+\beta+\beta/\tau_i$.  Set
$\zeta_i\coloneqq(\widehat s_i^2-1)^{-1}$.  Then
$\zeta_i\le16/s_\star^2$, so $\zeta_i$ lies inside
$\mathcal C_\star$.  The standing signal condition also ensures that
$m_\beta$ is analytic on and inside this contour.

Substituting $z=\zeta_i$ into \cref{eq:one-view-MP-equation} shows that
the solution analytic at zero satisfies
$m_\beta(\zeta_i)=(\widehat s_i^2-1)/(\tau_i+\beta)$.  Thus, writing
$k_{0,i}(z)$ and $k_{1,i}(z)$ for the $i$th diagonal entries of
$\bm{K}_0(z)$ and $\bm{K}_1(z)$, respectively,
\begin{equation}
\label{eq:one-view-population-contour-scalar-identities}
\begin{gathered}
  \zeta_i m_\beta(\zeta_i)
  =
  \frac{1}{\tau_i+\beta},
  \qquad
  (1+\zeta_i)m_\beta(\zeta_i)
  =
  \frac{\tau_i+1}{\tau_i},
  \\
  k_{0,i}(z)
  =
  \frac{\zeta_i z}{z-\zeta_i},
  \qquad
  k_{1,i}(z)
  =
  \frac{\zeta_i(1+z)}{z-\zeta_i}.
\end{gathered}
\end{equation}

\paragraph{The operators $\mathcal A_U$ and $\mathcal B$.}
Fix $i,j\in[p]$ and first suppose that $\zeta_i\ne\zeta_j$.  The two
integrands below have poles only at $\zeta_i$ and $\zeta_j$.  Using
\cref{eq:one-view-population-contour-scalar-identities}, the residue theorem
gives
\begin{align*}
  \frac{1}{2\pi\mathrm i}
  \oint_{\mathcal C_\star}
  \frac{m_\beta(z)}{z}\,
  k_{0,i}(z)k_{0,j}(z)\,dz
  &=
  \zeta_i\zeta_j
  \frac{
    \zeta_i m_\beta(\zeta_i)
    -
    \zeta_j m_\beta(\zeta_j)
  }{
    \zeta_i-\zeta_j
  } \\
  &=
  \frac{
    \tau_i\tau_j
  }{
    (\tau_i+\beta)(\tau_j+\beta)
    (\tau_i\tau_j-\beta)
  },
  \\
  \frac{1}{2\pi\mathrm i}
  \oint_{\mathcal C_\star}
  \frac{m_\beta(z)}{z}\,
  k_{1,i}(z)k_{0,j}(z)\,dz
  &=
  \zeta_i\zeta_j
  \frac{
    (1+\zeta_i)m_\beta(\zeta_i)
    -
    (1+\zeta_j)m_\beta(\zeta_j)
  }{
    \zeta_i-\zeta_j
  } \\
  &=
  \frac{1}{\tau_i\tau_j-\beta}.
\end{align*}
The same identities hold when $\zeta_i=\zeta_j$ by continuity, or
equivalently by evaluating the resulting double pole.

The $(i,j)$ entry of the right-hand side of
\cref{eq:one-view-population-AU-contour} is therefore
$\beta \widehat s_i\widehat s_j\tau_i\tau_j M_{ij}/
\{(\tau_i+\beta)(\tau_j+\beta)(\tau_i\tau_j-\beta)\}$.
Since the weight definitions give
$r_i/\widehat s_i=\tau_i/(\tau_i+\beta)$, this equals
$\beta r_ir_jM_{ij}/(\tau_i\tau_j-\beta)$, which is
$[\mathcal A_U(\bm{M})]_{ij}$.  The second residue identity immediately
gives
$[\mathcal B(\bm{M})]_{ij}
=\beta M_{ij}/(\tau_i\tau_j-\beta)$.
This proves
\cref{eq:one-view-population-AU-contour,eq:one-view-population-B-contour}.

\paragraph{The operator $\mathcal A_V$.}
The remaining integrand has poles at
$0,\zeta_i,\zeta_j$.  At zero,
$k_{1,i}(0)=k_{1,j}(0)=-1$, so the residue of
\[
  \frac{m_\beta(z)}{z}
  \bigl\{
    k_{1,i}(z)
    +
    k_{1,j}(z)
    +
    k_{1,i}(z)k_{1,j}(z)
  \bigr\}
\]
is $-1$.  Summing this residue with those at
$\zeta_i$ and $\zeta_j$, and again using
\cref{eq:one-view-population-contour-scalar-identities}, gives
\[
  \frac{1}{2\pi\mathrm i}
  \oint_{\mathcal C_\star}
  \frac{m_\beta(z)}{z}
  \bigl\{
    k_{1,i}(z)
    +
    k_{1,j}(z)
    +
    k_{1,i}(z)k_{1,j}(z)
  \bigr\}\,dz
  =
  \frac{
    \tau_i+\tau_j+\beta+1
  }{
    \tau_i\tau_j-\beta
  }.
\]
As above, the identity extends to $\zeta_i=\zeta_j$ by continuity.
Consequently, the coefficient multiplying $M_{ij}$ in the right-hand side
of \cref{eq:one-view-population-AV-contour} is
\[
  \frac{1}{\widehat s_i\widehat s_j}
  \left\{
    1+
    \beta
    \frac{
      \tau_i+\tau_j+\beta+1
    }{
      \tau_i\tau_j-\beta
    }
  \right\}
  =
  \frac{
    (\tau_i+\beta)(\tau_j+\beta)
  }{
    \widehat s_i\widehat s_j(\tau_i\tau_j-\beta)
  }
  =
  \frac{
    \ell_i\ell_j
  }{
    \tau_i\tau_j-\beta
  },
\]
where the last equality uses the direct identity
$\widehat s_i\ell_i=\tau_i+\beta$.  This is precisely
$[\mathcal A_V(\bm{M})]_{ij}$ and proves
\cref{eq:one-view-population-AV-contour}.

\subsubsection{Proof of~\cref{lem:one-view-coefficient-error}}

\begin{proof}

We first bound the $11$ contribution to
$\bm{R}_{U,\mathrm{coef}}$, namely the term containing
$\bm{C}_{U,11}^{(r+s)}$. 
The remaining terms can be controlled similarly. 
Minkowski's inequality and submultiplicativity
give
\begin{align*}
&\Lqn{
  \ind_{\cE_{\rm loc}}\bm{W}_L
  \sum_{r,s\ge0}
  \bm{D}_\Sigma^{-r-1}\bm\Sigma\bm{V}_\parallel^\trans
  \bm{C}_{U,11}^{(r+s)}
  \bm{V}_\parallel\bm\Sigma\bm{D}_\Sigma^{-s-1}
  \bm{W}_L
}{32}\\
&\quad\le
\sum_{r,s\ge0}
\left\|
  \ind_{\cE_{\rm loc}}
  \opn{\bm{W}_L\bm\Sigma\bm{D}_\Sigma^{-1}}^2
  \opn{\bm{D}_\Sigma^{-1}}^{r+s}
  \opn{\bm{C}_{U,11}^{(r+s)}}
\right\|_{L^{32}}\\
&\quad\le
C\sum_{r,s\ge0}
\left\|
  \ind_{\cE_{\rm loc}}
  \opn{\bm{D}_\Sigma^{-1}}^{r+s}
  \opn{\bm{C}_{U,11}^{(r+s)}}
\right\|_{L^{32}}.
\end{align*}
Here we used $\|\bm V_\parallel\|_{\mathrm{op}}\le1$ in the first
inequality.  For the second, \cref{eq:one-view-upper-branch-relation} and
the definition of $\ell_i$ give
$\widehat s_i\ell_i/(\widehat s_i^2-1)\le1$ on $\cE_{\rm loc}$.

We next bound the centered coefficient itself.  Conditional on
$\bm{Z}_{22}$, apply
\cref{eq:gaussian-coefficient-moments} to $\bm{Z}_{21}$ with
$\bm{A}=\bm{H}_{\rm bulk}^m$ and $q=32$.  Since
$\frn{\bm{H}_{\rm bulk}^m}\le\sqrt a\,\|\bm{H}_{\rm bulk}\|_{\mathrm{op}}^m$, we obtain
\begin{align*}
\left(
  \E\left[
    \opn{\bm{C}_{U,11}^{(m)}}^{32}
    \,\middle|\,\bm{Z}_{22}
  \right]
\right)^{1/32}
&\le
\frac{C}{b}
\left(
  \sqrt p\,\frn{\bm{H}_{\rm bulk}^m}
  +p\,\opn{\bm{H}_{\rm bulk}^m}
\right)\le
C\beta \sqrt{\frac{p}{a}}\opn{\bm{H}_{\rm bulk}}^m.
\end{align*}

It remains to combine the coefficient estimate with the exterior powers of
$\bm D_\Sigma^{-1}$. For $m=0$, the preceding conditional estimate directly
gives the required bound. Suppose that $m\ge1$. If
$\opn{\bm H_{\rm bulk}}=0$, then $\bm H_{\rm bulk}^m=\bze$ and hence
$\bm C_{U,11}^{(m)}=\bze$, so the contribution vanishes. On the complementary
event $\{\opn{\bm H_{\rm bulk}}>0\}$, the contraction bound
\cref{eq:one-view-contraction} gives
\[
\begin{aligned}
  \left\|
    \ind_{\cE_{\rm loc}}
    \opn{\bm D_\Sigma^{-1}}^m
    \opn{\bm C_{U,11}^{(m)}}
  \right\|_{L^{32}}
  &\le
  8^{-m}
  \left\|
    \ind\{\opn{\bm H_{\rm bulk}}>0\}
    \frac{\opn{\bm C_{U,11}^{(m)}}}
         {\opn{\bm H_{\rm bulk}}^m}
  \right\|_{L^{32}} \\
  &\le
  C\beta \sqrt{\frac{p}{a}}\,8^{-m}.
\end{aligned}
\]
The second inequality follows by dividing the preceding conditional
coefficient estimate by the positive, $\bm Z_{22}$-measurable quantity
$\opn{\bm H_{\rm bulk}}^m$ and then integrating over $\bm Z_{22}$. Thus the
displayed bound holds for every $m\ge0$ without requiring a convention for
division by zero.

All remaining contributions, except the explicitly displayed base--base
term involving $\bm{Z}_{12}\bm{Z}_{12}^\trans-\bm{I}_p$, follow from the same
argument, and we omit the repeated details.  Indeed, conditional on
$\bm{Z}_{22}$, each coefficient is a
centered Gaussian quadratic or bilinear form whose middle matrix has rank
at most $a$.  Its operator norm is bounded by
$C(1+\|\bm{Z}_{22}\|_{\mathrm{op}})^4$ times the corresponding power of $\|\bm{H}_{\rm bulk}\|_{\mathrm{op}}$.
The relevant powers are $m$ for the remaining left-Gram coefficients,
$r+s$ for $\bm{C}_{V,11}^{(r,s)}$,
$q+u$ for $\bm{C}_{V,22}^{(q,u)}$, $q$ for either centered base--series
coefficient in the $V,22$ family, $r+u$ for
$\bm{C}_{V,12}^{(r,u)}$, $q+s$ for $\bm{C}_{V,21}^{(q,s)}$, and $r$ or $s$
for the two base cross terms.
Thus \cref{eq:gaussian-coefficient-moments}, the exterior-factor bounds,
and \cref{eq:one-view-contraction} give the same geometrically summable
bound as for the $U,11$ contribution.  Since
$\|(1+\|\bm{Z}_{22}\|_{\mathrm{op}})^4\|_{L^{32}}\le C$ by
\cref{eq:gaussian-matrix-op-moment}, their total contribution is at most
$C\beta \sqrt{\frac{p}{a}}$.

It remains to treat the base--base term.  Using
$\|\bm{U}_\parallel\|_{\mathrm{op}}\le1$,
$\|\bm{W}_{\rm bal}\bm\Sigma^{-1}\|_{\mathrm{op}}\le\beta^{1/4}$, and
\cref{eq:gaussian-wishart-moment}, its weighted contribution is at most
\begin{align*}
\opn{\bm{W}_{\rm bal}\bm\Sigma^{-1}}^2
\left\|\bm{Z}_{12}\bm{Z}_{12}^\trans-\bm{I}_p\right\|_{L^{32}}
&\le
C\sqrt\beta
\sqrt{\frac pb}\le C\beta \sqrt{\frac{p}{a}},
\end{align*}
where the last inequality uses $a=\beta b$.
Combining these bounds proves
\cref{eq:one-view-coefficient-error-master}.
\end{proof}

\subsubsection{Proof of~\cref{lem:one-view-resolvent}}

Fix $z\in\mathcal C_\star$. The resolvent trace is well behaved on
$\cE_{\rm loc}$ but need not be globally Lipschitz as a function of the
Gaussian matrix. We therefore introduce a cutoff that agrees with the
resolvent trace on $\cE_{\rm loc}$. Gaussian concentration will control the
resulting statistic around its mean. We then use Gaussian integration by
parts to show that this mean satisfies the Marchenko--Pastur equation up to
an error of order $a^{-1}$, after which stability of that equation completes
the proof.

\paragraph{A globally Lipschitz statistic.}
For $\bm D\in\mathbb R^{a\times b}$, set
$\bm H_{\bm D}\coloneqq\bm D\bm D^\trans-\bm I_a$. Choose a
$4$-Lipschitz function $\varphi:[0,\infty)\to[0,1]$ that equals one on
$[0,1/2]$ and vanishes on $[3/4,\infty)$, and define
\begin{equation}
\label{eq:cutoff}
  \chi(\bm D)
  \coloneqq
  \varphi\!\left(
    \frac{64}{s_\star^2}\opn{\bm H_{\bm D}}
  \right).
\end{equation}
The cutoff restricts attention to matrices for which the resolvent is
uniformly bounded. Indeed, if $\chi(\bm D)>0$, then
$\opn{\bm H_{\bm D}}<3s_\star^2/256$. Since
$|z|=32/s_\star^2$, the matrix $\bm I_a-z\bm H_{\bm D}$ is invertible and
\begin{equation}
\label{eq:resolvent-bound}
  \opn{
    (\bm I_a-z\bm H_{\bm D})^{-1}
  }
  \le\frac85.
\end{equation}
We may therefore define a function on all of $\mathbb R^{a\times b}$ by
\begin{equation}
\label{eq:F-def}
  F(\bm D)
  \coloneqq
  \begin{cases}
    \displaystyle
    \frac{\chi(\bm D)}{a}
    \tr(\bm I_a-z\bm H_{\bm D})^{-1},
    &\chi(\bm D)>0,\\[6pt]
    0,&\chi(\bm D)=0.
  \end{cases}
\end{equation}
By \cref{eq:one-view-localization-event}, $\chi(\bm Z_{22})=1$ on
$\cE_{\rm loc}$, and hence
$F(\bm Z_{22})=\widehat m_{\rm bulk}(z)$ on this event.

We next verify the Lipschitz estimate needed for Gaussian concentration.
Take $\bm D_0,\bm D_1\in\mathbb R^{a\times b}$ and interpolate along
$\bm D_t=\bm D_0+t(\bm D_1-\bm D_0)$. Where the cutoff changes,
$\opn{\bm D_t}\lesssim\sqrt{1+s_\star^2}$; differentiating
$\opn{\bm D_t\bm D_t^\trans-\bm I_a}$ along this segment therefore gives
\[
  |\chi(\bm D_1)-\chi(\bm D_0)|
  \lesssim
  \frac{\sqrt{1+s_\star^2}}{s_\star^2}
  \frn{\bm D_1-\bm D_0}.
\]
If both cutoffs are positive, the resolvent identity, together with
\cref{eq:resolvent-bound}, gives the same bound for the normalized traces:
\[
  \left|
    \frac1a\tr(\bm I_a-z\bm H_{\bm D_1})^{-1}
    -
    \frac1a\tr(\bm I_a-z\bm H_{\bm D_0})^{-1}
  \right|
  \lesssim
  \frac{\sqrt{1+s_\star^2}}{s_\star^2}
  \frn{\bm D_1-\bm D_0}.
\]
If exactly one cutoff vanishes, the first estimate and
\cref{eq:resolvent-bound} control the difference of the two values of $F$;
if both vanish, the difference is zero. Thus
\begin{equation}
\label{eq:F-Lip-D}
  |F(\bm D_1)-F(\bm D_0)|
  \lesssim
  \frac{\sqrt{1+s_\star^2}}{s_\star^2}
  \frn{\bm D_1-\bm D_0}
\end{equation}
for all $\bm D_0,\bm D_1\in\mathbb R^{a\times b}$.

\paragraph{Fluctuation around the mean.}
Set $\bar m\coloneqq\mathbb EF(\bm Z_{22})$. Since
$F(\bm Z_{22})=\widehat m_{\rm bulk}(z)$ on $\cE_{\rm loc}$,
\begin{equation}
\label{eq:resolvent-fluctuation-bias-split}
  \left\|
    \ind_{\cE_{\rm loc}}
    \bigl(
      \widehat m_{\rm bulk}(z)-m_\beta(z)
    \bigr)
  \right\|_{L^{32}}
  \le
  \|F(\bm Z_{22})-\bar m\|_{L^{32}}
  +
  |\bar m-m_\beta(z)|.
\end{equation}
It remains to bound these two terms.

Let $\bm G=\sqrt b\,\bm Z_{22}$, so that the entries of $\bm G$ are
independent standard Gaussian variables. By \cref{eq:F-Lip-D}, the map
$\bm G\mapsto F(\bm G/\sqrt b)$ has Lipschitz constant at most
\[
  \frac{C\sqrt{1+s_\star^2}}{\sqrt b\,s_\star^2}
  =
  \frac{C}{\sqrt a}
  \left(
    \frac{\beta}{s_\star^4}
    +
    \frac{\beta}{s_\star^2}
  \right)^{1/2}
  \lesssim
  \frac1{\sqrt a},
\]
where the last inequality uses
$s_\star^2\wedge s_\star^4\ge\nu^2\beta$. Applying Gaussian concentration
to the real and imaginary parts of $F$ gives
\begin{equation}
\label{eq:fluctuation}
  \|F(\bm Z_{22})-\bar m\|_{L^{32}}
  \lesssim
  \frac1{\sqrt a}.
\end{equation}

\paragraph{An approximate Marchenko--Pastur equation.}
We next show that $\bar m$ satisfies the equation in
\cref{eq:one-view-MP-equation} up to an error of order $a^{-1}$.
In the remainder of the proof, $\bm D$ has the same distribution as
$\bm Z_{22}$, and we write
\[
  \chi\coloneqq\chi(\bm D),
  \qquad
  \bm R_z\coloneqq(\bm I_a-z\bm H_{\bm D})^{-1},
  \qquad
  m\coloneqq\frac1a\tr(\bm R_z)
\]
on the support of $\chi$. Let $\bm d_\ell$ denote the $\ell$th column of
$\bm D$, and set
$h_{\ell i}\coloneqq\bm d_\ell^\trans\bm R_z\bm e_i$. The product
$\chi h_{\ell i}$, extended by zero away from the support of $\chi$, is
globally Lipschitz by the same argument used for $F$. We may therefore apply
Gaussian integration by parts to it.

The identity
$\bm R_z^{-1}=(1+z)\bm I_a-z\bm D\bm D^\trans$ gives
\begin{equation}
\label{eq:trace-start}
  m
  =
  1-zm
  +
  \frac za
  \sum_{i=1}^a\sum_{\ell=1}^b
  D_{i\ell}h_{\ell i}.
\end{equation}
To evaluate the expectation of the last term, differentiate
$h_{\ell i}$. At every point of differentiability of $\chi$,
\[
  \partial_{i\ell}h_{\ell i}
  =
  (\bm R_z)_{ii}
  +
  zh_{\ell i}^2
  +
  z
  (\bm d_\ell^\trans\bm R_z\bm d_\ell)
  (\bm R_z)_{ii}.
\]
The three sums produced by this identity reduce to
\begin{equation}
\label{eq:resolvent-derivative-sums}
\begin{aligned}
  \sum_{i,\ell}(\bm R_z)_{ii}
  &=abm,\\
  \sum_{i,\ell}h_{\ell i}^2
  &=a\{m+(1+z)m'\},\\
  \sum_{i,\ell}
  (\bm d_\ell^\trans\bm R_z\bm d_\ell)
  (\bm R_z)_{ii}
  &=
  \frac{a^2}{z}
  m\{(1+z)m-1\},
\end{aligned}
\end{equation}
where $m'=\partial_zm$. For the second identity, use
$\sum_{i,\ell}h_{\ell i}^2
=\tr(\bm D\bm D^\trans\bm R_z^2)$,
$\bm R_z^2-\bm R_z=z\bm H_{\bm D}\bm R_z^2$, and
$m'=a^{-1}\tr(\bm H_{\bm D}\bm R_z^2)$. The third follows from
$\tr(\bm D\bm D^\trans\bm R_z)
=a\{(1+z)m-1\}/z$.

Multiply \cref{eq:trace-start} by $\chi$, take expectations, and apply
\[
  \mathbb E[D_{i\ell}g(\bm D)]
  =
  \frac1b\mathbb E[\partial_{i\ell}g(\bm D)]
\]
to $g=\chi h_{\ell i}$. The product rule and
\cref{eq:resolvent-derivative-sums} give
\begin{equation}
\label{eq:localized-loop}
\begin{aligned}
  \bar m
  ={}&
  \mathbb E\chi
  +
  \beta z
  \mathbb E\left[
    \chi m\{(1+z)m-1\}
  \right] \\
  &+
  \frac{z^2}{b}
  \mathbb E\left[
    \chi\{m+(1+z)m'\}
  \right]
  +
  \frac{z}{ab}
  \sum_{i,\ell}
  \mathbb E\left[
    h_{\ell i}\partial_{i\ell}\chi
  \right].
\end{aligned}
\end{equation}
The term generated by the first line of
\cref{eq:resolvent-derivative-sums} cancels the term $-z\bar m$ from
\cref{eq:trace-start}. This cancellation is why the remaining leading term
has exactly the form of the Marchenko--Pastur equation.

\paragraph{The error in the fixed-point equation.}
Define
\[
  r(z)
  \coloneqq
  \bar m
  -
  1
  -
  \beta z\bar m
  \bigl\{
    (1+z)\bar m-1
  \bigr\}.
\]
Comparing \cref{eq:localized-loop} with
\cref{eq:one-view-MP-equation} reduces the problem to bounding this
residual. Rearranging \cref{eq:localized-loop} gives the central
decomposition
\begin{equation}
\label{eq:residual-decomposition}
\begin{aligned}
  r(z)
  ={}&
  \mathbb E\chi-1
  +
  \beta z(1+z)
  \bigl\{
    \mathbb E[\chi m^2]-\bar m^2
  \bigr\} \\
  &+
  \frac{z^2}{b}
  \mathbb E\left[
    \chi\{m+(1+z)m'\}
  \right]
  +
  \frac{z}{ab}
  \sum_{i,\ell}
  \mathbb E\left[
    h_{\ell i}\partial_{i\ell}\chi
  \right].
\end{aligned}
\end{equation}
We bound the four terms in this order.

First, $\chi\ne1$ implies
$\opn{\bm H_{\bm Z_{22}}}>s_\star^2/128$. The Gaussian Wishart deviation
bound and $b(s_\star^2\wedge s_\star^4)\ge\nu^2a$ therefore give
\begin{equation}
\label{eq:cutoff-tail}
  \mathbb P\{\chi(\bm Z_{22})\ne1\}
  \le
  C\exp\{-cb(s_\star^2\wedge s_\star^4)\}
  \le
  Ce^{-ca}.
\end{equation}
In particular, $|\mathbb E\chi-1|\lesssim a^{-1}$.

For the second term, the exact identity
\[
  \mathbb E[\chi m^2]-\bar m^2
  =
  \mathbb E[(F-\bar m)^2]
  +
  \mathbb E[\chi(1-\chi)m^2]
\]
is valid also for complex $m$. Using
\cref{eq:resolvent-bound,eq:fluctuation,eq:cutoff-tail}, we obtain
\begin{equation}
\label{eq:quadratic-bias}
  \left|
    \mathbb E[\chi m^2]-\bar m^2
  \right|
  \le
  \mathbb E|F-\bar m|^2
  +
  \frac{64}{25}\mathbb P\{\chi\ne1\}
  \lesssim
  \frac1a.
\end{equation}
Moreover,
$\beta|z||1+z|
\lesssim\beta/s_\star^2+\beta/s_\star^4\lesssim1$, so the second term in
\cref{eq:residual-decomposition} is also of order $a^{-1}$.

It remains to bound the last two terms. On the support of $\chi$,
$m'=a^{-1}\tr(\bm H_{\bm D}\bm R_z^2)$, and hence
$|m'|\lesssim s_\star^2$. Where the cutoff changes,
$\frn{\nabla_{\bm D}\chi}
\lesssim\sqrt{1+s_\star^2}/s_\star^2$, while
$\frn{\bm D^\trans\bm R_z}
\lesssim\sqrt{a(1+s_\star^2)}$. Since
$(h_{\ell i})_{\ell,i}=\bm D^\trans\bm R_z$, these estimates give
\begin{equation}
\label{eq:remaining-residual-terms}
\begin{aligned}
  \frac{|z|^2}{b}
  \mathbb E\left[
    \chi|m+(1+z)m'|
  \right]
  &\lesssim
  \frac{1+s_\star^2}{bs_\star^4}
  =
  \frac1a
  \left(
    \frac{\beta}{s_\star^4}
    +
    \frac{\beta}{s_\star^2}
  \right)
  \lesssim
  \frac1a,\\
  \frac{|z|}{ab}
  \left|
    \sum_{i,\ell}
    \mathbb E\left[
      h_{\ell i}\partial_{i\ell}\chi
    \right]
  \right|
  &\le
  \frac{|z|}{ab}
  \mathbb E\left[
    \frn{\bm D^\trans\bm R_z}
    \frn{\nabla_{\bm D}\chi}
  \right] \\
  &\lesssim
  \frac{1}{a\sqrt a}
  \left(
    \frac{\beta}{s_\star^4}
    +
    \frac{\beta}{s_\star^2}
  \right)
  \lesssim
  \frac1a.
\end{aligned}
\end{equation}
Combining
\cref{eq:cutoff-tail,eq:quadratic-bias,eq:remaining-residual-terms} in
\cref{eq:residual-decomposition} proves
\begin{equation}
\label{eq:residual-bound}
  |r(z)|
  \lesssim
  \frac1a.
\end{equation}

\paragraph{Stability of the Marchenko--Pastur equation.}
It remains to convert the residual bound into a bound for
$\bar m-m_\beta(z)$. By \cref{eq:resolvent-bound},
$|F(\bm D)|\le8/5$ for every $\bm D$, and therefore
$|\bar m|\le8/5$. The centered Marchenko--Pastur law is supported on
$[\beta-2\sqrt\beta,\beta+2\sqrt\beta]$. Since
$s_\star^2\ge\nu\sqrt\beta$, the signal condition gives, for sufficiently
large $\nu$,
\[
  |z|(\beta+2\sqrt\beta)
  \le
  C\frac{\sqrt\beta}{s_\star^2}
  \le\frac14.
\]
The integral representation of $m_\beta$ consequently gives
$|m_\beta(z)|\le4/3$.

Subtracting \cref{eq:one-view-MP-equation} from the definition of $r(z)$
gives
\begin{equation}
\label{eq:fixed-point-stability}
  (\bar m-m_\beta(z))
  \left[
    1
    -
    \beta z
    \bigl\{
      (1+z)(\bar m+m_\beta(z))-1
    \bigr\}
  \right]
  =
  r(z).
\end{equation}
The bounds on $\bar m$ and $m_\beta(z)$ imply
\[
  \left|
    \beta z
    \bigl\{
      (1+z)(\bar m+m_\beta(z))-1
    \bigr\}
  \right|
  \lesssim
  \frac{\beta}{s_\star^2}
  +
  \frac{\beta}{s_\star^4}
  \le
  \frac{C}{\nu^2}.
\]
For sufficiently large $\nu$, the bracket in
\cref{eq:fixed-point-stability} has modulus at least $1/2$. Hence
\cref{eq:residual-bound,eq:fixed-point-stability} yield
\begin{equation}
\label{eq:bias-bound}
  |\bar m-m_\beta(z)|
  \lesssim
  \frac1a.
\end{equation}

Finally, \cref{eq:resolvent-fluctuation-bias-split,eq:fluctuation,eq:bias-bound}
give
\[
  \left\|
    \ind_{\cE_{\rm loc}}
    \bigl(
      \widehat m_{\rm bulk}(z)-m_\beta(z)
    \bigr)
  \right\|_{L^{32}}
  \lesssim
  \frac1{\sqrt a}.
\]
Every bound above is uniform over $z\in\mathcal C_\star$, because the
argument uses $z$ only through $|z|=32/s_\star^2$. Taking the supremum over
the contour proves \cref{eq:one-view-resolvent-bound}.

\subsubsection{Proof of~\cref{lem:one-view-bulk-error}}

Each term in the two bulk remainders has the same structure: the contour
integrand contains the scalar difference
$\widehat m_{\rm bulk}(z)-m_\beta(z)$ multiplied by a bounded matrix kernel.  We
give the complete argument for the $\mathcal A_U$ term and then record the
kernel bounds needed for the other three terms.

By
\cref{eq:one-view-empirical-AU-contour,eq:one-view-population-AU-contour},
\[
  (\widehat{\mathcal A}_U-\mathcal A_U)
  (\bm{V}_\parallel^\trans\bm{V}_\parallel)
  =
  \frac{\beta}{2\pi\mathrm i}
  \oint_{\mathcal C_\star}
  \frac{\widehat m_{\rm bulk}(z)-m_\beta(z)}{z}\,
  \bm\Sigma\bm{K}_0(z)
  \bm{V}_\parallel^\trans\bm{V}_\parallel
  \bm{K}_0(z)\bm\Sigma\,dz.
\]
On $\cE_{\rm loc}$, the identities
$\widehat s_i^2-1=\tau_i+\beta+\beta/\tau_i$ and
$\widehat s_i\ell_i=\tau_i+\beta$ give
$\|\bm{W}_L\bm\Sigma\bm{K}_0(z)\|_{\mathrm{op}}\le C$ uniformly over
$z\in\mathcal C_\star$.  Since $\|\bm{V}_\parallel\|_{\mathrm{op}}\le1$, Minkowski's integral
inequality and \cref{eq:one-view-resolvent-bound} yield
\begin{align*}
&\Lqn{
  \ind_{\cE_{\rm loc}}
  \bm{W}_L
  (\widehat{\mathcal A}_U-\mathcal A_U)
  (\bm{V}_\parallel^\trans\bm{V}_\parallel)
  \bm{W}_L
}{32}\\
&\quad\le
\frac{\beta}{2\pi}
\oint_{\mathcal C_\star}
\frac{1}{|z|}
\left\|
  \ind_{\cE_{\rm loc}}
  |\widehat m_{\rm bulk}(z)-m_\beta(z)|
  \opn{\bm{W}_L\bm\Sigma\bm{K}_0(z)}^{\,2}
\right\|_{L^{32}}
|dz|\\
&\quad\le
C\frac{\beta}{\sqrt a}
=
C\frac{\sqrt a}{b},
\end{align*}
where the last identity uses $a=\beta b$.

The remaining three operator differences are controlled identically, so
we omit the repeated contour calculations.  The only changes are the
matrix kernels.  On $\cE_{\rm loc}$, uniformly over $z\in\mathcal C_\star$, the
same scalar identities give
\[
\begin{gathered}
  |1+z^{-1}|\opn{\bm{W}_L\bm{K}_0(z)}^{\,2}\le C,\\
  |1+z^{-1}|\opn{\bm{W}_{\rm bal}\bm{K}_0(z)}^{\,2}\le C,\\
  |1+z^{-1}|
  \opn{\bm{W}_{\rm bal}\bm\Sigma^{-1}\bm{K}_0(z)}
  \opn{\bm{W}_{\rm bal}\bm\Sigma^{-1}}
  +
  |1+z^{-1}|^2
  \opn{\bm{W}_{\rm bal}\bm\Sigma^{-1}\bm{K}_0(z)}^{\,2}
  \le C.
\end{gathered}
\]

The first bound controls the $\widehat{\mathcal B}-\mathcal B$ term
in $\bm{R}_{U,\mathrm{bulk}}$, the second controls the corresponding term
in $\bm{R}_{V,\mathrm{bulk}}$, and the third controls
$\widehat{\mathcal A}_V-\mathcal A_V$.  For the last difference, the
base term $\bm\Sigma^{-1}\bm{M}\bm\Sigma^{-1}$ cancels between the
empirical and population operators.  Using
$\|\bm{U}_\parallel\|_{\mathrm{op}},\|\bm{V}_\parallel\|_{\mathrm{op}}\le1$, each of these three
terms is therefore bounded by the same contour argument by
$C\beta/\sqrt a=C\sqrt a/b$.

Finally, apply the triangle inequality to the decompositions in
\cref{eq:one-view-RU-bulk,eq:one-view-RV-bulk}.  Combining the four
operator bounds gives \cref{eq:one-view-bulk-error-master}.

\section{Completion of the upper-bound proof}

We now prove \cref{lem:stochastic-master}. We establish its three conclusions
in order. 

\subsection{Proof of~\cref{eq:first-order-expectation}}
\label{app:first-order-error}

Recall that $\bm R_k=\bm S_k^\dagger\bm U$. Using the two blocks of
$\bm Q_\star$, we obtain
\begin{equation}
\label{eq:first-order-split}
  (\bm L\bm J-\bm T\bm L)\bm Q_\star
  =
  \alpha_1\bm L_1\bm R_1
  +
  \alpha_2\bm L_2\bm R_2
  -
  \bm T
  \bigl(
    \bm L_1\bm R_1-\bm L_2\bm R_2
  \bigr).
\end{equation}

\paragraph{Step 1: Bound $\alpha_1\bm L_1\bm R_1
  +
  \alpha_2\bm L_2\bm R_2$. }
The first two terms can be bounded directly. On $\cE_{\rm emb}$,
\cref{eq:population-Rk-bounds} gives
$\opn{\bm R_k}\lesssim\gamma_k^{-1}$, while
\cref{eq:one-view-bounds} gives
$\|\bm L_k\|_{L^{64}}\lesssim\sigma\sqrt{a_k}
\lesssim\sigma\sqrt n$. Since
$\cE_{\rm good}\subseteq\cE_{\rm emb}$, these estimates imply
\[
  \Lqn{
    \ind_{\cE_{\rm good}}
    \left(
      \alpha_1\bm L_1\bm R_1
      +
      \alpha_2\bm L_2\bm R_2
    \right)
  }{32}
  \lesssim
  \sigma\sqrt n
  \left(
    \frac{\alpha_1}{\gamma_1}
    +
    \frac{\alpha_2}{\gamma_2}
  \right)
  \lesssim
  \frac{\sigma\sqrt n}{\gmax}.
\]
The last inequality uses
$\alpha_1/\gamma_1+\alpha_2/\gamma_2
=(\gamma_1+\gamma_2)/(\gamma_1^2+\gamma_2^2)
\lesssim\gmax^{-1}$.

\paragraph{Step 2: Bound $\bm T
  \bigl(
    \bm L_1\bm R_1-\bm L_2\bm R_2
  \bigr)$. }
The same argument is too crude for the term multiplied by $\bm T$.
Submultiplicativity would bound
$\opn{\bm T\bm L_k\bm R_k}$ by
$\opn{\bm T}\opn{\bm L_k}\opn{\bm R_k}$ and would therefore retain
the factor $\sqrt{a_k}$ from the bound for $\bm L_k$. To obtain the rank
dependence in \cref{eq:first-order-expectation}, we need to leverage the rotational invariance of $\bm L_k$.

For $k\in\{1,2\}$, define
\begin{equation}
\label{eq:Hk-app}
  \Hc_k
  \coloneqq
  \sigma\left(
    \bm S_1,\bm S_2,\bm L_{3-k},\bm L_k^\trans\bm L_k
  \right).
\end{equation}
The sigma-field $\Hc_k$ records the quantities that remain unchanged when
$\bm L_k$ is left multiplied by an orthogonal transformation on
$\cM_k^\perp$. The following lemma uses this invariance to bound a product
in terms of the operator and Frobenius norms of the matrices on either side
of $\bm L_k$.

\begin{lemma}
\label{lem:leakage-product}
Fix $k\in\{1,2\}$ and $q\ge2$. Let $\cE\in\Hc_k$, and let
$\bm\Psi\in\mathbb R^{n\times u}$ and
$\bm\Theta\in\mathbb R^{p_k\times v}$ be $\Hc_k$-measurable. Then
\begin{equation}
\label{eq:leakage-product}
  \Lqn{
    \ind_{\cE}\bm\Psi^\trans\bm L_k\bm\Theta
  }{q}
  \lesssim
  \sqrt{\frac q{a_k}}\,
  \Lqn{\bm L_k}{2q}
  \Lqn{
    \ind_{\cE}
    \left(
      \frn{\bm\Psi}\opn{\bm\Theta}
      +
      \opn{\bm\Psi}\frn{\bm\Theta}
    \right)
  }{2q}.
\end{equation}
\end{lemma}

We defer the proof until after the proof of
\cref{lem:stochastic-master}. The lemma does not require
$\bm\Psi$ and $\bm\Theta$ to be independent of $\bm L_k$; their
$\Hc_k$-measurability is sufficient. Taking $q=32$ in
\cref{eq:leakage-product} and using
$\|\bm L_k\|_{L^{64}}\lesssim\sigma\sqrt{a_k}$ gives the form used below:
\begin{equation}
\label{eq:leakage-product-32}
  \Lqn{
    \ind_{\cE}\bm\Psi^\trans\bm L_k\bm\Theta
  }{32}
  \lesssim
  \sigma
  \Lqn{
    \ind_{\cE}
    \left(
      \frn{\bm\Psi}\opn{\bm\Theta}
      +
      \opn{\bm\Psi}\frn{\bm\Theta}
    \right)
  }{64}.
\end{equation}

Since $\cE_{\rm emb}$ and $\bm R_k=\bm S_k^\dagger\bm U$ are
$\Hc_k$-measurable and $\bm T$ is deterministic, we may apply
\cref{eq:leakage-product-32} with
$\bm\Psi=\bm T^\trans$ and $\bm\Theta=\bm R_k$. Thus
\[
\begin{aligned}
  \Lqn{
    \ind_{\cE_{\rm good}}
    \bm T
    \bigl(
      \bm L_1\bm R_1-\bm L_2\bm R_2
    \bigr)
  }{32}
  &\le
  \sum_{k=1}^2
  \Lqn{
    \ind_{\cE_{\rm emb}}
    \bm T\bm L_k\bm R_k
  }{32} \\
  &\overset{\mathrm{(i)}}{\lesssim}
  \sigma\sum_{k=1}^2
  \left(
    \frn{\bm T}\opn{\bm R_k}
    +
    \opn{\bm T}\frn{\bm R_k}
  \right) \\
  &\overset{\mathrm{(ii)}}{\lesssim}
  \sigma\sum_{k=1}^2
  \left\{
    \frac{\sqrt{r_1+r_2}}{\gamma_k}
    +
    \frac{\sqrt{r_1\wedge r_2}+\sqrt r}
         {\gamma_k\sqrt\theta}
  \right\} \\
  &\lesssim
  \frac{\sigma\sqrt{r_1+r_2}}{\gmin}
  +
  \frac{\sigma\sqrt{p_{\min}}}{\gmin\sqrt\theta}.
\end{aligned}
\]
Here (i) is \cref{eq:leakage-product-32}. For (ii), we use
\cref{eq:transfer-op} in the forms
$\frn{\bm T}\lesssim
\sqrt{r_1+r_2}+\sqrt{r_1\wedge r_2}/\sqrt\theta$ and
$\opn{\bm T}\lesssim\theta^{-1/2}$, together with
\cref{eq:population-Rk-bounds} in the forms
$\opn{\bm R_k}\lesssim\gamma_k^{-1}$ and
$\frn{\bm R_k}\lesssim\sqrt r/\gamma_k$.
The last inequality uses $\gamma_k\ge\gmin$ and
$\sqrt{r+ (r_1\wedge r_2)}\asymp\sqrt{p_{\min}}$.

Combining this estimate with the bound for the first two terms in
\cref{eq:first-order-split} and using $L^1\le L^{32}$ proves
\cref{eq:first-order-expectation}.

\subsection{Proof of~\cref{eq:higher-order-expectation}}
\label{app:higher-order-error}

Write
$\bm\Delta_{\rm quad}\coloneqq\bm L^\trans\bm L-\bm D_\sigma$.
By \cref{eq:higher-order-deterministic}, it is enough to bound
$\bm\Delta_{\rm quad}$, $\bm P_{\col(\bm S)}\bm L$, and
$\bm\Delta_{+0}$.

\paragraph{The matrix $\bm\Delta_{\rm quad}$.}
The diagonal blocks of $\bm\Delta_{\rm quad}$ are
$\bm L_k^\trans\bm L_k-a_k\sigma^2\bm I_{p_k}$ and are controlled by
\cref{eq:one-view-bounds}. For the off-diagonal block,
$\bm L_2$ is $\Hc_1$-measurable, so
\cref{eq:leakage-product-32} applies with
$\bm\Psi=\bm L_2$ and $\bm\Theta=\bm I_{p_1}$. Hence
\[
\begin{aligned}
  \Lqn{\bm L_2^\trans\bm L_1}{32}
  &\overset{\mathrm{(i)}}{\lesssim}
  \sigma
  \left\|
    \frn{\bm L_2}+\sqrt{p_1}\opn{\bm L_2}
  \right\|_{L^{64}} \\
  &\le
  \sigma(\sqrt{p_1}+\sqrt{p_2})\Lqn{\bm L_2}{64}
  \overset{\mathrm{(ii)}}{\lesssim}
  \sigma^2\sqrt{np_{\max}}.
\end{aligned}
\]
Here (i) is \cref{eq:leakage-product-32}, and (ii) uses
\cref{eq:one-view-bounds} and $a_2\le n$. The transpose block satisfies
the same estimate, so the diagonal and off-diagonal bounds give
\begin{equation}
\label{eq:quad-moment}
  \Lqn{\bm\Delta_{\rm quad}}{32}
  \lesssim
  \sigma^2\sqrt{np_{\max}}.
\end{equation}

\paragraph{The matrix $\bm P_{\col(\bm S)}\bm L$.}
On $\cE_{\rm emb}$, we have
$\col(\bm S)=\cM_1+\cM_2$ and therefore
$\rank\{\bm P_{\col(\bm S)}\}=p_1+p_2-r$. The event and the projection
are $\Hc_k$-measurable, so
\cref{eq:leakage-product-32} applies with
$\bm\Psi=\bm P_{\col(\bm S)}$ and
$\bm\Theta=\bm I_{p_k}$. Thus
\begin{align}
  \Lqn{
    \ind_{\cE_{\rm emb}}
    \bm P_{\col(\bm S)}\bm L
  }{32}
  &\le
  \sum_{k=1}^2
  \Lqn{
    \ind_{\cE_{\rm emb}}
    \bm P_{\col(\bm S)}\bm L_k
  }{32} \nonumber \\
  &\lesssim
  \sigma\sum_{k=1}^2
  \left(
    \sqrt{p_1+p_2-r}+\sqrt{p_k}
  \right)
  \lesssim
  \sigma\sqrt{p_{\max}}.\label{eq:projected-leakage}
\end{align}

\paragraph{The block $\bm\Delta_{+0}$.}
Set $\bm P_0\coloneqq\bm P_{\ker(\bm G_0)}$. Since
$\bm S\bm P_0=\bze$, the definition of $\bm\Delta$ gives
\[
  \bm\Delta_{+0}
  =
  \bigl(\bm S(\bm G_0^\dagger)^{1/2}\bigr)^\trans
  \bm L\bm P_0
  +
  (\bm G_0^\dagger)^{1/2}
  \bm\Delta_{\rm quad}\bm P_0.
\]
For the first term, apply \cref{eq:leakage-product-32} to each view with
$\bm\Psi=\bm S(\bm G_0^\dagger)^{1/2}$ and
$\bm\Theta=(\bm P_0)_k$, where $(\bm P_0)_k$ denotes the $k$th block
rows of $\bm P_0$. Combining this estimate with
\cref{eq:population-Gamma-bounds,eq:population-SGamma-bounds,eq:quad-moment}
gives
\begin{align}
  \Lqn{
    \ind_{\cE_{\rm emb}}\bm\Delta_{+0}
  }{32}
  &\le
  \sum_{k=1}^2
  \Lqn{
    \ind_{\cE_{\rm emb}}
    \bigl(\bm S(\bm G_0^\dagger)^{1/2}\bigr)^\trans
    \bm L_k(\bm P_0)_k
  }{32}
  +
  \opn{(\bm G_0^\dagger)^{1/2}}
  \Lqn{\bm\Delta_{\rm quad}}{32} \nonumber \\
  &\overset{\mathrm{(i)}}{\lesssim}
  \sigma\sum_{k=1}^2
  \left\{
    \frn{\bm S(\bm G_0^\dagger)^{1/2}}\opn{(\bm P_0)_k}
    +
    \opn{\bm S(\bm G_0^\dagger)^{1/2}}\frn{(\bm P_0)_k}
  \right\}
  +
  \frac{\sigma^2\sqrt{np_{\max}}}{\gmin\sqrt\theta} \nonumber \\
  &\overset{\mathrm{(ii)}}{\lesssim}
  \sigma\sqrt{p_{\max}}
  +
  \frac{\sigma^2\sqrt{np_{\max}}}{\gmin\sqrt\theta}.\label{eq:relative-p0}
\end{align}
Here (i) uses
\cref{eq:leakage-product-32,eq:population-Gamma-bounds,eq:quad-moment}.
For (ii), \cref{eq:population-SGamma-bounds} gives
$\opn{\bm S(\bm G_0^\dagger)^{1/2}}\le1$ and the required Frobenius bound,
while $\opn{(\bm P_0)_k}\le1$ and
$\frn{(\bm P_0)_k}\le\sqrt r$.

We now apply \cref{eq:higher-order-deterministic}. The small-rate
assumption gives
$\sigma^2\sqrt{np_{\max}}/(\gmin^2\theta)\lesssim1$ and, since
$p_{\max}\le n$,
$\{\sigma\sqrt{p_{\max}}/(\gmin\sqrt\theta)\}^2
\le\sigma^2\sqrt{np_{\max}}/(\gmin^2\theta)\lesssim1$.
We also use the standing condition
$\sigma\sqrt n/\gmin\lesssim1$. Cauchy--Schwarz and
\cref{eq:quad-moment,eq:projected-leakage,eq:relative-p0} yield
\[
\begin{aligned}
  \mathbb E\left[
    \ind_{\cE_{\rm good}}\opn{\bm{\mathsf{Rem}}}
  \right]
  \lesssim{}&
  \frac{\sigma^2\sqrt{np_{\max}}}{\gmin^2\theta}
  +
  \frac{\sigma\sqrt n}{\gmin^2\sqrt\theta}
  \left(
    \sigma\sqrt{p_{\max}}
    +
    \frac{\sigma^2\sqrt{np_{\max}}}{\gmin\sqrt\theta}
  \right) \\
  &+
  \frac{\sigma\sqrt{p_{\max}}}{\gmin^2\theta}
  \left(
    \sigma\sqrt{p_{\max}}
    +
    \frac{\sigma^2\sqrt{np_{\max}}}{\gmin\sqrt\theta}
  \right) \\
  &+
  \frac{1}{\gmin^2\theta}
  \left(
    \sigma\sqrt{p_{\max}}
    +
    \frac{\sigma^2\sqrt{np_{\max}}}{\gmin\sqrt\theta}
  \right)^2
  \lesssim
  \frac{\sigma^2\sqrt{np_{\max}}}{\gmin^2\theta}.
\end{aligned}
\]
For the second term, the two summands are absorbed using respectively
$\theta\le1$ and $\sigma\sqrt n/\gmin\lesssim1$. For the third, use
$p_{\max}\le\sqrt{np_{\max}}$ and
$\sigma\sqrt{p_{\max}}/(\gmin\sqrt\theta)\lesssim1$. The final term follows
from the same inequalities and the small-rate assumption. Since
$p_{\max}=r+(r_1\vee r_2)$, this proves
\cref{eq:higher-order-expectation}.

\subsection{Proof of~\cref{eq:G-complement}}
\label{app:good-event-probability}

By the definition of $\cE_{\rm good}$,
\[
  \Pp(\cE_{\rm good}^c)
  \le
  \Pp(\cE_{\rm emb}^c)
  +
  \Pp\left\{
    \cE_{\rm emb},\,
    \opn{\bm\Delta_{++}}>\frac14
  \right\}
  +
  \Pp\left\{
    \cE_{\rm emb},\,
    \opn{\bm\Delta_{00}}
    +2\opn{\bm\Delta_{+0}}^2
    >
    \frac18\cemb^2\gmin^2\theta
  \right\}.
\]
Because $\cE_1\cap\cE_2\subseteq\cE_{\rm emb}$,
\cref{lem:one-view} and a union bound give
$\Pp(\cE_{\rm emb}^c)
\lesssim\sigma^2\sqrt{np_{\max}}/\gmin^2
\le\sigma^2\sqrt{np_{\max}}/(\gmin^2\theta)$.

It remains to bound $\bm\Delta_{00}$ and $\bm\Delta_{++}$; the bound for
$\bm\Delta_{+0}$ is already given by \cref{eq:relative-p0}. Since
$\bm S\bm P_0=\bze$, the two terms in $\bm\Delta$ containing one factor
$\bm L$ vanish after multiplication by $\bm P_0$ on both sides. Therefore
\begin{equation}
\label{eq:relative-00}
  \bm\Delta_{00}
  =
  \bm P_0\bm\Delta_{\rm quad}\bm P_0,
  \qquad
  \Lqn{
    \ind_{\cE_{\rm emb}}\bm\Delta_{00}
  }{32}
  \lesssim
  \sigma^2\sqrt{np_{\max}},
\end{equation}
where the bound follows from \cref{eq:quad-moment}.

For $\bm\Delta_{++}$, expanding \cref{eq:Delta} gives
\[
\begin{aligned}
  \bm\Delta_{++}
  ={}&
  \bigl(\bm S(\bm G_0^\dagger)^{1/2}\bigr)^\trans
  \bm L(\bm G_0^\dagger)^{1/2}
  +
  (\bm G_0^\dagger)^{1/2}\bm L^\trans
  \bm S(\bm G_0^\dagger)^{1/2} \\
  &+
  (\bm G_0^\dagger)^{1/2}
  \bm\Delta_{\rm quad}
  (\bm G_0^\dagger)^{1/2}.
\end{aligned}
\]
The first two terms are transposes. Apply
\cref{eq:leakage-product-32} to each view with
$\bm\Psi=\bm S(\bm G_0^\dagger)^{1/2}$ and
$\bm\Theta=((\bm G_0^\dagger)^{1/2})_k$, and use
\cref{eq:population-Gamma-bounds,eq:population-SGamma-bounds,eq:quad-moment}.
This gives
\begin{equation}
\label{eq:relative-pp}
\begin{aligned}
  \Lqn{
    \ind_{\cE_{\rm emb}}\bm\Delta_{++}
  }{32}
  &\lesssim
  \sigma
  \left\{
    \frn{\bm S(\bm G_0^\dagger)^{1/2}}
    \opn{(\bm G_0^\dagger)^{1/2}}
    +
    \opn{\bm S(\bm G_0^\dagger)^{1/2}}
    \frn{(\bm G_0^\dagger)^{1/2}}
  \right\}
  +
  \opn{(\bm G_0^\dagger)^{1/2}}^2
  \Lqn{\bm\Delta_{\rm quad}}{32} \\
  &\lesssim
  \frac{\sigma\sqrt{r_1+r_2}}{\gmin}
  +
  \frac{\sigma\sqrt{p_{\min}}}{\gmin\sqrt\theta}
  +
  \frac{\sigma\sqrt{p_{\max}}}{\gmin\sqrt\theta}
  +
  \frac{\sigma^2\sqrt{np_{\max}}}{\gmin^2\theta}.
\end{aligned}
\end{equation}

The last event in the initial union bound is contained in the union of
$\{\opn{\bm\Delta_{00}}>c_1\gmin^2\theta\}$ and
$\{\opn{\bm\Delta_{+0}}>c_2\gmin\sqrt\theta\}$ for sufficiently small
numerical constants $c_1,c_2>0$. Markov's inequality with exponent $32$,
followed by
\cref{eq:relative-00,eq:relative-p0,eq:relative-pp}, therefore gives
\begin{align}
\label{eq:relative-failure}
  &\Pp\left\{
    \cE_{\rm emb},\,
    \opn{\bm\Delta_{++}}>\frac14
  \right\}
  +
  \Pp\left\{
    \cE_{\rm emb},\,
    \opn{\bm\Delta_{00}}
    +2\opn{\bm\Delta_{+0}}^2
    >
    \frac18\cemb^2\gmin^2\theta
  \right\}
  \notag\\
  &\quad\lesssim
  \Lqn{
    \ind_{\cE_{\rm emb}}\bm\Delta_{++}
  }{32}^{32}
  +
  \left(
    \frac{
      \Lqn{\ind_{\cE_{\rm emb}}\bm\Delta_{00}}{32}
    }{\gmin^2\theta}
  \right)^{32}
  +
  \left(
    \frac{
      \Lqn{\ind_{\cE_{\rm emb}}\bm\Delta_{+0}}{32}
    }{\gmin\sqrt\theta}
  \right)^{32}
  \notag\\
  &\quad\lesssim
  \frac{\sigma\sqrt{r_1+r_2}}{\gmin}
  +
  \frac{\sigma\sqrt{p_{\min}}}{\gmin\sqrt\theta}
  +
  \frac{\sigma^2\sqrt{np_{\max}}}{\gmin^2\theta}.
\end{align}
For the last inequality, the small-rate assumption bounds the $32$nd power
of each displayed rate by the rate itself. The only additional term is
$\sigma\sqrt{p_{\max}}/(\gmin\sqrt\theta)$, whose $32$nd power is bounded by
its square and hence by
$\sigma^2\sqrt{np_{\max}}/(\gmin^2\theta)$ because $p_{\max}\le n$.

Combining this estimate with the bound for $\Pp(\cE_{\rm emb}^c)$ and using
$\sqrt{p_{\min}}\asymp\sqrt{r+ (r_1\wedge r_2)}$ proves
\cref{eq:G-complement}. Together with
\cref{eq:first-order-expectation,eq:higher-order-expectation}, this
completes the proof of \cref{lem:stochastic-master}.

\subsection{Proof of~\cref{lem:leakage-product}}
\label{app:proof-leakage-product}

\begin{proof}
The proof has two parts. We first use the rotational invariance of the
Gaussian noise to randomize the orientation of \(\bm L_k\) while preserving
every quantity in \(\Hc_k\). Conditional on the original data, the resulting
product is then controlled by a fixed-matrix Haar estimate.

Fix an orthonormal matrix
\(\bm N_k\in\St(n,a_k)\) whose columns span \(\cM_k^\perp\), and let
\(\bm O\) be an independent Haar matrix in \(\mathbb O(a_k)\). The
distributional identity used below is
\begin{equation}
\label{eq:leakage-rotation-app}
\begin{aligned}
  \bigl(
    \bm S_1,\bm S_2,\bm L_{3-k},
    \bm L_k^\trans\bm L_k,\bm L_k
  \bigr)
  \ \overset{\mathrm d}{=}\
  \bigl(
    \bm S_1,\bm S_2,\bm L_{3-k},
    \bm L_k^\trans\bm L_k,
    \bm N_k\bm O\bm N_k^\trans\bm L_k
  \bigr).
\end{aligned}
\end{equation}
We first derive the claimed moment bound from this identity and verify the
identity at the end of the proof.

Choose measurable representatives of
\(\ind_{\cE}\), \(\bm\Psi\), and \(\bm\Theta\) as functions of the tuple
defining \(\Hc_k\). These representatives are unchanged by the replacement
of \(\bm L_k\) in \cref{eq:leakage-rotation-app}. Define
\[
  \bm H\coloneqq\bm N_k^\trans\bm\Psi,
  \qquad
  \bm M\coloneqq\bm N_k^\trans\bm L_k\bm\Theta.
\]
Since \(\bm L_k=\bm N_k\bm N_k^\trans\bm L_k\), the product obtained after
the replacement in \cref{eq:leakage-rotation-app} is
\(\bm H^\trans\bm O\bm M\).

We next establish the fixed-matrix Haar estimate needed after conditioning.
Let \(\bm H_0\in\mathbb R^{a\times u}\) and
\(\bm M_0\in\mathbb R^{a\times v}\) be deterministic, and let
\(\bm\Gamma\in\mathbb R^{a\times a}\) have independent standard Gaussian
entries. Chevet's inequality
\cite[Section~8.7]{vershynin2018high}, applied with
\(T=\bm M_0\mathbb S^{v-1}\) and
\(S=\bm H_0\mathbb S^{u-1}\), gives
\[
  \mathbb E
  \|\bm H_0^\trans\bm\Gamma\bm M_0\|_{\mathrm{op}}
  \le
  \|\bm H_0\|_{\mathrm F}\|\bm M_0\|_{\mathrm{op}}
  +
  \|\bm H_0\|_{\mathrm{op}}\|\bm M_0\|_{\mathrm F}.
\]
The map
\(\bm\Gamma\mapsto
\|\bm H_0^\trans\bm\Gamma\bm M_0\|_{\mathrm{op}}\)
is
\(\|\bm H_0\|_{\mathrm{op}}\|\bm M_0\|_{\mathrm{op}}\)-Lipschitz with
respect to the Frobenius norm, because
\[
\begin{aligned}
  \left|
    \|\bm H_0^\trans\bm\Gamma_1\bm M_0\|_{\mathrm{op}}
    -
    \|\bm H_0^\trans\bm\Gamma_0\bm M_0\|_{\mathrm{op}}
  \right|
  &\le
  \|\bm H_0^\trans(\bm\Gamma_1-\bm\Gamma_0)\bm M_0\|_{\mathrm{op}} \\
  &\le
  \|\bm H_0\|_{\mathrm{op}}
  \|\bm\Gamma_1-\bm\Gamma_0\|_{\mathrm F}
  \|\bm M_0\|_{\mathrm{op}}.
\end{aligned}
\]
Combining Gaussian concentration with the expectation bound therefore yields
\[
\begin{aligned}
  \left(
    \mathbb E
    \|\bm H_0^\trans\bm\Gamma\bm M_0\|_{\mathrm{op}}^q
  \right)^{1/q}
  &\le
  C\left\{
    \|\bm H_0\|_{\mathrm F}\|\bm M_0\|_{\mathrm{op}}
    +
    \|\bm H_0\|_{\mathrm{op}}\|\bm M_0\|_{\mathrm F}
    +
    \sqrt q\,
    \|\bm H_0\|_{\mathrm{op}}\|\bm M_0\|_{\mathrm{op}}
  \right\} \\
  &\le
  C\sqrt q
  \left\{
    \|\bm H_0\|_{\mathrm F}\|\bm M_0\|_{\mathrm{op}}
    +
    \|\bm H_0\|_{\mathrm{op}}\|\bm M_0\|_{\mathrm F}
  \right\}.
\end{aligned}
\]
The second inequality uses \(q\ge2\) and bounds each operator norm by the
corresponding Frobenius norm.

Apply Tropp's comparison principle
\cite[Theorem~1]{tropp2012comparison} to the seminorm
\(\bm X\mapsto\|\bm H_0^\trans\bm X\bm M_0\|_{\mathrm{op}}\) and the
function \(t\mapsto t^q\). Comparing a Haar matrix with
\(\bm\Gamma/\sqrt a\) gives
\begin{equation}
\label{eq:fixed-haar-moment}
  \left(
    \mathbb E_{\bm O}
    \|\bm H_0^\trans\bm O\bm M_0\|_{\mathrm{op}}^q
  \right)^{1/q}
  \le
  C\sqrt{\frac qa}
  \left\{
    \|\bm H_0\|_{\mathrm F}\|\bm M_0\|_{\mathrm{op}}
    +
    \|\bm H_0\|_{\mathrm{op}}\|\bm M_0\|_{\mathrm F}
  \right\}.
\end{equation}

We apply \cref{eq:fixed-haar-moment} conditionally with
\(a=a_k\), \(\bm H_0=\bm H\), and \(\bm M_0=\bm M\).
The identity \cref{eq:leakage-rotation-app} and Tonelli's theorem give
\[
  \mathbb E\left[
    \ind_{\cE}
    \|\bm\Psi^\trans\bm L_k\bm\Theta\|_{\mathrm{op}}^q
  \right]
  =
  \mathbb E\mathbb E_{\bm O}\left[
    \ind_{\cE}
    \|\bm H^\trans\bm O\bm M\|_{\mathrm{op}}^q
  \right].
\]
Because \(\bm N_k\) has orthonormal columns,
\[
\begin{aligned}
  \|\bm H\|_{\mathrm F}\|\bm M\|_{\mathrm{op}}
  +
  \|\bm H\|_{\mathrm{op}}\|\bm M\|_{\mathrm F}
  \le
  \|\bm L_k\|_{\mathrm{op}}
  \left(
    \|\bm\Psi\|_{\mathrm F}\|\bm\Theta\|_{\mathrm{op}}
    +
    \|\bm\Psi\|_{\mathrm{op}}\|\bm\Theta\|_{\mathrm F}
  \right).
\end{aligned}
\]
Consequently, \cref{eq:fixed-haar-moment} followed by Hölder's inequality
gives
\[
\begin{aligned}
  \Lqn{
    \ind_{\cE}\bm\Psi^\trans\bm L_k\bm\Theta
  }{q}
  &\lesssim
  \sqrt{\frac q{a_k}}\,
  \left\|
    \ind_{\cE}\|\bm L_k\|_{\mathrm{op}}
    \left(
      \|\bm\Psi\|_{\mathrm F}\|\bm\Theta\|_{\mathrm{op}}
      +
      \|\bm\Psi\|_{\mathrm{op}}\|\bm\Theta\|_{\mathrm F}
    \right)
  \right\|_{L^q} \\
  &\le
  \sqrt{\frac q{a_k}}\,
  \Lqn{\bm L_k}{2q}
  \Lqn{
    \ind_{\cE}
    \left(
      \|\bm\Psi\|_{\mathrm F}\|\bm\Theta\|_{\mathrm{op}}
      +
      \|\bm\Psi\|_{\mathrm{op}}\|\bm\Theta\|_{\mathrm F}
    \right)
  }{2q}.
\end{aligned}
\]
This is the conclusion of \cref{lem:leakage-product}, subject only to the
distributional identity.

It remains to prove \cref{eq:leakage-rotation-app}. Fix
\(\bm o\in\mathbb O(a_k)\) and define
\[
  \bm R
  \coloneqq
  \bm P_{\cM_k}
  +
  \bm N_k\bm o\bm N_k^\trans.
\]
The matrix \(\bm R\) is orthogonal, acts as the identity on \(\cM_k\), and
acts as \(\bm o\) on \(\cM_k^\perp\). Since
\(\operatorname{col}(\bm X_k)\subseteq\cM_k\), we have
\(\bm R\bm X_k=\bm X_k\). Left rotational invariance of the Gaussian noise
and independence of the two views therefore imply
\[
  (\bm R\bm Y_k,\bm Y_{3-k})
  \ \overset{\mathrm d}{=}\
  (\bm Y_k,\bm Y_{3-k}).
\]

The bias-corrected empirical factor is left orthogonally equivariant:
\[
  \bm C_k(\bm R\bm Y_k)
  =
  \bm R\bm C_k(\bm Y_k)
  \qquad\text{almost surely}.
\]
Indeed, choose the right singular vectors measurably as functions of
\(\bm Y_k^\trans\bm Y_k\). Left multiplication by \(\bm R\) preserves this
Gram matrix and all singular values, while multiplying the corresponding
left singular vectors by \(\bm R\). The weights depend only on the singular
values, so the displayed equivariance follows.

Since \(\bm R\) acts as the identity on \(\cM_k\), this equivariance gives
\[
  \bm S_k(\bm R\bm Y_k)=\bm S_k(\bm Y_k),
  \qquad
  \bm L_k(\bm R\bm Y_k)
  =
  \bm N_k\bm o\bm N_k^\trans\bm L_k(\bm Y_k).
\]
The other view is unchanged, and the transformed leakage has the same Gram
matrix:
\[
  \bigl(
    \bm N_k\bm o\bm N_k^\trans\bm L_k
  \bigr)^\trans
  \bigl(
    \bm N_k\bm o\bm N_k^\trans\bm L_k
  \bigr)
  =
  \bm L_k^\trans\bm L_k.
\]
Thus \cref{eq:leakage-rotation-app} holds with \(\bm O\) replaced by any
fixed \(\bm o\). Averaging over an independent Haar matrix proves the stated
identity and completes the proof.
\end{proof}

\section{Probabilistic tools for the lower bound}

\paragraph{Gaussian conditioning.}
The following lemma is a simple consequence of the tail bound for the smallest singular value of a Gaussian random matrix. 
\begin{lemma}
\label{lem:gaussian-loading-lower-bound}
There exist numerical constants $c_B>0$ and $d_0\ge1$ such that,
whenever $d\ge d_0$ and $1\le p\le d/3$, a matrix
$\bm B\in\mathbb R^{d\times p}$ with independent
$\mathcal N(0,d^{-1})$ entries satisfies
\[
  \mathbb P\left\{
    \sigma_p^2(\bm B)<c_B
  \right\}
  \le \frac1{16}.
\]

\end{lemma}

\paragraph{Gaussian projector information.}
For an orthogonal projector $\bm P\in\R^{N\times N}$, $\sigma>0$, and
$\beta\ge0$, define
\[
  \mathsf G_{\sigma,\beta}(\bm P)
  \coloneqq
  \mathcal N\left(\bm 0,\sigma^2(\bm I_N+\beta\bm P)\right).
\]
The only covariance identity used later is the following standard
specialization: if $\bm P$ and $\bm Q$ are equal-rank orthogonal projectors,
then
\begin{equation}
\label{eq:gaussian-projector-kl}
  \operatorname{KL}\left(
    \mathsf G_{\sigma,\beta}(\bm P)
    \,\middle\|\,
    \mathsf G_{\sigma,\beta}(\bm Q)
  \right)
  =
  \frac{\beta^2}{4(1+\beta)}
  \frn{\bm P-\bm Q}^{\,2}.
\end{equation}
For independent columns, we apply additivity of relative entropy directly at
the point of use.

\paragraph{Fano.} Below is a simple consequence of classical Fano's inequality and the definition of the TV distance.
\begin{lemma}[Fano with proxy laws]
\label{lem:fano-proxy}
Let $\Pi_1,\ldots,\Pi_N$ be priors supported on a parameter class
$\mathfrak P$, with induced observation laws $P_1,\ldots,P_N$. Suppose that
the target functional is constant under each prior, with respective values
$\bm\Theta_1,\ldots,\bm\Theta_N$, and that
$d(\bm\Theta_i,\bm\Theta_j)\ge 2\delta \text{ for all }i\ne j.$
Let $Q_1,\ldots,Q_N$ be proxy laws such that
$\max_j\operatorname{TV}(P_j,Q_j)\le\varepsilon$. Then, for every reference index
$i_0\in[N]$,
\[
  \inf_{\widehat{\bm\Theta}}
  \sup_{\vartheta\in\mathfrak P}
  \E_\vartheta
  d\bigl(\widehat{\bm\Theta},\Theta(\vartheta)\bigr)
  \ge
  \delta\left\{
    1-\varepsilon
    -\frac{N^{-1}\sum_{j=1}^N
      \operatorname{KL}(Q_j\|Q_{i_0})+\log 2}{\log N}
  \right\}.
\]
\end{lemma}

\paragraph{Assouad.}
We also need the Assouad lemma. 
\begin{lemma}
\label{lem:assouad}
Let $\Omega_m\coloneqq\{-1,1\}^m$. For each
$\bm\omega\in\Omega_m$, let $\Pi_{\bm\omega}$ be a prior supported on
$\mathfrak P$, with induced observation law $P_{\bm\omega}$. Suppose
the target functional is constant under each prior:
\[
  \Theta(\vartheta)=\bm\Theta_{\bm\omega}
  \qquad
  \text{for $\Pi_{\bm\omega}$-almost every $\vartheta$}.
\]
Assume that, for some $a>0$,
\[
  d(\bm\Theta_{\bm\omega},\bm\Theta_{\bm\omega'})
  \ge
  a\sqrt{\operatorname{Ham}(\bm\omega,\bm\omega')}
  \qquad
  \text{for all }\bm\omega,\bm\omega'\in\Omega_m,
\]
where $\operatorname{Ham}(\bm\omega,\bm\omega')$
denotes the Hamming distance.
For $j\in[m]$, let $\bm\omega^{(j)}$ denote the vertex obtained by
flipping the $j$th coordinate of $\bm\omega$. If
\[
  \max_{\substack{\bm\omega\in\Omega_m\\j\in[m]}}
  \operatorname{TV}
  \bigl(P_{\bm\omega},P_{\bm\omega^{(j)}}\bigr)
  \le \eta<1,
\]
then
\[
  \inf_{\widehat{\bm\Theta}}
  \sup_{\vartheta\in\mathfrak P}
  \mathbb E_\vartheta
  d\bigl(\widehat{\bm\Theta},\Theta(\vartheta)\bigr)
  \ge
  \frac{a\sqrt m}{4}(1-\eta).
\]
\end{lemma}

\paragraph{Hellinger path.}

\begin{lemma}[Hellinger path bound]
\label{lem:hellinger-path}
Let $\{P_t:t\in[0,T]\}$ be probability measures dominated by $\mu$, with
strictly positive densities $p_t$. Suppose that
$t\mapsto\sqrt{p_t}$ is absolutely continuous as an $L^2(\mu)$-valued map
and that, for almost every $t$,
\[
  \frac{d}{dt}\sqrt{p_t}
  =
  \frac12 S_t\sqrt{p_t}
  \quad\text{in }L^2(\mu),
  \qquad
  S_t(x)\coloneqq\partial_t\log p_t(x).
\]
Define
$J_t\coloneqq\mathbb E_{P_t}[S_t(X)^2]$. Then, for the unnormalized
Hellinger distance
$H(P,Q)\coloneqq\|\sqrt p-\sqrt q\|_{L^2(\mu)}$,
\begin{equation}
\label{eq:hellinger-path-bound}
  \operatorname{TV}(P_0,P_T)
  \le
  H(P_0,P_T)
  \le
  \frac12\int_0^T\sqrt{J_t}\,dt.
\end{equation}
Consequently,
\begin{equation}
\label{eq:hellinger-path-sup}
  \operatorname{TV}(P_0,P_T)
  \le
  \frac{T}{2}
  \sup_{0\le t\le T}\sqrt{J_t}.
\end{equation}
\end{lemma}

\begin{proof}
The standard inequality $\operatorname{TV}\le H$ and the length bound for
the absolutely continuous path $t\mapsto\sqrt{p_t}$ give
\[
  \operatorname{TV}(P_0,P_T)
  \le
  \left\|\sqrt{p_T}-\sqrt{p_0}\right\|_{L^2(\mu)}
  \le
  \int_0^T
  \left\|\frac{d}{dt}\sqrt{p_t}\right\|_{L^2(\mu)}\,dt
  =
  \frac12\int_0^T\sqrt{J_t}\,dt.
\]
The final inequality in \eqref{eq:hellinger-path-bound} follows by taking
the supremum of $\sqrt{J_t}$ over $[0,T]$.
\end{proof}

\paragraph{Bound posterior mean by mutual information.}

\begin{lemma}
\label{lem:spherical-posterior}
Let $\bm z$ be uniform on $\mathbb S^{M-1}$ and let $\bm Y$ be any
observation jointly distributed with $\bm z$. Then
\[
  \E\bigl\|\E[\bm z\mid\bm Y]\bigr\|^2
  \le
  \frac{C}{M}\operatorname{I}(\bm z;\bm Y).
\]
In the Gaussian mean channel
$\bm Y=\lambda\bm z+\sigma\bm g$, where
$\bm g\sim\mathcal N(\bm 0,\bm I_M)$ is independent of $\bm z$, the sharper
bound
$\E\bigl\|\E[\bm z\mid\bm Y]\bigr\|^2 \le 9\frac{\lambda^2}{M\sigma^2}$
holds.
\end{lemma}
\begin{proof}
Let $\nu_M$ denote the uniform probability measure on
$\mathbb S^{M-1}$. Its one-dimensional marginals are strictly
$1/M$-sub-Gaussian \cite[Corollary~1]{marchal2017subgaussian}. Thus, by
rotational invariance,
\[
  \log\int \exp(\langle\bm t,\bm u\rangle)\,\nu_M(d\bm u)
  \le \frac{\|\bm t\|^2}{2M},
  \qquad \bm t\in\mathbb R^M.
\]
Indeed, for $M\ge2$ a rescaled spherical coordinate has the symmetric
$\operatorname{Beta}((M-1)/2,(M-1)/2)$ distribution; for $M=1$ this is
the Rademacher bound.

For an observation value $y$, let $\Pi_y$ be the conditional law of
$\bm z$ given $\bm Y=y$, set
$\bm m_y\coloneqq\int\bm u\,\Pi_y(d\bm u)$, and write
$D_y\coloneqq\operatorname{KL}(\Pi_y\,\|\,\nu_M)$. The
Donsker--Varadhan variational inequality states that
\[
  \int f\,d\Pi_y
  \le D_y+\log\int e^f\,d\nu_M.
\]
Applying it with $f(\bm u)=M\langle\bm m_y,\bm u\rangle$ gives
$M\|\bm m_y\|^2\le D_y+M\|\bm m_y\|^2/2$, and therefore
$\|\bm m_y\|^2\le2D_y/M$; the inequality is automatic if $D_y=\infty$.
Averaging over $\bm Y$ and using
$\E D_{\bm Y}=\operatorname{I}(\bm z;\bm Y)$ proves the first assertion
with $C=2$.

For the Gaussian channel, let
$Q_{\bm u}=\mathcal N(\lambda\bm u,\sigma^2\bm I_M)$, let $Q$ be the
marginal law of $\bm Y$, and set
$Q_0=\mathcal N(\bm 0,\sigma^2\bm I_M)$. The information-radius identity
and the Gaussian shift formula give
\[
  \operatorname{I}(\bm z;\bm Y)+\operatorname{KL}(Q\,\|\,Q_0)
  =\int\operatorname{KL}(Q_{\bm u}\,\|\,Q_0)\,\nu_M(d\bm u)
  =\frac{\lambda^2}{2\sigma^2}.
\]
Hence $\operatorname{I}(\bm z;\bm Y)\le\lambda^2/(2\sigma^2)$, so the
first part yields the stronger bound
\[
  \E\bigl\|\E[\bm z\mid\bm Y]\bigr\|^2
  \le \frac{\lambda^2}{M\sigma^2}.
\]
The stated bound follows.
\end{proof}

\newpage 

\section{Proof of \cref{thm:lower}}
To simplify notation, we use $\mathcal{R}_{\rm minimax}$ to denote the minimax risk. 
It suffices to prove the following three separate lower bounds:
\begin{subequations}
\label{eq:minimax-lower-appendix}
\begin{align}
  \mathcal{R}_{\rm minimax}
  &\ge
  c\min\left\{1,\,
  \frac{\sigma\sqrt n}{\gamma_{\max}}
  \right\},
  \label{eq:minimax-lower-theta-independent}
  \\
  \mathcal{R}_{\rm minimax}
  &\ge
  c\min\left\{1,\,
  \frac{\sigma(\sqrt{r+ (r_1\wedge r_2)})}
       {\gamma_{\min}\sqrt\theta}
  \right\},
  \label{eq:minimax-lower-sqrt-theta}
  \\
  \mathcal{R}_{\rm minimax}
  &\ge
  c\min\left\{1,\,
  \frac{\sigma^2\sqrt n(\sqrt{r+ (r_1\wedge r_2)})}
       {\gamma_{\min}^2\theta}
  \right\}.
  \label{eq:minimax-lower-theta}
\end{align}
\end{subequations}
The first term is independent of $\theta$ and captures the difficulty of estimating the shared subspace even when the view-specific subspaces are fixed and known. The remaining two terms capture the additional difficulty of distinguishing the shared directions from the view-specific directions when the two view-specific subspaces are nearly aligned. We prove the three bounds in order.

\subsection{Proof of~\cref{eq:minimax-lower-theta-independent}}

We restrict the parameter class to a family in which both view-specific
subspaces and $r-1$ shared directions are fixed.  Only the remaining shared
direction varies.  We construct exponentially many such directions whose
projectors are separated by order $\delta$, and then choose Gaussian loadings
for which the corresponding observation laws have relative entropy of order
$\delta^2\gamma_{\max}^2/\sigma^2$.  Fano's lemma will then give the claimed
rate.

\paragraph{A packing of shared subspaces.}
Choose a fixed matrix with orthonormal columns:
\[
  \bigl[
    \underbrace{\bm U_0}_{n\times(r-1)}
    \ \underbrace{\bm V_1}_{n\times r_1}
    \ \underbrace{\bm V_2}_{n\times r_2}
    \ \underbrace{\bm u_0}_{n\times 1}
  \bigr]
  \in\St(n,r+r_1+r_2).
\]
We use $\bm U_0$ to form the fixed part of the shared subspace,
$\bm V_1$ and $\bm V_2$ to form the fixed view-specific subspaces,
and $\bm u_0$ as the baseline for the remaining shared direction.

The orthogonal complement of the displayed columns has dimension
$n-r-r_1-r_2\asymp n$. A standard packing argument therefore gives unit
vectors $\bm z_1,\ldots,\bm z_N$ in this orthogonal complement such that
\begin{equation}
\label{eq:ambient-packing}
  \log N\ge cn,
  \qquad
  \|\bm z_i-\bm z_j\|_2\ge c
  \quad\text{for all }i\ne j.
\end{equation}
For $0<\delta\le1/2$, set
$\bm u_j\coloneqq\sqrt{1-\delta^2}\,\bm u_0+\delta\bm z_j$,
$\bm U_j\coloneqq[\bm U_0,\bm u_j]$, and
$\bm P_j\coloneqq\bm U_j\bm U_j^\trans$.  The choice of the
$\bm z_j$ ensures that every $\bm U_j$ has orthonormal columns.  Moreover,
$\bm P_i-\bm P_j
=\bm u_i\bm u_i^\trans-\bm u_j\bm u_j^\trans$ and
$1-\bm z_i^\trans\bm z_j
=\frac12\|\bm z_i-\bm z_j\|_2^2\asymp1$.  The formula for the distance
between two rank-one projectors therefore gives
\begin{equation}
\label{eq:fano-target-separation}
\begin{aligned}
  \opn{\bm P_i-\bm P_j}^2
  &=
  1-(\bm u_i^\trans\bm u_j)^2 \\
  &=
  \delta^2(1-\bm z_i^\trans\bm z_j)
  \left\{
    2-\delta^2(1-\bm z_i^\trans\bm z_j)
  \right\}
  \asymp\delta^2,
  \qquad
  \frn{\bm P_i-\bm P_j}^2
  =
  2\opn{\bm P_i-\bm P_j}^2.
\end{aligned}
\end{equation}
Thus the targets are separated by order $\delta$ in operator norm, while
their squared Frobenius distance is of order $\delta^2$, as required for
the information calculation below.

\paragraph{Admissible priors and unconditioned comparison laws.}
The loadings are randomized so that the unconditioned observation laws are
centered Gaussian, while conditioning on their smallest singular values
ensures that every realization belongs to the parameter class.  Independently
for $k=1,2$, let
$\bm G_k\in\mathbb R^{d_k\times p_k}$ have independent
$\mathcal N(0,d_k^{-1})$ entries, and define
\[
  \mathcal E_k^{\rm load}
  \coloneqq
  \left\{
    \bm G_k^\trans\bm G_k\succeq c_B\bm I_{p_k}
  \right\},
  \qquad
  \bm B_k
  \coloneqq
  \frac{s_k}{\sqrt{c_B}}\bm G_k,
  \qquad
  \bm X_k^{(j)}
  \coloneqq
  [\bm U_j,\bm V_k]\bm B_k^\trans.
\]
For the prior indexed by $j$, condition jointly on
$\mathcal E_1^{\rm load}\cap\mathcal E_2^{\rm load}$.  On this event,
$\sigma_{p_k}(\bm B_k)\ge s_k$.  In addition,
$[\bm U_j,\bm V_k]\in\operatorname{St}(n,p_k)$ and
$\bm V_1^\trans\bm V_2=\bm 0$, so the resulting prior is supported on
$\mathfrak P$ and its target is the fixed projector $\bm P_j$.  Let
$\mathsf P_j$ denote the induced observation law.

Let $\mathsf Q_j$ denote the observation law obtained from the same
construction without conditioning on $\mathcal E_1^{\rm load}\cap\mathcal E_2^{\rm load}$.
The Gaussian smallest-singular-value bound and contraction of total
variation under the observation channel give
\begin{equation}
\label{eq:fano-comparison-tv}
  \max_{j\in[N]}
  \operatorname{TV}(\mathsf P_j,\mathsf Q_j)
  \le
  \mathbb P\left(
  \left(\mathcal E_1^{\mathrm{load}}\right)^c
  \cup
  \left(\mathcal E_2^{\mathrm{load}}\right)^c
\right)
  \le\frac18.
\end{equation}
It remains to control the relative entropy between the unconditioned laws.

\paragraph{Applying Fano's bound.}
Under $\mathsf Q_j$, the columns in view $k$ are independent centered
Gaussian vectors with covariance
$\bm\Sigma_{k,j}
=\sigma^2\{\bm I_n+\beta_k
(\bm U_j\bm U_j^\trans+\bm V_k\bm V_k^\trans)\}$, where
$\beta_k\coloneqq s_k^2/(c_B d_k\sigma^2)$. Since all fixed directions
cancel when comparing $j$ with $1$, \cref{eq:gaussian-projector-kl} and
\cref{eq:fano-target-separation} give
\begin{align}
  \operatorname{KL}(\mathsf Q_j\,\|\,\mathsf Q_1)
  &=
  \sum_{k=1}^2
  \frac{d_k\beta_k^2}{4(1+\beta_k)}
  \frn{
    \bm u_j\bm u_j^\trans-\bm u_1\bm u_1^\trans
  }^2
  \notag\\
  &\lesssim
  \delta^2
  \sum_{k=1}^2
  \frac{d_k\beta_k^2}{1+\beta_k}
  \lesssim
  \delta^2\frac{\gamma_{\max}^2}{\sigma^2}.
\label{eq:ambient-gaussian-kl}
\end{align}

We now choose
\[
  \delta
  \coloneqq
  c_0\left(
    1\wedge\frac{\sigma\sqrt n}{\gamma_{\max}}
  \right),
\]
where $c_0>0$ is sufficiently small. Then
\cref{eq:ambient-gaussian-kl} is at most a sufficiently small multiple
of $n$, and hence of $\log N$ by
\cref{eq:ambient-packing}.  Apply
\cref{lem:fano-proxy} with the operator-norm distance, reference law
$\mathsf Q_1$, target separation $c\delta$ from
\cref{eq:fano-target-separation}, and comparison error at most $1/8$ from
\cref{eq:fano-comparison-tv}.  Fano's lemma gives a minimax risk lower
bound in~\cref{eq:minimax-lower-theta-independent}.

\subsection{A common construction for the two $\theta$-dependent terms}

We now construct a single family of target subspaces for the two lower bounds
that depend on $\theta$.  By exchanging the labels of the views, we may assume
that $\gamma_2=\gamma_{\min}$. 
Put $q\coloneqq r_1\wedge r_2$.  The
construction varies the decomposition of a fixed $(r+q)$-dimensional
subspace into $r$ shared directions and $q$ view-specific directions.  View~1
will have the same distribution under every hypothesis, so only View~2 can
distinguish the resulting target projectors.  The two subsequent proofs differ
only in whether certain view-specific directions in View~2 are fixed or
randomized.

\paragraph{The target subspaces.}
Set $m\coloneqq n-r-r_1-r_2+q$.  The rank assumptions imply
$r+r_1+r_2\le2n/3$, and hence $m\ge q$.  We may therefore choose
\[
  \bigl[
    \underbrace{\bm F}_{n\times(r+q)}
    \ \underbrace{\bm D_1}_{n\times(r_1-q)}
    \ \underbrace{\bm D_2}_{n\times(r_2-q)}
    \ \underbrace{\bm H}_{n\times m}
  \bigr]
  \in\operatorname{St}(n,n),
\]
where a block with zero columns is omitted.  The shared subspace will vary
inside $\operatorname{col}(\bm F)$.  The matrices $\bm D_1$ and $\bm D_2$
supply the remaining $r_1-q$ and $r_2-q$ view-specific directions, while
$\operatorname{col}(\bm H)$ contains the additional View-2 directions used
to satisfy the separation condition.

For $\bm T\in\mathbb R^{q\times r}$, the columns of
$[\bm I_r,\bm T^\trans]^\trans$ span an $r$-dimensional subspace of
$\mathbb R^{r+q}$.  We orthonormalize these columns and map them into
$\mathbb R^n$ by defining
\begin{equation}
\label{eq:lower-shared-basis}
  \bm U_{\bm T}^{\circ}
  \coloneqq
  \begin{bmatrix}
    \bm I_r\\
    \bm T
  \end{bmatrix}
  (\bm I_r+\bm T^\trans\bm T)^{-1/2},
  \qquad
  \bm U_{\bm T}
  \coloneqq
  \bm F\bm U_{\bm T}^{\circ}.
\end{equation}
The corresponding target is the orthogonal projector
\begin{equation}
\label{eq:lower-target-projector}
  \bm P_{\bm T}
  \coloneqq
  \bm U_{\bm T}\bm U_{\bm T}^\trans
  =
  \bm F
  \begin{bmatrix}
    \bm I_r\\
    \bm T
  \end{bmatrix}
  (\bm I_r+\bm T^\trans\bm T)^{-1}
  \begin{bmatrix}
    \bm I_r&\bm T^\trans
  \end{bmatrix}
  \bm F^\trans.
\end{equation}
Thus $\bm T$ is only a coordinate used to specify the shared subspace;
$\bm P_{\bm T}$ is the quantity being estimated.

To construct the associated view-specific directions, define
\begin{equation}
\label{eq:lower-graph-R}
  \bm R_{\bm T}^{\circ}
  \coloneqq
  \begin{bmatrix}
    -\bm T^\trans\\
    \bm I_q
  \end{bmatrix}
  (\bm I_q+\bm T\bm T^\trans)^{-1/2}.
\end{equation}
The columns of $\bm R_{\bm T}^{\circ}$ form an orthonormal basis for the
orthogonal complement of
$\operatorname{col}(\bm U_{\bm T}^{\circ})$ in $\mathbb R^{r+q}$.
Indeed, the inverse-square-root factors orthonormalize the two matrices,
and direct multiplication shows that their column spaces are orthogonal.
Consequently,
$[\bm U_{\bm T}^{\circ},\bm R_{\bm T}^{\circ}]
\in\mathbb O(r+q)$.

\paragraph{The Assouad family.}
We choose the matrices $\bm T$ so that the family has $r\vee q$ binary
coordinates and every nonzero difference has rank one.  For
$\bm\omega\in\{-1,1\}^{r\vee q}$, define
\[
  \bm T_{\bm\omega}
  \coloneqq
  \begin{cases}
    \delta\bm e_1\bm\omega^\trans,
      & r\ge q,\\[2mm]
    \delta\bm\omega\bm e_1^\trans,
      & q>r,
  \end{cases}
\]
where $\bm e_1\in\mathbb R^q$ in the first case and
$\bm e_1\in\mathbb R^r$ in the second.  Since every nonzero difference
has rank one, its operator and Frobenius norms agree.  Therefore
\begin{equation}
\label{eq:lower-cube-distance}
  \opn{\bm T_{\bm\omega}}
  =
  \delta\sqrt{r\vee q},
  \qquad
  \opn{
    \bm T_{\bm\omega}-\bm T_{\bm\omega'}
  }
  =
  \frn{
    \bm T_{\bm\omega}-\bm T_{\bm\omega'}
  }
  =
  2\delta
  \sqrt{\operatorname{Ham}(\bm\omega,\bm\omega')}.
\end{equation}
The rank-one choice serves two purposes.  It makes the operator-norm
separation proportional to the square root of the Hamming distance, and it
ensures that neighboring hypotheses differ in only one shared direction and
one direction in its orthogonal complement.

The following deterministic estimate transfers separation between the
coordinate matrices to separation between the target projectors.  We state it
here because it is used by both $\theta$-dependent lower bounds and defer its
proof to the subsection containing the proofs of the auxiliary lemmas.

\begin{lemma}
\label{lem:lower-projector-separation}
For $\bm T,\bm S\in\mathbb R^{q\times r}$,
\[
  \opn{\bm P_{\bm T}-\bm P_{\bm S}}
  \ge
  \frac{
    \opn{\bm T-\bm S}
  }{
    \sqrt{1+\opn{\bm T}^2}
    \sqrt{1+\opn{\bm S}^2}
  }.
\]
In particular, if
$\|\bm T\|_{\mathrm{op}}\vee\|\bm S\|_{\mathrm{op}}\le1/2$, then
\[
  \opn{\bm P_{\bm T}-\bm P_{\bm S}}
  \ge
  \frac45\opn{\bm T-\bm S}.
\]
\end{lemma}

We will choose $\delta$ so that
$\delta\sqrt{r\vee q}\le1/2$.  Applying
\cref{lem:lower-projector-separation} and
\cref{eq:lower-cube-distance} then gives
\begin{equation}
\label{eq:lower-target-separation}
  \opn{
    \bm P_{\bm T_{\bm\omega}}
    -
    \bm P_{\bm T_{\bm\omega'}}
  }
  \ge
  \delta
  \sqrt{\operatorname{Ham}(\bm\omega,\bm\omega')}.
\end{equation}

The rank-one property will also be essential when comparing neighboring
observation laws.  For neighboring vertices,
$\bm T_{\bm\omega'}-\bm T_{\bm\omega}
=\kappa\bm a\bm b^\trans$ for some unit vectors $\bm a$ and $\bm b$.
The difference vanishes on $\bm b^\perp$, so the corresponding shared
subspaces have $r-1$ directions in common.  Its transpose vanishes on
$\bm a^\perp$, so their orthogonal complements in
$\operatorname{col}(\bm F)$ have $q-1$ directions in common.
Consequently, only one shared direction and one orthogonal direction change,
and the two pairs are related by a rotation in a two-dimensional subspace.
The neighboring-law lemmas use this reduction to remove all directions whose
distributions are identical under the two hypotheses.

\paragraph{View~1 is identical under all hypotheses.}
For each $\bm T$, define
\begin{equation}
\label{eq:lower-view-one-frames}
  \bm V_{1,\bm T}
  \coloneqq
  [\bm F\bm R_{\bm T}^{\circ},\bm D_1]
  \in\operatorname{St}(n,r_1).
\end{equation}
The matrices $[\bm U_{\bm T},\bm V_{1,\bm T}]$ and
$[\bm F,\bm D_1]$ are orthonormal bases of the same subspace.  We can
therefore choose a first-view signal matrix that is independent of $\bm T$.
Fix $\bm W_1\in\operatorname{St}(d_1,p_1)$ and set
\begin{equation}
\label{eq:lower-fixed-first-view}
  \bm X_1
  \coloneqq
  s_1[\bm F,\bm D_1]\bm W_1^\trans.
\end{equation}
More precisely, there is an orthogonal matrix
$\bm O_{\bm T}\in\mathbb O(p_1)$ such that
$[\bm U_{\bm T},\bm V_{1,\bm T}]
=[\bm F,\bm D_1]\bm O_{\bm T}$.  Taking
$\bm B_{1,\bm T}\coloneqq s_1\bm W_1\bm O_{\bm T}$ represents
\cref{eq:lower-fixed-first-view} in the model, and every singular value of
$\bm B_{1,\bm T}$ equals $s_1$.  Hence the distribution of $\bm Y_1$ is
identical for every $\bm T$ and contains no information for distinguishing
the target projectors.

\paragraph{View~2 satisfies the separation condition.}
Let
$\rho\coloneqq1-2\theta$ and
$\rho_\perp\coloneqq\sqrt{1-\rho^2}$.  For
$\bm Z\in\operatorname{St}(m,q)$, define
\begin{equation}
\label{eq:lower-view-two-frame}
  \bm V_{2,\bm T,\bm Z}
  \coloneqq
  \left[
    \rho\bm F\bm R_{\bm T}^{\circ}
    +
    \rho_\perp\bm H\bm Z,\,
    \bm D_2
  \right]
  \in\operatorname{St}(n,r_2).
\end{equation}
The matrices $\bm F\bm R_{\bm T}^{\circ}$ and $\bm H\bm Z$ have
orthonormal columns and are mutually orthogonal.  Since
$\rho^2+\rho_\perp^2=1$, the first $q$ columns of
$\bm V_{2,\bm T,\bm Z}$ are orthonormal.  They are also orthogonal to
$\bm U_{\bm T}$ and $\bm D_2$, so
$[\bm U_{\bm T},\bm V_{2,\bm T,\bm Z}]$ has orthonormal columns.

The alignment with the first view is
\[
  \bm V_{1,\bm T}^\trans
  \bm V_{2,\bm T,\bm Z}
  =
  \begin{bmatrix}
    \rho\bm I_q&\bm 0\\
    \bm 0&\bm 0
  \end{bmatrix}.
\]
It follows that
\[
  \opn{
    \bm V_{1,\bm T}^\trans
    \bm V_{2,\bm T,\bm Z}
  }
  =
  \rho
  =
  1-2\theta,
\]
so the separation condition holds for every $\bm T$ and every $\bm Z$.
The target $\bm P_{\bm T}$ does not depend on $\bm Z$.  Thus $\bm Z$
may be fixed in one lower-bound construction and randomized in the other
without changing the target.

\paragraph{Reduction to neighboring laws.}
For either of the second-view loading constructions below, let
$\mathsf P_{\bm\omega}$ denote the observation law associated with
$\bm T_{\bm\omega}$.  Suppose that the priors are supported on
$\mathfrak P$ and that
\[
  \max_{\substack{
    \bm\omega\in\{-1,1\}^{r\vee q}\\
    j\in[r\vee q]
  }}
  \operatorname{TV}\bigl(
    \mathsf P_{\bm\omega},
    \mathsf P_{\bm\omega^{(j)}}
  \bigr)
  \le\frac14,
\]
where $\bm\omega^{(j)}$ is obtained by changing the $j$th coordinate of
$\bm\omega$.  Then \cref{eq:lower-target-separation,lem:assouad} imply
\begin{equation}
\label{eq:lower-assouad-reduction}
  \mathcal{R}_{\rm minimax}
  \gtrsim
  \delta\sqrt{r\vee q}.
\end{equation}
It therefore remains only to choose the second-view loading and $\delta$ so
that neighboring observation laws have total variation at most $1/4$.
Fixing $\bm Z$ gives the $\theta^{-1/2}$ term, whereas mixing over a
Haar-distributed $\bm Z$ gives the $\theta^{-1}$ term.

\begin{proof}[Proof of~\cref{lem:lower-projector-separation}]
   To verify this claim, consider two distinct vertices and write
$\bm T\coloneqq\bm T_{\bm\omega}$ and
$\bm S\coloneqq\bm T_{\bm\omega'}$.  Choose a unit vector
$\bm x\in\mathbb R^r$ satisfying
$\|(\bm T-\bm S)\bm x\|_2
=\|\bm T-\bm S\|_{\mathrm{op}}$, and consider the unit vector
\[
  \bm v
  \coloneqq
  \frac{1}{\sqrt{1+\|\bm T\bm x\|_2^2}}
  \begin{bmatrix}
    \bm x\\
    \bm T\bm x
  \end{bmatrix}
  \in\col(\bm U_{\bm T}^{\circ}).
\]
Because $\bm F$ has orthonormal columns, the distance between the target
projectors equals the distance between the corresponding projectors in
$\mathbb R^{r+q}$.  Moreover,
$\bm R_{\bm S}^{\circ}$ spans the orthogonal complement of
$\col(\bm U_{\bm S}^{\circ})$.  Applying the difference of the two
projectors to $\bm v$ therefore gives
\begin{align*}
  \opn{\bm P_{\bm T}-\bm P_{\bm S}}
  &\ge
  \|(\bm R_{\bm S}^{\circ})^\trans\bm v\|_2\\
  &=
  \frac{
    \|(\bm I_q+\bm S\bm S^\trans)^{-1/2}
      (\bm T-\bm S)\bm x\|_2
  }{
    \sqrt{1+\|\bm T\bm x\|_2^2}
  }\\
  &\ge
  \frac{
    \opn{\bm T-\bm S}
  }{
    \sqrt{1+\opn{\bm S}^2}
    \sqrt{1+\opn{\bm T}^2}
  }
  \ge
  \frac45\opn{\bm T-\bm S}.
\end{align*}
The penultimate inequality uses the smallest eigenvalue of
$(\bm I_q+\bm S\bm S^\trans)^{-1/2}$ and the choice of $\bm x$.
The final inequality follows from
$\|\bm T\|_{\mathrm{op}}\vee\|\bm S\|_{\mathrm{op}}\le1/2$. 
\end{proof}

\subsection{Proof of~\cref{eq:minimax-lower-sqrt-theta}}

For this bound, we fix the matrix $\bm Z$ in
\cref{eq:lower-view-two-frame}.  The additional View-2 directions therefore
have a known orientation, and the leading change in the second-view covariance
is proportional to $\rho_\perp$.  We construct an admissible Gaussian-loading
prior and state the neighboring-law estimate at the scale
$\delta\asymp\sigma/(\gamma_2\rho_\perp)$.  The common Assouad reduction in
\cref{eq:lower-assouad-reduction} will then give the result.

\paragraph{An admissible prior for View~2.}
Fix $\bm Z_0\in\operatorname{St}(m,q)$.  Let
$\bm G_2\in\mathbb R^{d_2\times p_2}$ have independent
$\mathcal N(0,d_2^{-1})$ entries, and define
\[
  \mathcal E_2^{\rm load}
  \coloneqq
  \left\{
    \bm G_2^\trans\bm G_2
    \succeq
    c_B\bm I_{p_2}
  \right\},
  \qquad
  \bm B_2
  \coloneqq
  \frac{s_2}{\sqrt{c_B}}\bm G_2.
\]
For each matrix $\bm T$ in the Assouad family, condition on
$\mathcal E_2^{\rm load}$ and set
\begin{equation}
\label{eq:lower-regular-signal}
  \bm X_{2,\bm T}
  \coloneqq
  [\bm U_{\bm T},\bm V_{2,\bm T,\bm Z_0}]
  \bm B_2^\trans.
\end{equation}
On $\mathcal E_2^{\rm load}$, the loading satisfies
$\sigma_{p_2}(\bm B_2)\ge s_2$.  Moreover,
$[\bm U_{\bm T},\bm V_{2,\bm T,\bm Z_0}]$ has orthonormal columns and
the separation condition was verified in the preceding subsection.  Together
with the fixed first-view signal in
\cref{eq:lower-fixed-first-view}, this defines a prior supported on
$\mathfrak P$ whose target is $\bm P_{\bm T}$.  Let
$\mathsf P_{\bm T}$ denote the induced observation law.

It remains to show that the laws corresponding to neighboring matrices in
the Assouad family are difficult to distinguish.  View~1 has the same
distribution under every hypothesis, so only View~2 contributes to their
total variation distance.  After removing the conditioning on
$\mathcal E_2^{\rm load}$, the columns of View~2 are independent centered Gaussian
vectors, and changing $\bm T$ changes only the projector onto the complete
second-view signal subspace.  The following lemma gives the required
comparison.  Its proof is deferred to the subsection containing the proofs
of the neighboring-law lemmas.

\begin{lemma}
\label{lem:lower-regular-edge}
There is a numerical constant $c>0$ such that the following holds.  Let
$\bm\omega$ and $\bm\omega'$ be neighboring vertices of the Assouad family,
and suppose that
\[
  \delta\sqrt{r\vee q}\le\frac12,
  \qquad
  \delta
  \le
  c\frac{\sigma}{\gamma_2\rho_\perp}.
\]
Then
\[
  \operatorname{TV}\bigl(
    \mathsf P_{\bm T_{\bm\omega}},
    \mathsf P_{\bm T_{\bm\omega'}}
  \bigr)
  \le\frac14.
\]
\end{lemma}

Choose
\[
  \delta
  \coloneqq
  c_0
  \left\{
    \frac1{\sqrt{r\vee q}}
    \wedge
    \frac{\sigma}{\gamma_2\rho_\perp}
  \right\},
\]
where $c_0>0$ is sufficiently small.  The first term ensures that
$\delta\sqrt{r\vee q}\le1/2$, while the second permits the application of
\cref{lem:lower-regular-edge} to every neighboring pair.  Consequently,
\cref{eq:lower-assouad-reduction} gives a minimax risk lower bound of order
\[
  \delta\sqrt{r\vee q}
  \asymp
  1\wedge
  \frac{
    \sigma\sqrt{r\vee q}
  }{
    \gamma_2\rho_\perp
  }
  \asymp
  1\wedge
  \frac{
    \sigma(\sqrt r+\sqrt q)
  }{
    \gamma_{\min}\sqrt\theta
  }.
\]
The final comparison uses
$\gamma_2=\gamma_{\min}$,
$\rho_\perp^2=4\theta(1-\theta)\asymp\theta$, and
$\sqrt{r\vee q}\asymp\sqrt r+\sqrt q$.  Since
$q=r_1\wedge r_2$, this proves
\cref{eq:minimax-lower-sqrt-theta}.

\subsection{Proof of~\cref{eq:minimax-lower-theta}}

For the third term, we randomize the matrix $\bm Z$ in
\cref{eq:lower-view-two-frame}.  Specifically, let
\[
  \bm Z
  \sim
  \operatorname{Haar}\{\operatorname{St}(m,q)\},
\]
independently of the loading and observation noise.  This is a prior on the
view-specific subspace only: for each $\bm T$, the target remains the fixed
projector $\bm P_{\bm T}$.  The purpose of the randomization is to make the
orientation of the additional View-2 directions unobserved.  A comparison
conditional on the entire matrix $\bm Z$ would give only the
$\theta^{-1/2}$ rate obtained in the preceding subsection.

Set
\[
  M
  \coloneqq
  m-q+1
  =
  n-r-r_1-r_2+1.
\]
The rank assumptions imply $M\asymp n$.  In the proof of the neighboring-law
bounds, we condition on $q-1$ columns of $\bm Z$.  The remaining column is
then uniform on a sphere in an $M$-dimensional subspace.  The resulting
factor $M^{-1}$ in the information bound is the source of the factor
$\sqrt n$ in the final rate.

We use different loading distributions according to whether
$s_2^2\le c_Bd_2\sigma^2$.  Below this threshold, a conditioned Gaussian
loading leads to a covariance model with a bounded covariance ratio.  Above
the threshold, we instead use a deterministic loading with orthonormal
columns.  This distinction is technical: both constructions give the same
neighboring-law bound at the scale
\[
  \delta
  \lesssim
  \frac{\sigma^2\sqrt M}
       {\gamma_2^2\rho_\perp^2}.
\]

\paragraph{The low-loading regime.}
Suppose first that $s_2^2\le c_Bd_2\sigma^2$.  Independently of $\bm Z$,
let $\bm G_2\in\mathbb R^{d_2\times p_2}$ have independent
$\mathcal N(0,d_2^{-1})$ entries, and define
\[
  \mathcal E_2^{\rm load}
  \coloneqq
  \left\{
    \bm G_2^\trans\bm G_2
    \succeq
    c_B\bm I_{p_2}
  \right\},
  \qquad
  \bm B_2
  \coloneqq
  \frac{s_2}{\sqrt{c_B}}\bm G_2.
\]
For each $\bm T$ in the Assouad family, condition on $\mathcal E_2^{\rm load}$ and set
\begin{equation}
\label{eq:lower-hidden-low-signal}
  \bm X_{2,\bm T}
  \coloneqq
  [\bm U_{\bm T},\bm V_{2,\bm T,\bm Z}]
  \bm B_2^\trans.
\end{equation}
Every realization satisfies
$\sigma_{p_2}(\bm B_2)\ge s_2$, and the separation condition holds for every
realization of $\bm Z$.  Together with the fixed first-view signal in
\cref{eq:lower-fixed-first-view}, this defines a prior supported on
$\mathfrak P$.  Let $\mathsf P_{\bm T}$ denote the observation law after
integrating over both the conditioned loading and $\bm Z$.

The next lemma controls the total variation distance between neighboring
laws.  Its proof first removes the conditioning on $\mathcal E_2^{\rm load}$ and then
compares the resulting Gaussian covariance mixtures after integrating over
the unobserved direction.  We defer the proof to the subsection containing
the neighboring-law calculations.

\begin{lemma}
\label{lem:lower-hidden-low-edge}
There is a numerical constant $c>0$ such that the following holds.  Let
$\bm\omega$ and $\bm\omega'$ be neighboring vertices of the Assouad family.
Suppose that
\[
  s_2^2\le c_Bd_2\sigma^2,
  \qquad
  \delta\sqrt{r\vee q}\le\frac12,
  \qquad
  \delta
  \le
  c\frac{\sigma^2\sqrt M}
          {\gamma_2^2\rho_\perp^2}.
\]
Then
\[
  \operatorname{TV}\bigl(
    \mathsf P_{\bm T_{\bm\omega}},
    \mathsf P_{\bm T_{\bm\omega'}}
  \bigr)
  \le\frac14.
\]
\end{lemma}

\paragraph{The high-loading regime.}
Suppose now that $s_2^2>c_Bd_2\sigma^2$.  In this case, choose
$\bm W_{2,0}\in\operatorname{St}(d_2,r+q)$ and
$\bm W_{2,1}\in\operatorname{St}(d_2,r_2-q)$ such that
$\bm W_{2,0}^\trans\bm W_{2,1}=\bm0$, with the second matrix omitted when
$r_2=q$.  For each $\bm T$, define the deterministic loading
\begin{equation}
\label{eq:lower-hidden-high-loading}
  \bm B_{2,\bm T}
  \coloneqq
  s_2
  \left[
    \bm W_{2,0}
    [\bm U_{\bm T}^{\circ},\bm R_{\bm T}^{\circ}],
    \,
    \bm W_{2,1}
  \right]
  \in\mathbb R^{d_2\times p_2},
\end{equation}
and set
\begin{equation}
\label{eq:lower-hidden-high-signal}
  \bm X_{2,\bm T}
  \coloneqq
  [\bm U_{\bm T},\bm V_{2,\bm T,\bm Z}]
  \bm B_{2,\bm T}^\trans.
\end{equation}
The columns of
$\bm W_{2,0}[\bm U_{\bm T}^{\circ},\bm R_{\bm T}^{\circ}]$ and
$\bm W_{2,1}$ are orthonormal and mutually orthogonal.  Hence every
singular value of $\bm B_{2,\bm T}$ equals $s_2$, so every realization
again belongs to the parameter class.  Let $\mathsf P_{\bm T}$ now denote
the observation law obtained after integrating over $\bm Z$.

The factor
$[\bm U_{\bm T}^{\circ},\bm R_{\bm T}^{\circ}]$ in
\cref{eq:lower-hidden-high-loading} is essential.  It rotates the right
coordinates together with the shared and view-specific directions on the
left.  To see the resulting cancellation, consider the limiting case
$\rho=1$ and $\rho_\perp=0$.  Then
\[
  [\bm U_{\bm T},
    \bm F\bm R_{\bm T}^{\circ}]
  [\bm U_{\bm T}^{\circ},
    \bm R_{\bm T}^{\circ}]^\trans
  =
  \bm F,
\]
so the part of $\bm X_{2,\bm T}$ involving $\bm W_{2,0}$ is independent
of $\bm T$.  For general $\theta$, only terms proportional to
$1-\rho$ or $\rho_\perp$ remain when the signal is differentiated with
respect to the rotation between neighboring hypotheses.

The corresponding neighboring-law estimate has the same conclusion as in
the low-loading case.

\begin{lemma}
\label{lem:lower-hidden-high-edge}
There is a numerical constant $c>0$ such that the following holds.  Let
$\bm\omega$ and $\bm\omega'$ be neighboring vertices of the Assouad family.
Suppose that
\[
  s_2^2>c_Bd_2\sigma^2,
  \qquad
  \delta\sqrt{r\vee q}\le\frac12,
  \qquad
  \delta
  \le
  c\frac{\sigma^2\sqrt M}
          {\gamma_2^2\rho_\perp^2}.
\]
Then
\[
  \operatorname{TV}\bigl(
    \mathsf P_{\bm T_{\bm\omega}},
    \mathsf P_{\bm T_{\bm\omega'}}
  \bigr)
  \le\frac14.
\]
\end{lemma}

\paragraph{Choosing the separation.}
In either loading regime, choose
\[
  \delta
  \coloneqq
  c_0
  \left\{
    \frac1{\sqrt{r\vee q}}
    \wedge
    \frac{\sigma^2\sqrt M}
         {\gamma_2^2\rho_\perp^2}
  \right\},
\]
where $c_0>0$ is sufficiently small.  The first term ensures that
$\delta\sqrt{r\vee q}\le1/2$, while the second permits the application of
the appropriate neighboring-law lemma.  Therefore
\cref{eq:lower-assouad-reduction} gives a minimax risk lower bound of order
\[
  \delta\sqrt{r\vee q}
  \asymp
  1\wedge
  \frac{
    \sigma^2\sqrt{M(r\vee q)}
  }{
    \gamma_2^2\rho_\perp^2
  }
  \asymp
  1\wedge
  \frac{
    \sigma^2\sqrt n\,(\sqrt r+\sqrt q)
  }{
    \gamma_{\min}^2\theta
  }.
\]
The final comparison uses
$M\asymp n$, $\gamma_2=\gamma_{\min}$,
$\rho_\perp^2=4\theta(1-\theta)\asymp\theta$, and
$\sqrt{r\vee q}\asymp\sqrt r+\sqrt q$.  Since
$q=r_1\wedge r_2$, this proves \cref{eq:minimax-lower-theta}.

\subsection{Proofs of the neighboring-law lemmas}

Fix neighboring vertices $\bm\omega,\bm\omega'$ and write
$\bm T_0\coloneqq\bm T_{\bm\omega}$ and
$\bm T_1\coloneqq\bm T_{\bm\omega'}$. 
The construction of the Assouad family
gives
\begin{align*}
\|\bm T_0\|_{\mathrm{op}}=\|\bm T_1\|_{\mathrm{op}}
&=\delta\sqrt{r\vee q}\le1/2,\\
\|\bm T_0-\bm T_1\|_{\mathrm F}&=2\delta.
\end{align*}

In all three constructions, View~1 has the same distribution under $\bm T_0$ and $\bm T_1$ and is independent of View~2.  The total variation distance between the joint observation laws therefore equals the distance between their View-2 marginals, which we consider from now on. 

For the regular and hidden low-loading constructions, we next replace the loading conditioned on the loading-band event by its unconditioned Gaussian law and account for the resulting total variation error; the high-loading construction uses a deterministic loading and requires no such reduction.  It remains to compare the two View-2 laws.  The regular case involves product Gaussian laws with different covariance projectors, the hidden low-loading case involves Haar mixtures of Gaussian covariance laws, and the hidden high-loading case involves Haar mixtures of Gaussian mean laws.

For the two Gaussian-loading constructions, write
\begin{equation}
  \beta
  \coloneqq
  \frac{s_2^2}{c_Bd_2\sigma^2}.
\label{eq:lower-gaussian-loading-beta}
\end{equation}

\subsubsection{Proof of \cref{lem:lower-regular-edge}}

\paragraph{Remove the conditioning.}
Let $\mathsf Q_0$ and $\mathsf Q_1$ denote the unconditioned View-2 laws
corresponding to $\bm T_0$ and $\bm T_1$.  By
\cref{lem:gaussian-loading-lower-bound}, we have 
\[
  \operatorname{TV}\bigl(
    \mathsf P_{\bm T_0},
    \mathsf P_{\bm T_1}
  \bigr)
  \le
  \frac1{16}
  +
  \operatorname{TV}(\mathsf Q_0,\mathsf Q_1)
  +
  \frac1{16}.
\]
It is therefore enough to prove
$\operatorname{TV}(\mathsf Q_0,\mathsf Q_1)\le1/8$.

\paragraph{The second-view Gaussian laws.}

After integrating out the unconditioned Gaussian loading, the $d_2$ columns
of View~2 are independent centered Gaussian vectors.  Under $\mathsf Q_i$,
their common covariance is
$\sigma^2(\bm I_n+\beta\bm\Pi_i)$, where
\[
  \bm\Pi_i
  \coloneqq
  \bm P_{\operatorname{col}
  [\bm U_{\bm T_i},\bm V_{2,\bm T_i,\bm Z_0}]}.
\]
Pinsker's inequality, additivity of relative entropy over the independent
columns, and \cref{eq:gaussian-projector-kl} give
\begin{align}
  \operatorname{TV}^2(\mathsf Q_0,\mathsf Q_1)
  &\le
  \frac12
  \operatorname{KL}(\mathsf Q_0\,\|\,\mathsf Q_1)
  \notag\\
  &=
  \frac{d_2\beta^2}{8(1+\beta)}
  \|\bm\Pi_0-\bm\Pi_1\|_{\mathrm F}^2
  \lesssim
  \frac{\gamma_2^2}{\sigma^2}
  \|\bm\Pi_0-\bm\Pi_1\|_{\mathrm F}^2.
  \label{eq:lower-regular-TV-projector-reduction}
\end{align}
For the last inequality, the definitions of $\beta$ and $\gamma_2$ imply
\[
  \frac{d_2\beta^2}{1+\beta}
  =
  \frac{s_2^4}
       {c_B\sigma^2(s_2^2+c_Bd_2\sigma^2)}
  \lesssim
  \frac{s_2^4}
       {\sigma^2(s_2^2+d_2\sigma^2)}
  =
  \frac{\gamma_2^2}{\sigma^2}.
\]
Consequently, it remains only to bound the Frobenius distance between the
two covariance projectors.

\paragraph{Distance between the covariance projectors.}
Set
$\bm K_i\coloneqq
\bm R_{\bm T_i}^{\circ}
(\bm R_{\bm T_i}^{\circ})^\trans$.
Because
$[\bm U_{\bm T_i}^{\circ},\bm R_{\bm T_i}^{\circ}]$
is an orthogonal matrix, the projector onto
$\operatorname{col}(\bm U_{\bm T_i}^{\circ})$ is
$\bm I_{r+q}-\bm K_i$.  Expanding $\bm\Pi_i$ using
\cref{eq:lower-view-two-frame} and cancelling the terms that do not depend
on $i$ therefore gives
\begin{align}
  \bm\Pi_0-\bm\Pi_1
  ={}&
  -\rho_\perp^2
  \bm F(\bm K_0-\bm K_1)\bm F^\trans
  \notag\\
  &+
  \rho\rho_\perp
  \left\{
    \bm F
    (\bm R_{\bm T_0}^{\circ}-\bm R_{\bm T_1}^{\circ})
    \bm Z_0^\trans\bm H^\trans
    +
    \bm H\bm Z_0
    (\bm R_{\bm T_0}^{\circ}-\bm R_{\bm T_1}^{\circ})^\trans
    \bm F^\trans
  \right\}.
\label{eq:lower-regular-projector-difference}
\end{align}

We next compare the two matrices
$\bm R_{\bm T_i}^{\circ}$.  Define
$\bm L_i\coloneqq
(\bm I_q+\bm T_i\bm T_i^\trans)^{-1/2}$, so that
$\bm R_{\bm T_i}^{\circ}
=[-\bm T_i^\trans,\bm I_q]^\trans\bm L_i$.
For positive definite matrices $\bm A,\bm B\succeq\bm I_q$, the integral
representation of the inverse square root and the resolvent identity give
\(\|\bm A^{-1/2}-\bm B^{-1/2}\|_{\mathrm F}
\le \frac12\|\bm A-\bm B\|_{\mathrm F}\).  Moreover,
\[
  \|\bm T_0\bm T_0^\trans-\bm T_1\bm T_1^\trans\|_{\mathrm F}
  \le
  \bigl(
    \|\bm T_0\|_{\mathrm{op}}
    +
    \|\bm T_1\|_{\mathrm{op}}
  \bigr)
  \|\bm T_0-\bm T_1\|_{\mathrm F}
  \le
  \|\bm T_0-\bm T_1\|_{\mathrm F}.
\]
Using these bounds and
$\|\bm T_0\|_{\mathrm{op}}\vee\|\bm T_1\|_{\mathrm{op}}\le1/2$, we obtain
\begin{equation}
\begin{aligned}
  \|\bm L_0-\bm L_1\|_{\mathrm F}
  &\le
  \frac12\|\bm T_0-\bm T_1\|_{\mathrm F},\\
  \|\bm R_{\bm T_0}^{\circ}
      -\bm R_{\bm T_1}^{\circ}\|_{\mathrm F}
  &\le
  2\|\bm T_0-\bm T_1\|_{\mathrm F},\\
  \|\bm K_0-\bm K_1\|_{\mathrm F}
  &\le
  4\|\bm T_0-\bm T_1\|_{\mathrm F}.
\end{aligned}
\label{eq:lower-regular-frame-bounds}
\end{equation}
The last inequality follows by expanding
$\bm K_0-\bm K_1$ and using that both
$\bm R_{\bm T_0}^{\circ}$ and
$\bm R_{\bm T_1}^{\circ}$ have orthonormal columns.

The three terms in
\cref{eq:lower-regular-projector-difference} belong respectively to the
$\operatorname{col}(\bm F)$--$\operatorname{col}(\bm F)$,
$\operatorname{col}(\bm F)$--$\operatorname{col}(\bm H)$, and
$\operatorname{col}(\bm H)$--$\operatorname{col}(\bm F)$ matrix blocks.
They are therefore pairwise orthogonal in the Frobenius inner product.
Since $\bm F$ and $\bm H\bm Z_0$ have orthonormal columns,
\begin{align}
  \|\bm\Pi_0-\bm\Pi_1\|_{\mathrm F}^2
  &=
  \rho_\perp^4\|\bm K_0-\bm K_1\|_{\mathrm F}^2
  +
  2\rho^2\rho_\perp^2
  \|\bm R_{\bm T_0}^{\circ}
      -\bm R_{\bm T_1}^{\circ}\|_{\mathrm F}^2
  \notag\\
  &\le
  24\rho_\perp^2
  \|\bm T_0-\bm T_1\|_{\mathrm F}^2.
\label{eq:lower-regular-projector-bound}
\end{align}
\paragraph{Completing the proof.}
Combining this projector bound with
\cref{eq:lower-regular-TV-projector-reduction} and using
$\|\bm T_0-\bm T_1\|_{\mathrm F}=2\delta$ gives
\[
  \operatorname{TV}^2(\mathsf Q_0,\mathsf Q_1)
  \lesssim
  \frac{\gamma_2^2\rho_\perp^2\delta^2}{\sigma^2}.
\]
The assumption
$\delta\le c\sigma/(\gamma_2\rho_\perp)$ therefore implies
$\operatorname{TV}(\mathsf Q_0,\mathsf Q_1)\le1/8$ when the numerical
constant $c$ is sufficiently small.  Returning to the conditioning bound,
\[
  \operatorname{TV}(\mathsf P_{\bm T_0},\mathsf P_{\bm T_1})
  \le
  \frac1{16}+\frac18+\frac1{16}
  =\frac14.
\]
This proves \cref{lem:lower-regular-edge}.

\subsubsection{Proof of \cref{lem:lower-hidden-low-edge}}

\paragraph{Remove the conditioning.}
Let $\mathsf Q_i$ denote the unconditioned View-2 law corresponding to
$\bm T_i$, for $i\in\{0,1\}$.  By
\cref{lem:gaussian-loading-lower-bound},
\[
  \operatorname{TV}(\mathsf P_{\bm T_0},\mathsf P_{\bm T_1})
  \le
  \frac1{16}+\operatorname{TV}(\mathsf Q_0,\mathsf Q_1)+\frac1{16}.
\]
It is therefore enough to prove
$\operatorname{TV}(\mathsf Q_0,\mathsf Q_1)\le1/8$.  We reduce these
unconditioned laws to an $(M+2)$-dimensional covariance mixture and compare
the resulting laws along a Hellinger path.

\paragraph{The second-view covariance mixtures.}
For $\bm T\in\{\bm T_0,\bm T_1\}$ and
$\bm Z\in\operatorname{St}(m,q)$, define
\[
  \bm\Pi_{\bm T,\bm Z}
  \coloneqq
  \bm P_{\operatorname{col}
  [\bm U_{\bm T},\bm V_{2,\bm T,\bm Z}]}.
\]
After integrating out the Gaussian loading, the $d_2$ columns of View~2
are conditionally independent centered Gaussian vectors.  Thus
\begin{equation}
  \mathsf Q_i
  =
  \int
  \left[
    \mathcal N\left(
      \bm0,\,
      \sigma^2\{\bm I_n+\beta\bm\Pi_{\bm T_i,\bm Z}\}
    \right)
  \right]^{\otimes d_2}
  d\mu_{\mathrm{Haar}}(\bm Z),
\label{eq:lower-hidden-mixture-law}
\end{equation}
where $\beta$ is defined in
\cref{eq:lower-gaussian-loading-beta} and satisfies $\beta\le1$ in the
low-loading regime.  We next use Haar right invariance to give the directions
that agree under $\bm T_0$ and $\bm T_1$ the same coordinates.

\paragraph{Reduction to an \texorpdfstring{$(M+2)$}{(M+2)}-dimensional covariance mixture.}
Because $\bm\omega$ and $\bm\omega'$ are neighboring vertices,
\begin{equation}
  \bm T_1-\bm T_0
  =
  \kappa\bm a\bm b^\trans,
  \qquad
  \|\bm a\|_2=\|\bm b\|_2=1.
\label{eq:lower-hidden-rank-one-difference}
\end{equation}
Choose $\bm A_\perp\in\St(q,q-1)$ and
$\bm B_\perp\in\St(r,r-1)$ whose columns span $\bm a^\perp$ and
$\bm b^\perp$, respectively.  Then
\begin{equation}
  \begin{bmatrix}\bm I_r\\ \bm T_0\end{bmatrix}\bm B_\perp
  =
  \begin{bmatrix}\bm I_r\\ \bm T_1\end{bmatrix}\bm B_\perp,
  \qquad
  \begin{bmatrix}-\bm T_0^\trans\\ \bm I_q\end{bmatrix}\bm A_\perp
  =
  \begin{bmatrix}-\bm T_1^\trans\\ \bm I_q\end{bmatrix}\bm A_\perp.
\label{eq:lower-hidden-common-directions}
\end{equation}
Choose orthonormal matrices $\bm U_*\in\St(r+q,r-1)$ and
$\bm R_*\in\St(r+q,q-1)$ with column spaces equal to the two column spaces
in \cref{eq:lower-hidden-common-directions}.  For each $i\in\{0,1\}$,
choose unit vectors $\bm u_i,\bm r_i\in\mathbb R^{r+q}$ so that
\begin{equation}
  \operatorname{col}(\bm U_{\bm T_i}^{\circ})
  =
  \operatorname{col}[\bm U_*,\bm u_i],
  \qquad
  \operatorname{col}(\bm R_{\bm T_i}^{\circ})
  =
  \operatorname{col}[\bm R_*,\bm r_i].
\label{eq:lower-hidden-common-bases}
\end{equation}
The matrix $[\bm U_*,\bm u_i,\bm R_*,\bm r_i]$ is orthogonal.  Its first
$r-1$ and next $q-1$ columns do not depend on $i$, and hence
$\operatorname{col}[\bm u_0,\bm r_0]
=\operatorname{col}[\bm u_1,\bm r_1]$.

For each $i$, define the change-of-basis matrix
\begin{equation}
  \bm Q_i
  \coloneqq
  [\bm R_*,\bm r_i]^\trans\bm R_{\bm T_i}^{\circ}
  \in\Ogp(q),
  \qquad
  \bm R_{\bm T_i}^{\circ}\bm Q_i^\trans
  =
  [\bm R_*,\bm r_i].
\label{eq:lower-hidden-change-of-basis}
\end{equation}
Right multiplication by $\bm Q_i^\trans$ does not change a column
projector, and Haar right invariance gives
$\bm Z\bm Q_i^\trans\overset{\mathrm d}{=}\bm Z$.  Therefore the mixture
in \cref{eq:lower-hidden-mixture-law} is unchanged if
$\bm\Pi_{\bm T_i,\bm Z}$ is replaced by
\begin{align}
  \widetilde{\bm\Pi}_{i,\bm Z}
  \coloneqq{}&
  \bm F\bm U_*\bm U_*^\trans\bm F^\trans
  +
  (\bm F\bm u_i)(\bm F\bm u_i)^\trans
  \notag\\
  &+
  \bigl(
    \rho\bm F[\bm R_*,\bm r_i]
    +
    \rho_\perp\bm H\bm Z
  \bigr)
  \bigl(
    \rho\bm F[\bm R_*,\bm r_i]
    +
    \rho_\perp\bm H\bm Z
  \bigr)^\trans
  +
  \bm D_2\bm D_2^\trans.
\label{eq:lower-hidden-aligned-projector}
\end{align}

Write $\bm Z=[\bm Z_-,\bm z]$, where
$\bm Z_-\in\St(m,q-1)$.  Conditional on $\bm Z_-$, choose
$\bm O=\bm O(\bm Z_-)\in\St(m,M)$ by a fixed measurable rule so that
$\operatorname{col}(\bm O)=\operatorname{col}(\bm Z_-)^\perp$.  Then
\begin{equation}
  \bm z=\bm O\bm\zeta,
  \qquad
  \bm\zeta\sim\operatorname{Unif}(\mathbb S^{M-1}).
\label{eq:lower-hidden-zeta}
\end{equation}
Define
\begin{equation}
  \bm\Pi_{\rm fix}(\bm Z_-)
  \coloneqq
  \bm F\bm U_*\bm U_*^\trans\bm F^\trans
  +
  (\rho\bm F\bm R_*+\rho_\perp\bm H\bm Z_-)
  (\rho\bm F\bm R_*+\rho_\perp\bm H\bm Z_-)^\trans
  +
  \bm D_2\bm D_2^\trans.
\label{eq:lower-hidden-Pi-fix}
\end{equation}
Substitution in \cref{eq:lower-hidden-aligned-projector} gives
\begin{align}
  \widetilde{\bm\Pi}_{i,\bm Z}
  ={}&
  \bm\Pi_{\rm fix}(\bm Z_-)
  +
  (\bm F\bm u_i)(\bm F\bm u_i)^\trans
  \notag\\
  &+
  (\rho\bm F\bm r_i+\rho_\perp\bm H\bm O\bm\zeta)
  (\rho\bm F\bm r_i+\rho_\perp\bm H\bm O\bm\zeta)^\trans.
\label{eq:lower-hidden-projector-decomposition}
\end{align}

The last two terms in
\cref{eq:lower-hidden-projector-decomposition} lie in the column space of
the following isometry.  Define their coordinates by
\begin{equation}
\begin{aligned}
  \bm J
  &\coloneqq
  [\bm F\bm u_0,\bm F\bm r_0,\bm H\bm O]
  \in\St(n,M+2),\\
  \bm a_i&\coloneqq\bm J^\trans\bm F\bm u_i,
  &\bm b_i&\coloneqq\bm J^\trans\bm F\bm r_i,
  &\widetilde{\bm z}&\coloneqq
  \begin{bmatrix}0&0&\bm\zeta^\trans\end{bmatrix}^\trans.
\end{aligned}
\label{eq:lower-hidden-reduced-coordinates}
\end{equation}
By construction, $\bm\Pi_{\rm fix}(\bm Z_-)\bm J=\bm0$.  Moreover,
$\bm J\bm a_i=\bm F\bm u_i$, $\bm J\bm b_i=\bm F\bm r_i$, and
$\bm J\widetilde{\bm z}=\bm H\bm O\bm\zeta$.  Hence, with
\begin{equation}
  \bm\Gamma_i(\widetilde{\bm z})
  \coloneqq
  \bm a_i\bm a_i^\trans
  +
  (\rho\bm b_i+\rho_\perp\widetilde{\bm z})
  (\rho\bm b_i+\rho_\perp\widetilde{\bm z})^\trans,
\label{eq:lower-hidden-Gamma-i}
\end{equation}
the projector in \cref{eq:lower-hidden-projector-decomposition} satisfies
\begin{equation}
  \widetilde{\bm\Pi}_{i,\bm Z}
  =
  \bm\Pi_{\rm fix}(\bm Z_-)
  +
  \bm J\bm\Gamma_i(\widetilde{\bm z})\bm J^\trans.
\label{eq:lower-hidden-final-projector-decomposition}
\end{equation}

The relation $\bm\Pi_{\rm fix}(\bm Z_-)\bm J=\bm0$ makes the Gaussian
coordinates in $\operatorname{col}(\bm J)$ independent of those in its
orthogonal complement.  Define
$\bm W_j\coloneqq\bm J^\trans\bm y_{2,j}\in\mathbb R^{M+2}$.  Conditional
on $\bm\zeta$,
\begin{equation}
  \bm W_j
  \overset{\mathrm{iid}}{\sim}
  \mathcal N\left(
    \bm0,\,
    \sigma^2\{\bm I_{M+2}+\beta\bm\Gamma_i(\widetilde{\bm z})\}
  \right),
  \qquad
  j=1,\ldots,d_2.
\label{eq:lower-hidden-reduced-law}
\end{equation}
Let $\mathcal Q_i$ denote the law obtained by first drawing
$\bm\zeta\sim\operatorname{Unif}(\mathbb S^{M-1})$ and then drawing the
vectors in \cref{eq:lower-hidden-reduced-law}.  After revealing
$(\bm Z_-,\bm O)$, the remaining Gaussian coordinates are independent of
$\bm W_1,\ldots,\bm W_{d_2}$ and have the same law for $i=0,1$.  Therefore
\begin{equation}
  \operatorname{TV}(\mathsf Q_0,\mathsf Q_1)
  \le
  \operatorname{TV}(\mathcal Q_0,\mathcal Q_1).
\label{eq:lower-hidden-TV-reduction}
\end{equation}

\paragraph{Hellinger comparison and conclusion.}
The two orthonormal bases $(\bm a_0,\bm b_0)$ and
$(\bm a_1,\bm b_1)$ span the same two-dimensional subspace.  After choosing
their signs consistently, there is an angle $\varphi\in[0,\pi/2]$ such that
\begin{equation}
\begin{aligned}
  \bm a_1&=\cos(\varphi)\bm a_0+\sin(\varphi)\bm b_0,\\
  \bm b_1&=-\sin(\varphi)\bm a_0+\cos(\varphi)\bm b_0.
\end{aligned}
\label{eq:lower-hidden-endpoint-rotation}
\end{equation}
Because $\bm u_0\in\operatorname{col}(\bm U_{\bm T_0}^{\circ})$, there is
$\bm x\in\mathbb R^r$ with $\|\bm x\|_2\le1$ such that
$\bm u_0=[\bm I_r;\bm T_0]\bm x$.  Since
$[\bm I_r;\bm T_1]\bm x$ belongs to
$\operatorname{col}(\bm U_{\bm T_1}^{\circ})$,
\begin{equation}
  \sin\varphi
  =
  \operatorname{dist}\bigl(
    \bm u_0,\operatorname{col}(\bm U_{\bm T_1}^{\circ})
  \bigr)
  \le
  \|(\bm T_0-\bm T_1)\bm x\|_2
  \le
  2\delta.
\label{eq:lower-hidden-angle-sine}
\end{equation}
It follows that
\begin{equation}
  \varphi\le2\sin\varphi\le4\delta.
\label{eq:lower-hidden-angle-bound}
\end{equation}

For $t\in[0,\varphi]$, set
\begin{equation}
\begin{aligned}
  \bm a_t&\coloneqq\cos(t)\bm a_0+\sin(t)\bm b_0,\\
  \bm b_t&\coloneqq-\sin(t)\bm a_0+\cos(t)\bm b_0,\\
  \bm\Gamma_t(\widetilde{\bm z})
  &\coloneqq
  \bm a_t\bm a_t^\trans
  +
  (\rho\bm b_t+\rho_\perp\widetilde{\bm z})
  (\rho\bm b_t+\rho_\perp\widetilde{\bm z})^\trans.
\end{aligned}
\label{eq:lower-hidden-path-projector}
\end{equation}
Let $\mathcal Q_t$ be the law obtained by drawing
$\bm\zeta\sim\operatorname{Unif}(\mathbb S^{M-1})$, forming
$\widetilde{\bm z}$ as in \cref{eq:lower-hidden-reduced-coordinates}, and
then drawing
\begin{equation}
  \bm W_j
  \overset{\mathrm{iid}}{\sim}
  \mathcal N\left(
    \bm0,\,
    \sigma^2\{\bm I_{M+2}+\beta\bm\Gamma_t(\widetilde{\bm z})\}
  \right),
  \qquad
  j=1,\ldots,d_2.
\label{eq:lower-hidden-path-law}
\end{equation}
Then $\mathcal Q_{t=0}=\mathcal Q_0$ and
$\mathcal Q_{t=\varphi}=\mathcal Q_1$.

Let $q_t$ be the density of $\mathcal Q_t$, and define
\begin{equation}
  S_t(\bm W)\coloneqq\partial_t\log q_t(\bm W),
  \qquad
  J_t\coloneqq\mathbb E_{\mathcal Q_t}S_t(\bm W)^2.
\label{eq:lower-hidden-J-definition}
\end{equation}
By \cref{lem:hellinger-path},
\begin{equation}
  \operatorname{TV}(\mathcal Q_0,\mathcal Q_1)
  \le
  \frac12\int_0^\varphi\sqrt{J_t}\,dt
  \le
  \frac{\varphi}{2}
  \sup_{0\le t\le\varphi}\sqrt{J_t}.
\label{eq:lower-hidden-TV-by-J}
\end{equation}
It remains to bound $J_t$ uniformly over the path.  The coefficient that
appears in this bound satisfies
\begin{equation}
  I
  \coloneqq
  \frac{d_2\beta^2}{1+\beta}
  =
  \frac{s_2^4}{c_B\sigma^2(s_2^2+c_Bd_2\sigma^2)}
  \asymp
  \frac{\gamma_2^2}{\sigma^2}.
\label{eq:lower-hidden-I}
\end{equation}
Since $\beta\le1$, we have $I\le d_2$.  The signal-strength assumption and
$M\asymp n$ also give
\begin{equation}
  d_2\ge I\ge c_1\nu^2M.
\label{eq:lower-hidden-information-comparison}
\end{equation}

The law $\mathcal Q_t$ is a mixture over $\widetilde{\bm z}$, so the score
cannot be computed as if $\widetilde{\bm z}$ were observed.  The following
lemma controls the mixture score.  Its proof follows the three neighboring-law
proofs.

\begin{lemma}
\label{lem:lower-hidden-path-information}
Under \cref{eq:lower-hidden-information-comparison}, uniformly for
$t\in[0,\varphi]$,
\begin{equation}
  J_t
  \le
  C\left(
    I\rho_\perp^4
    +
    \frac{I^2\rho_\perp^4}{M}
    +
    I\rho_\perp^2e^{-c_2d_2}
  \right).
\label{eq:lower-hidden-J-bound}
\end{equation}
\end{lemma}

Combining
\cref{eq:lower-hidden-TV-by-J,eq:lower-hidden-J-bound} yields
\begin{equation}
  \operatorname{TV}^2(\mathcal Q_0,\mathcal Q_1)
  \le
  C\varphi^2
  \left(
    I\rho_\perp^4
    +
    \frac{I^2\rho_\perp^4}{M}
    +
    I\rho_\perp^2e^{-c_2d_2}
  \right).
\label{eq:lower-hidden-reduced-TV-bound}
\end{equation}
By the assumed bound on $\delta$,
\cref{eq:lower-hidden-angle-bound,eq:lower-hidden-I} imply
\begin{equation}
  \varphi
  \le
  Cc\frac{\sqrt M}{I\rho_\perp^2}.
\label{eq:lower-hidden-angle-information}
\end{equation}
Consequently,
\begin{align*}
  \varphi^2\frac{I^2\rho_\perp^4}{M}
  &\le Cc^2,
  &
  \varphi^2I\rho_\perp^4
  &\le Cc^2\frac{M}{I}
  \le Cc^2\nu^{-2}.
\end{align*}
Also, $\varphi\le2$, $I\le d_2$, and $\rho_\perp\le1$ give
$\varphi^2I\rho_\perp^2e^{-c_2d_2}\le4d_2e^{-c_2d_2}$.  Since
$d_2\ge c_1\nu^2M$, the last quantity is sufficiently small for sufficiently
large $\nu$.  Choosing the numerical constant $c$ in the edge assumption
sufficiently small now gives
$\operatorname{TV}(\mathcal Q_0,\mathcal Q_1)\le1/8$.  Hence
\cref{eq:lower-hidden-TV-reduction} and the conditioning bound at the start
of the proof yield
\[
  \operatorname{TV}(\mathsf P_{\bm T_0},\mathsf P_{\bm T_1})
  \le
  \frac1{16}+\frac18+\frac1{16}
  =
  \frac14.
\]
This proves \cref{lem:lower-hidden-low-edge}.

\subsubsection{Proof of \cref{lem:lower-hidden-high-edge}}

Let $\mathsf P_i^{(2)}$ denote the View-2 law under $\bm T_i$, for
$i\in\{0,1\}$.  The loading is deterministic in this construction, so there
is no conditioning error.  By the common reduction above,
\[
  \operatorname{TV}\bigl(
    \mathsf P_{\bm T_0},
    \mathsf P_{\bm T_1}
  \bigr)
  =
  \operatorname{TV}\bigl(
    \mathsf P_0^{(2)},
    \mathsf P_1^{(2)}
  \bigr).
\]
We reduce the View-2 laws to an $(M+2)\times2$ Gaussian mean model and then
compare the reduced laws along a Hellinger path.

\paragraph{The second-view mean mixtures.}
Write
\[
  \bm M_{\bm T,\bm Z}
  \coloneqq
  \left[
    \bm U_{\bm T},
    \rho\bm F\bm R_{\bm T}^{\circ}
      +\rho_\perp\bm H\bm Z
  \right]
  [\bm U_{\bm T}^{\circ},\bm R_{\bm T}^{\circ}]^\trans.
\]
Conditional on $\bm Z\in\operatorname{St}(m,q)$, View~2 has entrywise noise
variance $\sigma^2$ and the mean appearing below.  Consequently,
\begin{equation}
  \mathsf P_i^{(2)}
  =
  \int
  \mathcal N\left(
    \operatorname{vec}\left\{
      s_2\bm M_{\bm T_i,\bm Z}\bm W_{2,0}^\trans
      +
      s_2\bm D_2\bm W_{2,1}^\trans
    \right\},
    \sigma^2\bm I_{nd_2}
  \right)
  d\mu_{\mathrm{Haar}}(\bm Z).
\label{eq:lower-high-mixture-law}
\end{equation}
It remains to express the two mixtures in coordinates in which all terms
that depend on $i$ appear in one Gaussian block.

\paragraph{Reduction to the Gaussian mean model.}
Use the matrices and vectors in
\cref{eq:lower-hidden-common-bases,eq:lower-hidden-change-of-basis}.  The
identities in those displays give
\begin{equation}
\begin{aligned}
  \bm F\bm U_{\bm T_i}^{\circ}
  (\bm U_{\bm T_i}^{\circ})^\trans
  &=
  \bm F[\bm U_*,\bm u_i][\bm U_*,\bm u_i]^\trans,\\
  &(\rho\bm F\bm R_{\bm T_i}^{\circ}+\rho_\perp\bm H\bm Z)
  (\bm R_{\bm T_i}^{\circ})^\trans\\
  &\hspace{1cm}=
  (\rho\bm F[\bm R_*,\bm r_i]
    +\rho_\perp\bm H\bm Z\bm Q_i^\trans)
  [\bm R_*,\bm r_i]^\trans.
\end{aligned}
\label{eq:lower-high-basis-change}
\end{equation}
Since $\bm Z\bm Q_i^\trans\overset{\mathrm d}{=}\bm Z$, the mixture in
\cref{eq:lower-high-mixture-law} is unchanged if
$\bm M_{\bm T_i,\bm Z}$ is replaced by
\begin{equation}
  \widetilde{\bm M}_{i,\bm Z}
  \coloneqq
  \bm F[\bm U_*,\bm u_i][\bm U_*,\bm u_i]^\trans
  +
  (\rho\bm F[\bm R_*,\bm r_i]+\rho_\perp\bm H\bm Z)
  [\bm R_*,\bm r_i]^\trans.
\label{eq:lower-high-aligned-mean}
\end{equation}

Write $\bm Z=[\bm Z_-,\bm z]$ and condition on $\bm Z_-$.  Choose
$\bm O\in\St(m,M)$ with
$\operatorname{col}(\bm O)=\operatorname{col}(\bm Z_-)^\perp$ as in
\cref{eq:lower-hidden-zeta}, so that
$\bm z=\bm O\bm\zeta$ for
$\bm\zeta\sim\operatorname{Unif}(\mathbb S^{M-1})$.  Expanding
\cref{eq:lower-high-aligned-mean} gives
\begin{align}
  \widetilde{\bm M}_{i,\bm Z}
  ={}&
  \bm F\bm U_*\bm U_*^\trans
  +
  (\rho\bm F\bm R_*+\rho_\perp\bm H\bm Z_-)\bm R_*^\trans
  \notag\\
  &+
  (\bm F\bm u_i)\bm u_i^\trans
  +
  (\rho\bm F\bm r_i+\rho_\perp\bm H\bm O\bm\zeta)\bm r_i^\trans.
\label{eq:lower-high-mean-decomposition}
\end{align}
The first line is independent of $i$.  We now isolate the last two terms by
defining
\begin{equation}
\begin{aligned}
  \bm J_{\rm L}
  &\coloneqq
  [\bm F\bm u_0,\bm F\bm r_0,\bm H\bm O]
  \in\St(n,M+2),
  &
  \bm J_{\rm R}
  &\coloneqq
  [\bm u_0,\bm r_0]
  \in\St(r+q,2),\\
  \bm a_i&\coloneqq\bm J_{\rm L}^\trans\bm F\bm u_i,
  &
  \bm b_i&\coloneqq\bm J_{\rm L}^\trans\bm F\bm r_i,\\
  \overline{\bm a}_i&\coloneqq\bm J_{\rm R}^\trans\bm u_i,
  &
  \overline{\bm b}_i&\coloneqq\bm J_{\rm R}^\trans\bm r_i,
  &
  \widetilde{\bm z}&\coloneqq
  \begin{bmatrix}0&0&\bm\zeta^\trans\end{bmatrix}^\trans.
\end{aligned}
\label{eq:lower-high-reduced-coordinates}
\end{equation}
With these coordinates, the second line of
\cref{eq:lower-high-mean-decomposition} equals
$\bm J_{\rm L}\bm{\mathcal M}_i(\widetilde{\bm z})\bm J_{\rm R}^\trans$,
where
\begin{equation}
  \bm{\mathcal M}_i(\widetilde{\bm z})
  \coloneqq
  \bm a_i\overline{\bm a}_i^\trans
  +
  (\rho\bm b_i+\rho_\perp\widetilde{\bm z})
  \overline{\bm b}_i^\trans.
\label{eq:lower-high-reduced-mean-i}
\end{equation}

Conditional on $(\bm Z_-,\bm O)$, apply orthogonal changes of coordinates
whose first left and right blocks are $\bm J_{\rm L}$ and
$\bm W_{2,0}\bm J_{\rm R}$, respectively.  The corresponding block of the
mean is exactly $s_2\bm{\mathcal M}_i(\widetilde{\bm z})$; all other mean
blocks are independent of $i$, and the Gaussian noise blocks are independent.
Let $\mathcal Q_i$ denote the law of
\begin{equation}
  \bm W
  =
  s_2\bm{\mathcal M}_i(\widetilde{\bm z})+\sigma\bm G,
  \qquad
  \bm G\in\mathbb R^{(M+2)\times2},
\label{eq:lower-high-reduced-observation-i}
\end{equation}
where $\bm G$ has independent standard Gaussian entries and
$\bm\zeta$ is uniform on $\mathbb S^{M-1}$.  Revealing
$(\bm Z_-,\bm O)$ can only increase total variation, and the resulting
conditional distance does not depend on the revealed variables.  Therefore
\begin{equation}
  \operatorname{TV}(\mathsf P_0^{(2)},\mathsf P_1^{(2)})
  \le
  \operatorname{TV}(\mathcal Q_0,\mathcal Q_1).
\label{eq:lower-high-tv-reduction}
\end{equation}

\paragraph{Hellinger comparison and conclusion.}
The rotation in \cref{eq:lower-hidden-endpoint-rotation} acts on both the
left and right coordinates in \cref{eq:lower-high-reduced-mean-i}.  For
$t\in[0,\varphi]$, define
\begin{equation}
\begin{aligned}
  \bm a_t&\coloneqq\cos(t)\bm a_0+\sin(t)\bm b_0,
  &
  \bm b_t&\coloneqq-\sin(t)\bm a_0+\cos(t)\bm b_0,\\
  \overline{\bm a}_t
  &\coloneqq
  \cos(t)\overline{\bm a}_0+\sin(t)\overline{\bm b}_0,
  &
  \overline{\bm b}_t
  &\coloneqq
  -\sin(t)\overline{\bm a}_0+\cos(t)\overline{\bm b}_0.
\end{aligned}
\label{eq:lower-high-path-coordinates}
\end{equation}
Set
\begin{equation}
  \bm{\mathcal M}_t(\widetilde{\bm z})
  \coloneqq
  \bm a_t\overline{\bm a}_t^\trans
  +
  (\rho\bm b_t+\rho_\perp\widetilde{\bm z})
  \overline{\bm b}_t^\trans,
\label{eq:lower-high-reduced-mean}
\end{equation}
and let $\mathcal Q_t$ be the law obtained by drawing $\bm\zeta$ uniformly
on $\mathbb S^{M-1}$ and observing
\begin{equation}
  \bm W
  =
  s_2\bm{\mathcal M}_t(\widetilde{\bm z})+\sigma\bm G.
\label{eq:lower-high-path-law}
\end{equation}
Then $\mathcal Q_{t=0}=\mathcal Q_0$,
$\mathcal Q_{t=\varphi}=\mathcal Q_1$, and
\cref{eq:lower-hidden-angle-bound} gives $\varphi\le4\delta$.

Let $q_t$ be the density of $\mathcal Q_t$, and define
\begin{equation}
  S_t(\bm W)\coloneqq\partial_t\log q_t(\bm W),
  \qquad
  J_t\coloneqq\mathbb E_{\mathcal Q_t}S_t(\bm W)^2.
\label{eq:lower-high-J-definition}
\end{equation}
By \cref{lem:hellinger-path},
\begin{equation}
  \operatorname{TV}(\mathcal Q_0,\mathcal Q_1)
  \le
  \frac12\int_0^\varphi\sqrt{J_t}\,dt.
\label{eq:lower-high-TV-by-J}
\end{equation}
The following lemma supplies the uniform bound needed in this inequality;
its proof is given below with the proof of the covariance-mixture information
bound.

\begin{lemma}
\label{lem:high-hidden-J-bound}
Uniformly for $t\in[0,\varphi]$,
\begin{equation}
  J_t
  \le
  8\frac{s_2^2\theta^2}{\sigma^2}
  +
  144\frac{s_2^4\theta^2}{M\sigma^4}.
\label{eq:high-hidden-J-bound}
\end{equation}
\end{lemma}

Combining \cref{eq:lower-high-TV-by-J,eq:high-hidden-J-bound} gives
\begin{equation}
  \operatorname{TV}^2(\mathcal Q_0,\mathcal Q_1)
  \le
  C\varphi^2
  \left(
    \frac{s_2^2\theta^2}{\sigma^2}
    +
    \frac{s_2^4\theta^2}{M\sigma^4}
  \right).
\label{eq:lower-high-Hellinger}
\end{equation}
The definition of $\gamma_2$ and the high-loading condition imply
\begin{equation}
  \frac{c_B}{1+c_B}s_2^2
  \le
  \gamma_2^2
  \le
  s_2^2.
\label{eq:lower-high-gamma-s-comparison}
\end{equation}
The signal-strength assumption and $M\asymp n$ then give
$M\sigma^2/s_2^2\le C\nu^{-2}$.  The assumed bound on $\delta$,
\cref{eq:lower-hidden-angle-bound,eq:lower-high-gamma-s-comparison}, and
$\rho_\perp^2\ge2\theta$ yield
\begin{equation}
  \varphi
  \le
  Cc\frac{\sigma^2\sqrt M}{s_2^2\theta}.
\label{eq:lower-high-angle-information}
\end{equation}
Substitution in \cref{eq:lower-high-Hellinger} gives
\begin{equation}
  \operatorname{TV}^2(\mathcal Q_0,\mathcal Q_1)
  \le
  Cc^2\left(1+\frac{M\sigma^2}{s_2^2}\right)
  \le
  Cc^2(1+\nu^{-2}).
\label{eq:lower-high-final-TV}
\end{equation}
Choosing the numerical constant $c$ sufficiently small gives
$\operatorname{TV}(\mathcal Q_0,\mathcal Q_1)\le1/4$.  Combining
\cref{eq:lower-high-tv-reduction} with the equality between the joint and
View-2 distances at the start of the proof proves
\[
  \operatorname{TV}(\mathsf P_{\bm T_0},\mathsf P_{\bm T_1})
  \le
  \frac14.
\]
This proves \cref{lem:lower-hidden-high-edge}.

\subsection{Information bounds for the hidden constructions}

\subsubsection{Proof of~\cref{lem:lower-hidden-path-information}}

Let $q_{t,\widetilde{\bm z}}$ denote the conditional density of the
observations in \cref{eq:lower-hidden-path-law} for fixed
$\widetilde{\bm z}$, and define its score by
\begin{equation}
  S_{t,\widetilde{\bm z}}(\bm W)
  \coloneqq
  \partial_t\log q_{t,\widetilde{\bm z}}(\bm W).
\label{eq:lower-hidden-conditional-score-definition}
\end{equation}
Because the conditional law is a product of centered Gaussian laws,
differentiating its covariance gives
\begin{equation}
  S_{t,\widetilde{\bm z}}(\bm W)
  =
  \frac{\beta\rho_\perp}{\sigma^2(1+\beta)}
  \sum_{j=1}^{d_2}
  (\bm a_t^\trans\bm W_j)
  \left\langle
    \rho_\perp\bm b_t-\rho\widetilde{\bm z},\bm W_j
  \right\rangle.
\label{eq:lower-hidden-component-score}
\end{equation}
Define
$\widehat{\bm z}_t\coloneqq
\mathbb E[\widetilde{\bm z}\mid\bm W]$.  Fisher's identity and
\cref{eq:lower-hidden-component-score} yield
\begin{equation}
  S_t(\bm W)
  =
  \frac{\beta\rho_\perp}{\sigma^2(1+\beta)}
  \sum_{j=1}^{d_2}
  (\bm a_t^\trans\bm W_j)
  \left\langle
    \rho_\perp\bm b_t-\rho\widehat{\bm z}_t,\bm W_j
  \right\rangle.
\label{eq:lower-hidden-mixture-score}
\end{equation}

To compute its second moment, set
\[
  X_j\coloneqq\bm a_t^\trans\bm W_j,
  \qquad
  \bm R_j\coloneqq\bm W_j-X_j\bm a_t.
\]
Conditional on $\widetilde{\bm z}$, the variables $X_j$ are independent
$\mathcal N(0,\sigma^2(1+\beta))$ variables whose joint law does not depend
on $\widetilde{\bm z}$, and they are independent of
$(\bm R_1,\ldots,\bm R_{d_2})$.  These independence statements remain true
after integrating over $\widetilde{\bm z}$, and
\begin{equation}
  \widehat{\bm z}_t
  =
  \mathbb E\!\left[
    \widetilde{\bm z}
    \,\middle|\,
    \bm R_1,\ldots,\bm R_{d_2}
  \right].
\label{eq:lower-hidden-posterior-residuals}
\end{equation}
Moreover, $\rho_\perp\bm b_t-\rho\widehat{\bm z}_t$ is orthogonal to
$\bm a_t$.  Thus, if
\begin{equation}
  c_j
  \coloneqq
  \left\langle
    \rho_\perp\bm b_t-\rho\widehat{\bm z}_t,\bm W_j
  \right\rangle
  =
  \left\langle
    \rho_\perp\bm b_t-\rho\widehat{\bm z}_t,\bm R_j
  \right\rangle,
\label{eq:lower-hidden-cj}
\end{equation}
then $c_1,\ldots,c_{d_2}$ are fixed conditional on
$(\bm R_1,\ldots,\bm R_{d_2})$.  The independence and centering of the
$X_j$ therefore give
\begin{align}
  &\mathbb E\!\left[
    S_t(\bm W)^2
    \,\middle|\,
    \bm R_1,\ldots,\bm R_{d_2}
  \right]
  \notag\\
  &\quad=
  \frac{\beta^2\rho_\perp^2}{\sigma^4(1+\beta)^2}
  \mathbb E\!\left[
    \left(\sum_{j=1}^{d_2}X_jc_j\right)^2
    \,\middle|\,
    \bm R_1,\ldots,\bm R_{d_2}
  \right]
  =
  \frac{\beta^2\rho_\perp^2}{\sigma^2(1+\beta)}
  \sum_{j=1}^{d_2}c_j^2.
\label{eq:lower-hidden-conditional-score-second-moment}
\end{align}
After taking expectation and using $(x+y)^2\le2x^2+2y^2$, we obtain
\begin{align}
  J_t
  \le{}&
  \frac{2\beta^2\rho_\perp^4}{\sigma^2(1+\beta)}
  \mathbb E\sum_{j=1}^{d_2}
  \langle\bm b_t,\bm W_j\rangle^2
  +
  \frac{2\beta^2\rho_\perp^2}{\sigma^2(1+\beta)}
  \mathbb E\sum_{j=1}^{d_2}
  \langle\widehat{\bm z}_t,\bm W_j\rangle^2.
\label{eq:lower-hidden-J-split}
\end{align}
Since $\bm b_t^\trans\bm W_j$ has variance
$\sigma^2(1+\beta\rho^2)$, the first term in
\cref{eq:lower-hidden-J-split} is at most $CI\rho_\perp^4$.

We next bound the second term.  The identity
\begin{equation}
  \sum_{j=1}^{d_2}
  \langle\widehat{\bm z}_t,\bm W_j\rangle^2
  =
  \widehat{\bm z}_t^\trans
  \left(\sum_{j=1}^{d_2}\bm W_j\bm W_j^\trans\right)
  \widehat{\bm z}_t
\label{eq:lower-hidden-posterior-quadratic-form}
\end{equation}
requires an operator-norm bound for the sample covariance.  Define
\begin{equation}
  \mathcal E_t
  \coloneqq
  \left\{
    \left\|\sum_{j=1}^{d_2}\bm W_j\bm W_j^\trans\right\|_{\op}
    \le
    Cd_2\sigma^2
  \right\}.
\label{eq:lower-hidden-scatter-event}
\end{equation}
We prove below that, uniformly in $t$,
\begin{equation}
  \mathbb P(\mathcal E_t^c)\le Ce^{-c_2d_2},
  \qquad
  \mathbb E\left[
    \left\|\sum_{j=1}^{d_2}\bm W_j\bm W_j^\trans\right\|_{\op}
    \bm1_{\mathcal E_t^c}
  \right]
  \le
  Cd_2\sigma^2e^{-c_2d_2}.
\label{eq:W-event}
\end{equation}

It remains to control $\mathbb E\|\widehat{\bm z}_t\|_2^2$.  For fixed
$\widetilde{\bm z},\widetilde{\bm z}'$, the corresponding conditional
Gaussian laws satisfy
\begin{equation}
  \operatorname{KL}\bigl(
    \mathcal Q_{t,\widetilde{\bm z}}
    \,\|\,
    \mathcal Q_{t,\widetilde{\bm z}'}
  \bigr)
  \le
  \frac{I\rho_\perp^2}{2}
  \|\widetilde{\bm z}-\widetilde{\bm z}'\|_2^2.
\label{eq:lower-hidden-conditional-KL}
\end{equation}
Indeed, the covariance projectors share $\bm a_t\bm a_t^\trans$ and differ
only in the projectors onto
$\rho\bm b_t+\rho_\perp\widetilde{\bm z}$ and
$\rho\bm b_t+\rho_\perp\widetilde{\bm z}'$, so
\cref{eq:gaussian-projector-kl} gives this bound.  Averaging over an
independent copy $\widetilde{\bm z}'$ and using convexity of relative entropy
in its second argument yield
$\operatorname{I}(\widetilde{\bm z};\bm W)\le CI\rho_\perp^2$.
Therefore \cref{lem:spherical-posterior} gives
\begin{equation}
  \mathbb E\|\widehat{\bm z}_t\|_2^2
  \le
  C\frac{I\rho_\perp^2}{M}.
\label{eq:lower-hidden-posterior-bound}
\end{equation}

On $\mathcal E_t$, combine
\cref{eq:lower-hidden-posterior-quadratic-form,eq:lower-hidden-posterior-bound};
on $\mathcal E_t^c$, use $\|\widehat{\bm z}_t\|_2\le1$ and
\cref{eq:W-event}.  This gives
\begin{equation}
  \mathbb E\sum_{j=1}^{d_2}
  \langle\widehat{\bm z}_t,\bm W_j\rangle^2
  \le
  Cd_2\sigma^2\frac{I\rho_\perp^2}{M}
  +
  Cd_2\sigma^2e^{-c_2d_2}.
\label{eq:lower-hidden-posterior-energy}
\end{equation}
Substituting this estimate into \cref{eq:lower-hidden-J-split} and using
$I=d_2\beta^2/(1+\beta)$ gives the bound in
\cref{eq:lower-hidden-J-bound}.  This proves
\cref{lem:lower-hidden-path-information}.

\paragraph{Proof of \cref{eq:W-event}.}
Conditional on $\widetilde{\bm z}$, the vectors
$\bm W_1,\ldots,\bm W_{d_2}$ are independent centered Gaussian vectors with
covariance
\begin{equation}
  \bm\Sigma_{t,\widetilde{\bm z}}
  \coloneqq
  \sigma^2\{\bm I_{M+2}+\beta\bm\Gamma_t(\widetilde{\bm z})\},
  \qquad
  \|\bm\Sigma_{t,\widetilde{\bm z}}\|_{\mathrm{op}}
  \le2\sigma^2.
\label{eq:lower-hidden-conditional-covariance}
\end{equation}
The norm bound follows because $\bm\Gamma_t(\widetilde{\bm z})$ is an
orthogonal projector and $\beta\le1$.  If
$\bm G\in\mathbb R^{(M+2)\times d_2}$ has independent
$\mathcal N(0,d_2^{-1})$ entries, then, conditionally on
$\widetilde{\bm z}$,
\begin{equation}
  [\bm W_1,\ldots,\bm W_{d_2}]
  \overset{\mathrm d}{=}
  \sqrt{d_2}\,\bm\Sigma_{t,\widetilde{\bm z}}^{1/2}\bm G,
  \qquad
  \left\|\sum_{j=1}^{d_2}\bm W_j\bm W_j^\trans\right\|_{\op}
  \le
  2d_2\sigma^2\opn{\bm G\bm G^\trans}.
\label{eq:lower-hidden-scatter-conditional}
\end{equation}

By \cref{eq:lower-hidden-information-comparison},
$d_2\ge I\ge c_1\nu^2M$.  Applying
\cref{eq:gaussian-wishart-tail} with $u=M+2$ and $b=d_2$ therefore yields,
uniformly in $\widetilde{\bm z}$,
\begin{equation}
  \mathbb P\left(
    \left.
    \left\|\sum_{j=1}^{d_2}\bm W_j\bm W_j^\trans\right\|_{\op}
    >Cd_2\sigma^2
    \,\right|\,
    \widetilde{\bm z}
  \right)
  \le
  Ce^{-c_2d_2}.
\label{eq:lower-hidden-scatter-tail}
\end{equation}
Averaging over $\widetilde{\bm z}$ proves the first inequality in
\cref{eq:W-event}.  For the second, combine
\cref{eq:lower-hidden-scatter-conditional}, Cauchy--Schwarz, and
\cref{eq:gaussian-wishart-moment} to obtain, uniformly in
$\widetilde{\bm z}$,
\begin{align*}
  &\mathbb E\left[
    \left.
    \left\|\sum_{j=1}^{d_2}\bm W_j\bm W_j^\trans\right\|_{\op}
    \bm1_{\mathcal E_t^c}
    \,\right|\,
    \widetilde{\bm z}
  \right]\le
  Cd_2\sigma^2
  \|\bm G\bm G^\trans\|_{L^2}
  \mathbb P(\mathcal E_t^c\mid\widetilde{\bm z})^{1/2}
  \le
  Cd_2\sigma^2e^{-c_2d_2}.
\end{align*}
Averaging over $\widetilde{\bm z}$ proves the second inequality in
\cref{eq:W-event}.

\subsubsection{Proof of~\cref{lem:high-hidden-J-bound}}

For fixed $\widetilde{\bm z}$, the law in
\cref{eq:lower-high-path-law} is a Gaussian location model with covariance
$\sigma^2\bm I$.  Since $\dot{\bm a}_t=\bm b_t$ and
$\dot{\bm b}_t=-\bm a_t$, with the same identities for the barred vectors,
differentiating \cref{eq:lower-high-reduced-mean} gives
\begin{equation}
  \dot{\bm{\mathcal M}}_t(\widetilde{\bm z})
  =
  \{(1-\rho)\bm b_t-\rho_\perp\widetilde{\bm z}\}
  \overline{\bm a}_t^\trans
  +
  (1-\rho)\bm a_t\overline{\bm b}_t^\trans.
\label{eq:lower-high-mean-derivative}
\end{equation}
The terms independent of $\rho$ cancel because the left and right
coordinates undergo the same rotation.  The remaining terms are
proportional to $1-\rho=2\theta$ or $\rho_\perp$.

Let $S_{t,\widetilde{\bm z}}$ denote the score conditional on
$\widetilde{\bm z}$.  The Gaussian location score is
\[
  S_{t,\widetilde{\bm z}}(\bm W)
  =
  \frac{s_2}{\sigma^2}
  \left\langle
    \bm W-s_2\bm{\mathcal M}_t(\widetilde{\bm z}),
    \dot{\bm{\mathcal M}}_t(\widetilde{\bm z})
  \right\rangle_{\mathrm F}.
\]
Define
\begin{equation}
  \bm\varepsilon_t
  \coloneqq
  \bm W\overline{\bm a}_t-s_2\bm a_t,
  \qquad
  \xi_t
  \coloneqq
  \langle\bm W\overline{\bm b}_t,\bm a_t\rangle.
\label{eq:lower-high-epsilon-xi}
\end{equation}
Using \cref{eq:lower-high-mean-derivative} and the orthogonality of
$\bm a_t$ to $\bm b_t$ and $\widetilde{\bm z}$ gives
\begin{equation}
  S_{t,\widetilde{\bm z}}(\bm W)
  =
  \frac{s_2}{\sigma^2}
  \left[
    \left\langle
      \bm\varepsilon_t,
      (1-\rho)\bm b_t-\rho_\perp\widetilde{\bm z}
    \right\rangle
    +(1-\rho)\xi_t
  \right].
\label{eq:lower-high-component-score}
\end{equation}

Set
$\widehat{\bm z}_t\coloneqq
\mathbb E[\widetilde{\bm z}\mid\bm W]$.  Fisher's identity averages
\cref{eq:lower-high-component-score} over the conditional distribution of
$\widetilde{\bm z}$, so
\begin{equation}
  S_t(\bm W)
  =
  \frac{s_2}{\sigma^2}
  \left[
    \left\langle
      \bm\varepsilon_t,
      (1-\rho)\bm b_t-\rho_\perp\widehat{\bm z}_t
    \right\rangle
    +(1-\rho)\xi_t
  \right].
\label{eq:lower-high-mixture-score}
\end{equation}

Rotating the observed columns by
$(\overline{\bm a}_t,\overline{\bm b}_t)$ gives, conditionally on
$\widetilde{\bm z}$,
\begin{equation}
  \bm W\overline{\bm a}_t
  =
  s_2\bm a_t+\sigma\bm g_a,
  \qquad
  \bm W\overline{\bm b}_t
  =
  s_2(\rho\bm b_t+\rho_\perp\widetilde{\bm z})+\sigma\bm g_b,
\label{eq:lower-high-rotated-observations}
\end{equation}
where $\bm g_a$ and $\bm g_b$ are independent standard Gaussian vectors.
Only the last $M$ coordinates of the second vector in
\cref{eq:lower-high-rotated-observations} depend on
$\widetilde{\bm z}$; they have the Gaussian mean-channel form with signal
strength $s_2\rho_\perp$ and noise level $\sigma$.  Thus
\cref{lem:spherical-posterior} gives
\begin{equation}
  \mathbb E\|\widehat{\bm z}_t\|_2^2
  \le
  9\frac{s_2^2\rho_\perp^2}{M\sigma^2}.
\label{eq:lower-high-posterior}
\end{equation}

Conditional on these last $M$ coordinates,
$\widehat{\bm z}_t$ is fixed,
$\bm\varepsilon_t=\sigma\bm g_a$, and
$\xi_t=\sigma\langle\bm g_b,\bm a_t\rangle$.  The latter two Gaussian
quantities are centered and independent of one another and of the
conditioning variables.  Since $\widehat{\bm z}_t$ is orthogonal to
$\bm b_t$, taking the conditional second moment in
\cref{eq:lower-high-mixture-score}, then averaging and applying
\cref{eq:lower-high-posterior}, gives
\begin{align*}
  J_t
  &=
  \frac{s_2^2}{\sigma^2}
  \left\{
    2(1-\rho)^2
    +
    \rho_\perp^2\mathbb E\|\widehat{\bm z}_t\|_2^2
  \right\} \le
  8\frac{s_2^2\theta^2}{\sigma^2}
  +
  144\frac{s_2^4\theta^2}{M\sigma^4},
\end{align*}
where we used $1-\rho=2\theta$ and $\rho_\perp^2\le4\theta$.
This is \cref{eq:high-hidden-J-bound} and completes the proof.

\end{document}